\documentclass[11pt]{article}
\pdfoutput=1
\usepackage{amssymb}
\usepackage{graphicx}
\usepackage{url}
\usepackage{setspace}
\usepackage[pdftex,bookmarksnumbered,bookmarksopen,
colorlinks,citecolor=blue,linkcolor=blue]{hyperref}
\usepackage{xcolor}

\usepackage{times}
\usepackage{enumitem}
\usepackage{varwidth}
\usepackage{graphicx}
\usepackage{wrapfig}
\usepackage{enumerate}
\usepackage[noend]{algpseudocode}

\usepackage[utf8]{inputenc} % allow utf-8 input
\usepackage[T1]{fontenc}    % use 8-bit T1 fonts
\usepackage{hyperref}       % hyperlinks
\usepackage{url}            % simple URL typesetting
\usepackage{booktabs}       % professional-quality tables
\usepackage{amsfonts}       % blackboard math symbols
\usepackage{nicefrac}       % compact symbols for 1/2, etc.
\usepackage{microtype}      % microtypography

\usepackage[tight]{subfigure}
\usepackage{graphicx}
\usepackage{appendix}
\usepackage{amsmath,amsfonts,amsthm}
\usepackage{algorithm,algorithmicx,algpseudocode}
\usepackage{mdwlist}
\usepackage{xspace}
\usepackage{color}
\usepackage{mathrsfs}

\usepackage{booktabs}
\usepackage{comment}
\usepackage{geometry}

\usepackage{multirow}

\newcommand{\Appendix}[1]{the full version for}

\newtheorem{problem}{Problem}[section]
\newtheorem{theorem}{Theorem}[section]
\newtheorem{lemma}[theorem]{Lemma}

\newtheorem{remark}{Remark}
\newtheorem{claim}{Claim}

\newtheorem{definition}{Definition}

\newcommand{\e}{\mathbf{e}}

\newcommand{\s}{\mathbf{s}}
\renewcommand{\u}{\mathbf{u}}
\renewcommand{\v}{\mathbf{v}}
\newcommand{\w}{\mathbf{w}}
\newcommand{\x}{\mathbf{x}}
\newcommand{\y}{\mathbf{y}}

\newcommand{\A}{\mathbf{A}}
\newcommand{\B}{\mathbf{B}}
\newcommand{\C}{\mathbf{C}}
\newcommand{\D}{\mathbf{D}}
\newcommand{\E}{\mathbf{E}}
\newcommand{\F}{\mathbb{F}}
\newcommand{\G}{\mathbf{G}}
\renewcommand{\H}{\mathbf{H}}
\newcommand{\I}{\mathbf{I}}

\renewcommand{\L}{\mathbf{L}}
\newcommand{\M}{\mathbf{M}}

\newcommand{\R}{\mathbb{R}}
\renewcommand{\S}{\mathbf{S}}
\newcommand{\T}{\mathbf{T}}
\newcommand{\U}{\mathbf{U}}
\newcommand{\V}{\mathbf{V}}
\newcommand{\W}{\mathbf{W}}
\newcommand{\X}{\mathbf{X}}

\newcommand{\rank}{\mathsf{rank}}

\newcommand{\0}{\mathbf{0}}
\newcommand{\1}{\mathbf{1}}
\renewcommand{\comment}[1]{}
\newcommand{\red}[1]{{\color{red}#1}}

\newcommand{\cA}{\mathcal{A}}

\newcommand{\cC}{\mathcal{C}}
\newcommand{\cD}{\mathcal{D}}
\newcommand{\cE}{\mathcal{E}}

\newcommand{\cG}{\mathcal{G}}

\newcommand{\cI}{\mathcal{I}}

\newcommand{\cL}{\mathcal{L}}
\newcommand{\cR}{\mathcal{R}}
\newcommand{\cS}{\mathcal{S}}
\newcommand{\cT}{\mathcal{T}}
\newcommand{\cU}{\mathcal{U}}

\newcommand{\cM}{\mathcal{M}}
\newcommand{\cN}{\mathcal{N}}
\newcommand{\cO}{\mathcal{O}}

\newcommand{\cQ}{\mathcal{Q}}

\newcommand{\cX}{\mathcal{X}}

\newcommand{\bbF}{\mathbb{F}}
\newcommand{\bO}{\mathbf{O}}

\newcommand{\BHH}{\mathsf{BHH}}
\newcommand{\RS}{\mathsf{RS}}
\newcommand{\YES}{\mathsf{YES}}
\newcommand{\NO}{\mathsf{NO}}
\newcommand{\srank}{\mathsf{srank}}

\DeclareMathOperator{\erfc}{erfc}
\DeclareMathOperator*{\bbE}{\mathbb{E}}
\DeclareMathOperator{\tr}{tr}
\DeclareMathOperator{\poly}{\mathsf{poly}}

\algnewcommand\algorithmicinput{\textbf{Input:}}
\algnewcommand\algorithmicoutput{\textbf{Output:}}
\algnewcommand\INPUT{\item[\algorithmicinput]}
\algnewcommand\OUTPUT{\item[\algorithmicoutput]}
\algnewcommand{\LineComment}[1]{\Statex \(\triangleright\) #1}

\title{Testing Matrix Rank, Optimally}
\author{Maria-Florina Balcan\thanks{Carnegie Mellon University. Email: ninamf@cs.cmu.edu.} \and Yi Li\thanks{Nanyang Technological University. Email: yili@ntu.edu.sg.} \and David P. Woodruff\thanks{Carnegie Mellon University. Email: dwoodruf@cs.cmu.edu.} \and
Hongyang Zhang\thanks{Carnegie Mellon University. Email: hongyanz@cs.cmu.edu. % Corresponding author.  % not a journal submission, no need of this for arxiv
}
}

\date{}

\allowdisplaybreaks

\begin{document}

\maketitle

\begin{abstract}
We show that for the problem of testing if a matrix $\A \in \mathbb{F}^{n \times n}$ has rank at most $d$, or requires changing an $\epsilon$-fraction of entries to have rank at most $d$, there is a {\it non-adaptive} query algorithm making $\widetilde{\cO}(d^2/\epsilon)$ queries. Our algorithm works for any field $\mathbb{F}$. This improves upon the previous $\cO(d^2/\epsilon^2)$ bound (Krauthgamer and Sasson, SODA '03), and bypasses an $\Omega(d^2/\epsilon^2)$ lower bound of (Li, Wang, and Woodruff, KDD '14) which holds if the algorithm is required to 
read a submatrix. 
% read entire rows or columns of $\A$. 
Our algorithm is the first such algorithm which does not read a submatrix, and instead reads a carefully selected non-adaptive pattern of entries in rows and columns of $\A$. We complement our algorithm with a matching $\widetilde{\Omega}(d^2/\epsilon)$ query complexity lower bound for non-adaptive testers over any field. We also give tight bounds of $\widetilde{\Theta}(d^2)$ queries in the sensing model for which query access comes in the form of $\langle \X_i, \A\rangle:=\tr(\X_i^\top \A)$; perhaps surprisingly these bounds do not depend on $\epsilon$.

Testing rank is only one of many tasks in determining if a matrix has low intrinsic dimensionality. We next develop a novel property testing framework for testing numerical properties of a real-valued matrix $\A$ more generally, which includes the stable rank, Schatten-$p$ norms, and SVD entropy. Specifically, we propose a {\it bounded entry model}, where $\A$ is required to have entries bounded by $1$ in absolute value. Such a model provides a meaningful framework for testing numerical quantities and avoids trivialities caused by single entries being arbitrarily large. It is also well-motivated by recommendation systems. 
We give upper and lower bounds for a wide range of problems in this model, and discuss connections to the sensing model above. We obtain several results for estimating the operator norm that may be of independent interest. For example, we show that if the stable rank is constant, $\|\A\|_F = \Omega(n)$, and the singular value gap $\sigma_1(\A)/\sigma_2(\A) = (1/\epsilon)^{\gamma}$ for any constant $\gamma >0$, then the operator norm can be estimated up to a $(1\pm\epsilon)$-factor non-adaptively by querying $\cO(1/\epsilon^2)$ entries. This should be contrasted to adaptive methods such as the power method, or previous non-adaptive sampling schemes based on matrix Bernstein inequalities which read a $1/\epsilon^2\times 1/\epsilon^2$ submatrix and thus make $\Omega(1/\epsilon^4)$ queries. Similar to our non-adaptive algorithm for testing rank, our scheme instead reads a carefully selected pattern of entries.
\end{abstract}

\thispagestyle{empty}
\setcounter{page}{0}

\newpage
\tableofcontents

\section{Introduction}
Data intrinsic dimensionality is a central object of study in compressed sensing, sketching, numerical linear algebra, machine learning, and many other domains~\cite{lin2017low,krauthgamer2003property,zhang2013counterexample,woodruff2014sketching,clarkson2013low,zhang2015relations,zhang2014robust}. In compressed sensing and sketching, the study of intrinsic dimensionality has led to significant advances in compressing the data to a size that is far smaller than the ambient dimension while still preserving useful properties of the signal~\cite{nakos2018improved,awasthi2016learning}. In numerical linear algebra and machine learning, understanding intrinsic dimensionality serves as a necessary condition for the success of various subspace recovery problems~\cite{hardt2013algorithms}, e.g., matrix completion~\cite{zhang2016completing,hardt2014understanding,jain2013low,sun2015guaranteed} and robust PCA~\cite{balcan2018matrix,zhang2015exact,bouwmans2018applications}. The focus of this work is on the intrinsic dimensionality of matrices, such as the rank, stable rank, Schatten-$p$ norms, and SVD entropy. The stable rank is defined to be the squared ratio of the Frobenius norm and the largest singular value, and the Schatten-$p$ norm is the $\ell_p$ norm of the singular values (see Appendix \ref{section: Entropy} for our definition of SVD entropy). We study these quantities in the framework of non-adaptive property testing~\cite{parnas2003testing,chan2014optimal,daskalakis2013testing}: given non-adaptive query access to the unknown matrix $\A\in\mathbb{F}^{n\times n}$ over a field $\mathbb{F}$, our goal is to determine whether $\A$ is of dimension $d$ (where dimension depends on the specific problem), or is $\epsilon$-far from having this property. The latter means that at least an $\epsilon$-fraction of entries of $\A$ should be modified in order to have dimension $d$. Query access typically comes in the form of reading a single entry of the matrix, though we will also discuss sensing models where a query returns the value $\langle\X_i,\A\rangle:= \tr(\X_i^\top\A)$ for a given $\X_i$. Without making assumptions on $\A$, we would like to choose our sample pattern or set $\{\X_i\}$ of query matrices so that the query complexity is as small as possible.

Despite a large amount of work on testing matrix rank, many fundamental questions remain open. In the rank testing problem in the sampling model, one such question is to design an efficient algorithm that can distinguish rank-$d$ vs. $\epsilon$-far from rank-$d$ with optimal sample complexity. %In the adaptive setting where the choice of each query is allowed to depend on the observations in the past rounds, the sample complexity can be made as small as $\widetilde{\cO}(d^2/\epsilon)$~\cite{li2014improved}. However, this problem becomes more challenging in the context of the non-adaptive setting. 
The best-known sampling upper bound for non-adaptive rank testing for general $d$ is $\cO(d^2/\epsilon^2)$, which is achieved simply by sampling an $\cO(d/\epsilon)\times \cO(d/\epsilon)$ submatrix uniformly at random~\cite{krauthgamer2003property}. For arbitrary fields $\mathbb{F}$, only an $\Omega((1/\epsilon)\log(1/\epsilon))$ lower bound for constant $d$ is known~\cite{li2014improved}. 
%for a special case of $d=1$, Li et al.~\cite{li2014improved} claimed that they improved the sample complexity to $\widetilde{\cO}(1/\epsilon)$,\footnote{Unfortunately, we find an error in the proof of Theorem 3 in \cite{li2014improved}; see Section \ref{section: Our Techniques} for detailed discussions.} 
%but a tight bound for the general $d$ is still left as an open question. 
%In the rank testing problem in the sensing model, although a na\"ive $\cO(d^2)$ upper bound is possible by random sketching, a matching lower bound is typically hard to achieve as the query oracle becomes more powerful compared with the sampling model. Therefore, the problem of having tight query complexity bounds for rank testing remains open under both sampling and sensing models.

Besides the rank problem above, testing many numerical properties of real matrices has yet to be explored. For example, it is unknown what the query complexity is for the stable rank, which is a natural relaxation of rank in applications. Other examples for which previously we had no bounds are the Schatten-$p$ norms and SVD entropy. We discuss these problems in a new property testing framework that we call the \emph{bounded entry model}. This model has many realistic applications in the Netflix challenge~\cite{koren2009matrix}, where each entry of the matrix corresponds to the rating from a customer to a movie, ranging from 1 to 5. Understanding the query complexity of testing numerical properties in the bounded entry model is an important problem in recommendation systems and applications of matrix completion, where often entries are bounded.

\comment{ 
There are some lower bounds that come from the sketching literature: for stable rank $d$ matrices, with fewer than $\cO(d^2/\epsilon^2)$ non-adaptive samples, a $(1\pm\epsilon)$-approximation to the largest singular value is not possible \cite{li2016tight}. When no assumptions are made on the stable rank, there is an $\cO(n^2/\alpha^4)$ upper bound for an $\alpha$-approximation to the largest singular value. Another line of research~\cite{kong2017spectrum,khetan2017spectrum} considers estimating  the Schatten-$p$ norm up to a constant factor, and when $p$ is as large as $\Theta(\log n/\epsilon)$, this estimator approximates the largest singular value to the desired $(1\pm\epsilon)$-factor accuracy. However, the sample complexity in those works is exponential in $p$.
}

\comment{
\medskip
\noindent{\textbf{Connections to Machine Learning.}}
We provide both positive and negative answers to the questions above. Our study is motivated by connections between property testing and machine learning~\cite{goldreich1998property,ron2008property}: property testing can determine inexpensively whether learning with a low-rank hypothesis class is worthwhile~\cite{balcan2012active}. Problems related to our model include matrix completion~\cite{balcan2016noise,zhang2016completing,balcan2018matrix,ge2016matrix,hardt2014computational} and robust PCA~\cite{zhang2015exact}, which have applications in recommendation system, video surveillance, etc~\cite{lin2017low}. In these problems, the data points are the column vectors of a matrix $\A=\L+\S$, where $\L$ is a low-rank matrix and $\S$ is only guaranteed to be a sparse matrix; note the entries in $\S$ may be arbitrarily large. Our result gives a near-optimal decision version of the robust PCA problem with a sample size independent of $n$: either $\A=\L$ for a rank-$d$ matrix $\L$, or if $\A=\L+\S$, then necessarily more than an $\epsilon$-fraction of entries of $\S$ are non-zero. Distinguishing the two cases helps detect, inexpensively, whether there is a rank-$d$ matrix that agrees with the observed entries which are corrupted with sparse noise~\cite{recht2011simpler,candes2011robust,pimentel2016converse,ashraphijuo2018deterministic}. In comparison, learning algorithms require $\Omega(nd\log n)$ samples for complete recovery of an underlying rank-$d$ matrix~\cite{candes2010power}. 

Our study of various numerical properties of matrices has applications to approximate matrix multiplication~\cite{cohen2015optimal}, random sparsification~\cite{tropp2015introduction}, column subset selection, matrix factorization, and eigenvalue optimization~\cite{tropp2009column,rudelson2007sampling}, since these measure the intrinsic dimensionality of matrices~\cite{tropp2015introduction}. A byproduct of our algorithms is a non-adaptive algorithm for estimating the largest singular value of an $n\times n$ matrix.
%, which is probably one of the most important problems in machine learning --- the PCA problem~\cite{koltchinskii2016asymptotics}. 
It is well-known that the power method is an efficient \emph{adaptive} algorithm for this problem. However, it typically requires an eigengap and needs to read all the entries of a matrix in order to compute a matrix-vector product. Other work has applied cycle-like estimators with sample complexity of $\mathsf{poly}(n)$~\cite{kong2017spectrum,khetan2017spectrum,li2014sketching}. In contrast, we provide the first \emph{non-adaptive} algorithm to estimate the largest singular value up to $(1\pm\epsilon)$ relative error with $\mathsf{poly}(d\log (n)/\epsilon)$ samples.}

\subsection{Problem Setup, Related Work, and Our Results}
\label{section: Problem Setup and Our Results}
Our work has two parts: (1)
we resolve the query complexity of non-adaptive matrix rank testing, a well-studied problem in this model, and (2) we develop a new framework for testing 
numerical properties of real matrices, including the stable rank, the Schatten-$p$ norms and the SVD entropy.
Our results are summarized in Table~\ref{table: result}. We use $\widetilde\cO$ and $\widetilde\Omega$ notation to hide polylogarithmic factors in the arguments inside. For the rank testing results, the hidden polylogarithmic factors depend only on $d$ and $1/\epsilon$ and do not depend on $n$; for the other problems, they may depend on $n$.

\medskip
\noindent{\textbf{Rank Testing.}}
We first study the rank testing problem when we can only non-adaptively query entries. The goal is to design a sampling scheme on the entries of the unknown matrix $\A$ and an algorithm so that we can distinguish whether $\A$ is of rank $d$, or at least an $\epsilon$-fraction of entries of $\A$ should be modified in order to reduce the rank to $d$. This problem was first proposed by Krauthgamer and Sasson in~\cite{krauthgamer2003property} with a sample complexity upper bound of $\cO(d^2/\epsilon^2)$. In this work, 
we improve this to $\widetilde{\cO}(d^2/\epsilon)$ for every $d$ and $\epsilon$, and complement this with a matching lower bound, showing that any algorithm with constant success probability requires at least $\widetilde{\Omega}(d^2/\epsilon)$ samples:

\medskip
\noindent{\textbf{Theorems \ref{theorem: correctness of algorithm for general d}, \ref{theorem: lower bound of rank testing over finite fields under sampling model}, and \ref{theorem: lower bound of rank testing over reals}} (Informal)\textbf{.}}
\emph{For any matrix $\A\in\mathbb{F}^{n\times n}$ over any field, there is a randomized non-adaptive sampling algorithm which reads $\widetilde{\cO}(d^2/\epsilon)$ entries and runs in $\mathsf{poly}(d/\epsilon)$ time, and with high probability correctly solves the rank testing problem. Further, any non-adaptive algorithm with constant success probability requires $\widetilde{\Omega}(d^2/\epsilon)$ samples over $\mathbb{R}$ or any finite field.}

\medskip
Our non-adaptive sample complexity bound of $\widetilde\cO(d^2/\epsilon)$ matches what is known with adaptive queries~\cite{li2014improved}, and thus we show the best known upper bound might as well be non-adaptive.
%Our result is also related to work of Barman et al.~\cite{barman2016dictionary}, who tested whether $\rank(\A)\le d$ or $\A$ is $\epsilon$-far from $\rank(\A)\le 20d/\epsilon^2$ using a different definition of ``$\epsilon$-far'' in terms of $\epsilon$-approximate rank~\cite{alon2013approximate}. In comparison, our testing result requires no rank gap between the two cases.

\begin{table}[t]
\caption{Query complexity results in this paper for non-adaptive testing of the rank, stable rank, Schatten-$p$ norms, and SVD entropy. The testing of the stable rank, Schatten $p$-norm and SVD entropy are considered in the bounded entry model.}
\centering
\label{table: result}
{
\centering
\begin{tabular}{c||c|c|c|c}
\hline
Testing Problems & Rank & Stable Rank & Schatten-$p$ Norm & Entropy\\
 \hline\hline
\multirow{2}{1.5cm}{Sampling} & $\widetilde\cO(d^2/\epsilon)$ (all fields) & $\widetilde\cO(d^3/\epsilon^4)$ & & \multirow{4}{0.8cm}{$\Omega(n)$\textsuperscript{\textdagger}}\\
& $\widetilde\Omega(d^2/\epsilon)$ (finite fields and $\R$) & $\widetilde\Omega(d^2/\epsilon^2)$\textsuperscript{\textdagger} & $\widetilde\cO(1/\epsilon^{4p/(p-2)})$ ($p>2$) & \\
\cline{1-3}
\multirow{2}{1.4cm}{Sensing} & $\cO(d^2)$ (all fields) & $\widetilde{\cO}(d^{2.5}/\epsilon^2)$ & $\Omega(n)$ ($p\in[1,2)$) & \\
& $\widetilde\Omega(d^2)$ (finite fields) & $\widetilde\Omega(d^2/\epsilon^2)$\textsuperscript{\textdagger} &  & \\
\hline
\end{tabular}
}

\rule{0in}{1.2em}{\footnotesize \textsuperscript{\textdagger} The lower bound involves a reparameterization of the testing problem. Please see the respective theorem for details.}
\end{table}

\medskip
\noindent{\textbf{New Framework for Testing Matrix Properties.}} Testing rank is only one of many tasks in determining if a matrix
has low intrinsic dimensionality. In several applications, we require a less fragile measure of the collinearity of rows and columns, which is known as the stable rank~\cite{tropp2015introduction}. We introduce what we call the \emph{bounded entry model} as a new framework for studying such problems through the lens of property testing. In this model, we require all entries of a matrix to be bounded by $1$ in absolute value. Boundedness has many natural applications in recommendation systems, e.g., the user-item matrix of preferences for products by customers has bounded entries in the Netflix challenge~\cite{koren2009matrix}. Indeed, there are many user rating matrices, etc., which naturally have a small
number of discrete values, and therefore fit into a bounded entry model. 
The boundedness of entries also avoids trivialities in which one can modify
a matrix to have 
a property by setting a single entry to be 
arbitrarily large, which, e.g., could make 
the stable rank arbitrarily close to $1$. 

%An example problem we study is whether $\A\in\mathbb{R}^{n\times n}$ is of stable rank $d$, or at least an $\epsilon/d$-fraction of entries of $\A$ should be modified in order to become a matrix of stable rank $d$. 
Our model is a generalization of previous work in which stable rank testing
was done in a model for which all rows had to have bounded norm
\cite{li2014improved}, and the algorithm is only allowed to change entire
rows at a time. As our non-adaptive rank testing algorithm will illustrate,
one can sometimes do better by only reading certain carefully selected
entries in rows and columns. Indeed, this is precisely the source of our
improvement over prior work. Thus, the restriction of having to read an entire
row is often unnatural, and further motivates our bounded entry model.  
%
%Our study relates to the previous work which discussed the stable rank testing in the \emph{bounded row model}~\cite{li2014improved}, where ``$\epsilon/d$-far'' means it requires changing $\epsilon/d$-fraction of the rows with bounded $\ell_2$-norm $1$ in order to have stable rank at most $d$. Li et al.~\cite{li2014improved} provided an algorithm in that setting with $\widetilde\cO(d/\epsilon^2)$ samples. 
%However, the analysis is more challenging in our bounded entry model, because we are allowed to have more flexible entry-level change of the matrix while it is unclear how adversarial change of entries would affect the spectrum of a matrix. 
We first informally state our main theorems on stable rank testing in this model. 

\medskip
\noindent{\textbf{Theorem \ref{theorem: stable rank testing upper bound}} (Informal)\textbf{.}}
\emph{There is a randomized algorithm for the stable rank testing problem to decide whether a matrix is of stable rank at most $d$ or is $\epsilon$-far from stable rank at most $d$, with failure probability at most $1/3$, and which reads $\widetilde\cO(d^3/\epsilon^4)$ entries.}

\medskip
Theorem \ref{theorem: stable rank testing upper bound} relies on a new $(1\pm\tau)$-approximate non-adaptive estimator of the largest singular value of a matrix, which may be of independent interest. 
%
%under the sampling and sensing models with $\mathsf{polylog}(n)$ query complexity, respectively. We present our results below as it may be of independent interest.

\medskip
\noindent{\textbf{Theorem \ref{theorem: operator norm estimator by Ismail}} (Informal)\textbf{.}}
\emph{Suppose that $\A\in\mathbb{R}^{n\times n}$ has stable rank $\cO(d)$ and $\|\A\|_F^2=\Omega(\tau n^2)$. Then in the bounded entry model, there is a randomized non-adaptive sampling algorithm which reads $\widetilde\cO(d^2/\tau^4)$ entries and with probability at least $0.9$, outputs a $(1\pm\tau)$-approximation to the largest singular value of $\A$.}

\medskip
We remark that when the stable rank is constant and the singular value gap $\sigma_1(\A)/\sigma_2(\A) = (1/\tau)^{\gamma}$ for an arbitrary constant $\gamma >0$, the operator norm can be estimated up to a $(1\pm\tau)$-factor by querying $\cO(1/\tau^2)$ entries non-adaptively. We defer these and related results to Appendix \ref{section: Estimation with Eigengap}.

Other measures of intrinsic dimensionality include matrix norms, such as the Schatten-$p$ norm $\|\cdot\|_{\cS_p}$, which measures the central tendency of the singular values. Familiar special cases are $p=1$, $2$ and $\infty$, which have applications in differential privacy~\cite{hardt2012simple} and non-convex optimization~\cite{balcan2018matrix,deshpande2011algorithms} for $p=1$, and in numerical linear algebra~\cite{mahoney2011randomized} for $p\in\{2,\infty\}$. Matrix norms have been studied extensively in the streaming literature \cite{li2014sketching,li2016approximating,li2016tight,LW17}, though their study in property testing models is lacking. 

We study non-adaptive algorithms for these problems in the bounded entry model. We consider distinguishing whether $\|\A\|_{\cS_p}^p$ is at least $cn^p$ for $p>2$ (at least $cn^{1+1/p}$ for $p<2$), or at least an $\epsilon$-fraction of entries of $\A$ should be modified in order to have this property, where $c$ is a constant (depending only on $p$). We choose the threshold $n^p$ for $p>2$ and $n^{1+1/p}$ for $p<2$ because they are the largest possible value of $\|\A\|_{\cS_p}^p$ for $\A$ under the bounded entry model. When $p>2$, $\|\A\|_{\cS_p}$ is maximized when $\A$ is of rank $1$, and so this gives us an alternative ``measure'' of how close we are to a rank-$1$ matrix. Testing whether $\|\A\|_{\cS_p}$ is large in sublinear time allows us to quickly determine whether $\A$ can be well approximated by a low-rank matrix, which could save us from running more expensive low-rank approximation algorithms. In contrast, when $p<2$, $\|\A\|_{\cS_p}$ is maximized when $\A$ has a flat spectrum, and so is a measure of how well-conditioned $\A$ is. A fast tester could save us from running expensive pre-conditioning algorithms. We state our main theorems informally below.

\medskip
\noindent{\textbf{Theorem \ref{theorem: Schatten-p norm upper bound}} (Informal)\textbf{.}}
\emph{For constant $p>2$, there is a randomized algorithm for the Schatten-$p$ norm testing problem with failure probability at most $1/3$ which reads $\widetilde\cO(1/\epsilon^{4p/(p-2)})$ entries.}

\medskip
\noindent{\textbf{Results for Sensing Algorithms.}}
We also consider a more powerful query oracle known as the \emph{sensing model}, where query access comes in the form of $\langle\X_i,\A\rangle := \tr(\X_i^\top\A)$ for some sensing matrices $\X_i$ of our choice. 
These matrices are chosen non-adaptively.
We show differences in the complexity of the above problems in this and the above sampling model.  For the testing and the estimation problems above, we have the following results in the sensing model:

\medskip
\noindent{\textbf{Theorem \ref{theorem: lower bound over finite fields in general basis}} (Informal)\textbf{.}}
\emph{Over an arbitrary finite field, any non-adaptive algorithm with constant success probability for the rank testing problem in the sensing model requires $\widetilde\Omega(d^2)$ queries.}

\medskip
\noindent{\textbf{Theorems \ref{theorem: stable rank testing upper bound} and \ref{theorem: stable rank lower bound}} (Informal)\textbf{.}}
\emph{There is a randomized algorithm for the stable rank testing problem with failure probability at most $1/3$ in the sensing model with $\widetilde\cO(d^{2.5}/\epsilon^2)$ queries. Further, any algorithm with constant success probability requires $\widetilde{\Omega}(d^2/\epsilon^2)$ queries.}

\medskip
\noindent{\textbf{Theorem \ref{thm:schatten_p_lb}} (Informal)\textbf{.}}
\emph{For $p\in[1,2)$, any algorithm for the Schatten-$p$ norm testing problem with failure probability at most $1/3$ requires $\Omega(n)$ queries.}

\medskip
\noindent{\textbf{Theorem \ref{theorem: operator norm estimator under sensing model}} (Informal)\textbf{.}}
\emph{Suppose that $\A\in\mathbb{R}^{n\times n}$ has stable rank $\cO(d)$ and $\|\A\|_F^2=\Omega(\tau n^2)$. In the bounded entry model, there is a randomized sensing algorithm with sensing complexity $\widetilde\cO(d^2/\tau^2)$ which outputs a $(1\pm\tau)$-approximation to the largest singular value with probability at least $0.9$. This sensing complexity is optimal up to polylogarithmic factors.}

\medskip
We also provide an $\Omega(n)$ query lower bound for the SVD entropy testing in the sensing model. We defer the definition of the problem and related results to Section~\ref{section: Entropy}.

\subsection{Our Techniques}
\label{section: Our Techniques}
We now discuss the techniques in more detail, starting with the rank testing problem. 

%We begin by clarifying the flawed proof of Theorem 3 in \cite{li2014improved} for $d=1$. 
%We start by discussing one aspect of \cite{li2014improved}
Prior to the work of \cite{li2014improved}, the only known algorithm for $d = 1$ was to sample an $\cO(1/\epsilon)\times \cO(1/\epsilon)$ submatrix. In contrast, for rank $1$ an algorithm in \cite{li2014improved} samples $\cO(\log(1/\epsilon))$ blocks of varying shapes ``within a random $\cO(1/\epsilon)\times \cO(1/\epsilon)$ submatrix'' and argues that these shapes are sufficient to expose a rank-$2$ submatrix. For $d=1$ the goal is to augment a $1\times 1$ matrix to a full-rank $2\times 2$ matrix. One can show that with good probability, one of the shapes ``catches'' an entry that enlarges the $1\times 1$ matrix to a full-rank $2\times 2$ matrix. For instance, in Figure~\ref{figure: our sampling scheme}, $(r,c)$ is our $1\times 1$ matrix and the leftmost vertical block catches an ``augmentation element'' $(r',c')$ which makes $\left[\begin{smallmatrix} (r,c') & (r,c)\\ (r',c') & (r',c)\end{smallmatrix}\right]$ a full-rank $2\times 2$ matrix. Hereby, the ``augmentation element'' means the entry by adding which we augment a $r\times r$ matrix to a $(r+1)\times(r+1)$ matrix. In \cite{li2014improved}, an argument was claimed for $d = 1$, though we note an omission in their analysis. Namely, the ``augmentation entry'' $(r',c')$ can be the $1\times 1$ matrix we begin with (meaning that $\A_{r',c'}\neq 0$, which might not be true), and since one can show that both $(r,c)$ and $(r',c')$ fall inside the same sampling block with good probability, the $2\times 2$ matrix would be fully observed and the algorithm would thus be able to determine that it has rank $2$. However, it is possible that $\A_{r',c'} = 0$ and $(r',c')$ would not be a starting point (i.e., a $1\times 1$ rank-$1$ matrix), and in this case, $(r',c)$ may not be observed, as illustrated in Figure~\ref{figure: our sampling scheme}. In this case the algorithm will not be able to determine whether the augmented $2\times 2$ matrix is of full rank. For $d > 1$, nothing was known. One issue is that the probability of fully observing a $d\times d$ submatrix within these shapes is very small. To overcome this, we propose what we call \emph{rebasing} and \emph{transformation to a canonical structure}. These arguments allow us to tolerate unobserved entries and conveniently obtain an algorithm for every $d$, completing the analysis of \cite{li2014improved} for $d = 1$ in the process. 

\begin{figure}
\centering
%\vspace{-0.6cm}
\includegraphics[width=4.5cm]{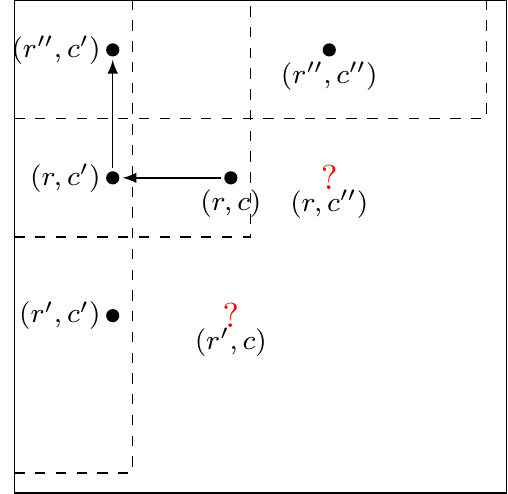}
%\vspace{-0.35cm}
\caption{\small{Our sampling scheme (the region enclosed by the dotted lines modulo permutation of rows and columns) and our path of augmenting a $1\times 1$ submatrix. The whole region is the $\cO(d/\epsilon)\times \cO(d/\epsilon)$ submatrix sampled from the $n\times n$ matrix.}}
\label{figure: our sampling scheme}
\end{figure}

\medskip
\noindent{\textbf{Rebasing Argument + Canonical Structure.}} The best previous result for the rank testing problem uniformly samples an $\cO(d/\epsilon)\times \cO(d/\epsilon)$ submatrix and argues that one can find a $(d+1)\times (d+1)$ full-rank submatrix within it when $\A$ is $\epsilon$-far from rank-$d$~\cite{krauthgamer2003property}. In contrast, our algorithm follows from subsampling an $\cO(\epsilon)$-fraction of entries in this $\cO(d/\epsilon)\times \cO(d/\epsilon)$ submatrix. Let $\cR_1\subseteq \cdots \subseteq \cR_m$ and $\cC_1\supseteq \cdots \supseteq \cC_m$ be the indices of subsampled rows and columns, respectively, with $m=\cO(\log(1/\epsilon))$. We choose these indices uniformly at random such that $|\cR_i|=\widetilde\cO(d2^i)$ and $|\cC_i|=\widetilde\cO(d/(2^i\epsilon))$, and sample the entries in all $m$ blocks determined by the $\{\cR_i,\cC_i\}$ (see Figure~\ref{figure: our sampling scheme}, where our sampled regions are enclosed by the dotted lines). Since there are $\widetilde\cO(\log(1/\epsilon))$ blocks and in each block we sample $\widetilde\cO(d^2/\epsilon)$ entries, the sample complexity of our algorithm is as small as $\widetilde\cO(d^2/\epsilon)$.

The correctness of our algorithm for $d=1$ follows from what we call a rebasing argument. Starting from an empty matrix, our goal is to maintain and augment the matrix to a $2\times 2$ full-rank matrix when $\A$ is $\epsilon$-far from rank-$d$. By a level-set argument, we show an oracle lemma which states that \emph{we can augment any $r\times r$ full-rank matrix to an $(r+1)\times (r+1)$ full-rank matrix by an augmentation entry in the sampled region}, as long as $r\le d$ and $\A$ is $\epsilon$-far from rank-$d$. Therefore, as a first step we successfully find a $1\times 1$ full-rank matrix, say with index $(r,c)$, in the sampled region. We then argue that we can either (a) find a $2\times 2$ fully-observed full-rank submatrix or a $2\times 2$ submatrix which is not fully observed but we know must be of full rank, or (b) move our maintained $1\times 1$ full-rank submatrix upwards or leftwards to a new $1\times 1$ \emph{full-rank} submatrix and repeat checking whether case (a) happens or not; if not, we implement case (b) again and repeat the procedure. To see case (a), by the oracle lemma, if the augmented entry is $(r'',c')$ (see Figure \ref{figure: our sampling scheme}), then we fully observe the submatrix determined by $(r'',c')$ and $(r,c)$ and so the algorithm is correct in this case. On the other hand, if the augmented entry is $(r',c')$, then we fail to see the entry at $(r',c)$. In this case, when $\A_{r,c'}= 0$, then we must have $\A_{r',c'}\neq 0$; otherwise, $(r',c')$ is not an augment of $(r,c)$, which leads to a contradiction with the oracle lemma. Thus we find a $2\times 2$ matrix with structure
\begin{equation}
\label{equ: structure for d=1}
\begin{bmatrix}
\A_{r,c'} & \A_{r,c}\\
\A_{r',c'} & \A_{r',c}
\end{bmatrix}
=
\begin{bmatrix}
0 & \neq 0\\
\neq 0 & ?
\end{bmatrix},
\end{equation}
which must be of rank $2$ despite an unobserved entry, and the algorithm therefore is correct in this case. The remaining case of the analysis above is when $\A_{r,c'}\neq 0$. Instead of trying to augment $\A_{r,c}$, we augment $\A_{r,c'}$ in the next step. Note that the index $(r,c')$ is to the left of $(r,c)$. This leads to case (b). In the worst case, we move the $1\times 1$ non-zero matrix to the uppermost left corner,\footnote{The upper-left corner refers to the intersection of all sampled blocks, namely, $\cR_1\times \cC_m$; it does not mean the top-left entry.} e.g., $(r'',c')$. Fortunately, since $(r'',c')$ is in the uppermost left corner, we can, as guaranteed by the oracle lemma, augment it to a $2\times 2$ \emph{fully-observed} full-rank matrix. Again the algorithm outputs correctly in this case.

The analysis becomes more challenging for general $d$, since the number of unobserved/unimportant entries (i.e., those entries marked as ``$?$'') may propagate as we augment an $r\times r$ submatrix ($r=1,2,...,d$) in each round. To resolve the issue, we maintain a structure (modulo elementary transformations) similar to structure \eqref{equ: structure for d=1} for the $r\times r$ submatrix, that is,
\begin{equation}
\label{equ: structure for general d}
\begin{bmatrix}
0 & 0 & \cdots & 0 & \cdots & 0 & \ne 0\\
0 & 0 & \cdots & 0 & \cdots & \ne 0 & ?\\
\vdots & \vdots &  & \vdots &  & \vdots & \vdots\\
0 & \ne 0 & \cdots & ? & \cdots & ? & ?\\
\ne 0 & ? & \cdots & ? & \cdots & ? & ?\\
\end{bmatrix}.
\end{equation}
Since the proposed structure has non-zero determinant, the submatrix is always of full rank. Similar to the case for $d=1$, we show that we can either (a) augment the $r\times r$ submatrix to an $(r+1)\times (r+1)$ submatrix with the same structure \eqref{equ: structure for general d} (modulo elementary transformations); or (b) find another $r\times r$ submatrix of structure \eqref{equ: structure for general d} that is closer to the upper-left corner than the original $r\times r$ matrix. Hence the algorithm is correct for general $d$. More details are provided in the proof of Theorem \ref{theorem: correctness of algorithm for general d}.

\medskip
\noindent{\textbf{Pivot-Node Assignment.}} Our rank testing lower bound under the sampling model over a finite field $\mathbb{F}$ follows from distinguishing two hard instances $\U\V^\top$ vs. $\W$, where $\U,\V\in\bbF^{t\times d}$ and $\W\in\bbF^{t\times t}$ have i.i.d.\ entries that are uniform over $\mathbb{F}$. For an observed subset $\cS$ of entries with $|\cS|=\cO(d^2)$, we bound the total variation distance between the distributions of the observed entries in the two cases by a small constant. In particular, we show that the probability $\Pr[(\U\V^\top)|_{\cS}=\x]$ is large for any observation $\x\in\mathbb{F}^{|\cS|}$, by a \emph{pivot-node assignment} argument, as follows. We reformulate our problem as a bipartite graph assignment problem $G=(L\cup R,E)$, where $L$ corresponds to the rows of $\U$, $R$ the rows of $\V$ and each edge of $E$ one entry in $\cS$. We want to assign each node a vector/affine subspace, meaning that the corresponding row in $\U$ or $\V$ will be that vector or in that affine subspace, such that they agree with our observation, i.e., $(\U\V^\top)|_{\cS}=\x$. Since $\U,\V$ are random matrices, we assign random vectors to nodes adaptively, one at a time, and try to maintain consistency with the fact that $(\U\V^\top)|_{\cS} = \x$. Note that the order of the assignment is important, as a bad choice for an earlier node may invalidate any assignment to a later node. To overcome this issue, we choose nodes of large degrees as \emph{pivot nodes} and assign each non-pivot node adaptively in a careful manner so as to guarantee that the incident pivot nodes will always have valid assignments (which in fact form an affine subspace). In the end we assign the pivot node vectors from their respective affine subspaces. We employ a counting argument for each step in this assignment procedure to lower bound the number of valid assignments, and thus lower bound the probability $\Pr[(\U\V^\top)|_{\cS}=\x]$.

The above analysis gives us an $\Omega(d^2)$ lower bound for constant $\epsilon$ since $\W$ is constant-far from being of rank $d$. The desired $\Omega(d^2/\epsilon)$ lower bound follows from planting $\U\V^\top$ vs. $\W$ with $t=\sqrt{\epsilon}n$ into an $n\times n$ matrix at uniformly random positions, and padding zeros everywhere else.

\medskip
\noindent{\textbf{New Analytical Framework for Stable Rank, Schatten-$p$ Norm, and Entropy Testing.}} We propose a new analytical framework by reducing the testing problem to a sequence of estimation problems \emph{without involving $\mathsf{poly}(n)$ in the sample complexity}. There is a two-stage estimation in our framework: (1) a constant-approximation to some statistic $X$ of interest (e.g., stable rank) which enables us to distinguish $X\le d$ vs. $X\ge 10d$ for the threshold parameter $d$ of interest. If $X\ge 10d$, we can safely output ``$\A$ is far from $X\le d$''; otherwise, the statistic is at most $10d$, and (2) we show that $X$ has a $(1\pm\epsilon)$-factor difference between ``$X\leq d$'' and ``far from $X\leq d$'', and so we implement a more accurate $(1\pm\epsilon)$-approximation to distinguish the two cases. The sample complexity does not depend on $n$ polynomially because (1) the first estimator is ``rough'' and gives only a constant-factor approximation and (2) the second estimator operates under the condition that $X\leq 10d$ and thus $\A$ has a low intrinsic dimension. We apply the proposed framework to the testing problems of the stable rank and the Schatten-$p$ norm by plugging in our estimators in Theorem~\ref{theorem: operator norm estimator by Ismail} and Theorem~\ref{theorem: operator norm estimator under sensing model}. This analytical framework may be of independent interest to other property testing problems more broadly.

In a number of these problems, a key difficulty is arguing about spectral
properties 
of a matrix $\A$
when it is $\epsilon$-far from having a property, such as having
stable rank at most $d$. 
Because of the fact that
the entries must always be bounded by $1$ in absolute value, it becomes non-trivial
to argue, for example, that if $\A$ is $\epsilon$-far from having stable
rank at most $d$, that its stable rank is even slightly larger than $d$. A natural
approach is to argue that you could change an $\epsilon$-fraction of rows 
of $\A$ to agree with a multiple of the
top left or right singular vector of $\A$, and since we are
still guaranteed to have stable rank at least $d$ after changing such entries, 
it means that the operator
norm of $\A$ must have been small to begin with (which says something about the
original stable rank of $\A$, since its Frobenius norm can also be estimated).
The trouble is, if the top singular vector has some entries that are very large,
and others that are small, one cannot scale the singular vector by a large
amount since then we would violate the boundedness criterion of our model. We
get around this by arguing there either needs to exist a left or a right singular
vector of large $\ell_1$-norm (in some cases such vectors may only be right
singular vectors, and in other cases only left singular vectors). The $\ell_1$-norm
is a natural norm to study in this context, since it is dual to the $\ell_{\infty}$-norm, which we use to capture the boundedness property of the matrix. 

Our lower bounds for the above problems follow from the corresponding sketching lower bounds for the estimation problem in \cite{li2016tight,LNW:schatten_unpublished}, together with rigidity-type results \cite{valiant1977graph} for the hard instances regarding the respective statistic of interest.

\section{Preliminaries}
We shall use bold capital letters $\A$, $\B$, ... to indicate matrices, bold lower-case letters $\u$, $\v$, ... to indicate vectors, and lower-case letters $a$, $b$, ... to indicate scalars. We adopt the convention of abbreviating the set $\{1,2,...,n\}$ as $[n]$. We write $f\gtrsim g$ (resp. $f\lesssim g$) if there exists a constant $C>0$ such that $f\ge Cg$ (resp. $f\le Cg$).

For matrix functions, denote by $\rank(\A)$ and $\srank(\A)$ the rank and the stable rank of $\A\neq\0$, respectively. It always holds that $1\le \srank(\A)\le \rank(\A)$. For matrix norms, let $\|\A\|_{\cS_p}$ denote the Schatten-$p$ norm of $\A$, defined as $\|\A\|_{\cS_p}=\left(\sum_{i=1}^n \sigma_i^p(\A)\right)^{1/p}$. The Frobenius norm $\|\A\|_F$ is a special case of the Schatten-$p$ norm when $p=2$, the operator norm or the spectral norm (the largest singular value) of $\|\A\|$ equals to the limit as $p\rightarrow +\infty$. When $0 <p<1$, $\|\A\|_{\cS_p}$ is not a norm but is still a well-defined quantity, and it tends to $\rank(\A)$ as $p\to 0^+$.  Let $\|\A\|_0$ denote the number of non-zero entries in $\A$, and $\|\A\|_\infty$ denote the entrywise $\ell_\infty$ norm of $\A$, i.e., $\|\A\|_\infty = \max_{i,j} |\A_{i,j}|$.
The \emph{rigidity} of a matrix $\A$ over a field $\F$, denoted by $\cR_{\A}^{\F}(r)$, is the least number of entries of $\A$ that must be changed in order to reduce the rank of $\A$ to a value at most $r$:
\begin{equation*}
\cR_{\A}^{\F}(r):=\min\{\|\C\|_0:\rank_{\F}(\A+\C)\le r\}.
\end{equation*}
Sometimes we may omit the subscript $\A$ in $\cR_{\A}^{\F}(r)$ when the matrix of interest is clear from the context.

We define the entropy of an unnormalized distribution $(p_1,\dots,p_n)$ ($ 0 <p_1+\cdots+p_n\leq 1$ with $p_i \geq 0$ for all $i$) to be
\[
H(p_1,\dots,p_n) = -\sum_i p_i\log p_i.
\]
Let $\A\in \R^{n\times n}$, we define its entropy as
\begin{equation}\label{eqn:entropy_defn}
H(\A) = H\left(\frac{\sigma_1^2(\A)}{n^2},\dots,\frac{\sigma_n^2(\A)}{n^2}\right) = \frac{-\sum_i \frac{\sigma_i^2(\A)}{n^2}\log\frac{\sigma_i^2(\A)}{n^2}}{\sum_i \frac{\sigma_i^2(\A)}{n^2}}.
\end{equation}
with the convention that $0\cdot\infty = 0$. For matrices $\A$ satisfying $\|\A\|_\infty\leq 1$, it holds that $\sigma_i(\A)\leq n$ for all $i$ and the entropy above coincides with the usual Shannon entropy. Note that scaling only changes the entropy additively; that is, $H(\beta\A) = H(\A) - \log\beta^2$.

Let $\cG(m,n)$ denote the distribution of $m\times n$ i.i.d. standard Gaussian matrix over $\R$ and $\cU_{\mathbb{F}}(m,n)$ (or $\cU(\cS)$) represent $m\times n$ i.i.d. uniform matrix over a finite field $\mathbb{F}$ (or a finite set $\cS$). We use $d_{TV}(\cL_1,\cL_2)$ to denote the total variation distance between two distributions $\cL_1$ and $\cL_2$.

We shall also frequently use $c$, $c'$, $c_0$, $C$, $C'$, $C_0$, etc., to represent constants, which are understood to be absolute constants unless the dependency is otherwise specified.

\section{Non-Adaptive Rank Testing}
In this section, we study the following problem of testing low-rank matrices.
\begin{problem}[Rank Testing with Parameter $(n,d,\epsilon)$ in the Sampling Model]
\label{problem: problem of property testing of rank}
Given a field $\mathbb{F}$ and a matrix $\A\in\mathbb{F}^{n\times n}$ which has one of promised properties:
\begin{itemize}
\item[$\mathsf{H0.}$]
$\A$ has rank at most $d$;
\item[$\mathsf{H1.}$]
$\A$ is $\epsilon$-far from having rank at most $d$, meaning that $\A$ requires changing at least an $\epsilon$-fraction of its entries to have rank at most $d$.
\end{itemize}
The problem is to design a property testing algorithm that outputs $\mathsf{H0}$ with probability $1$ if $\A\in \mathsf{H0}$, and output $\mathsf{H1}$ with probability at least 0.99 if $\A\in\mathsf{H1}$, with the least number of queried entries.
\end{problem}

\subsection{Positive Results}
\label{section: positive results}
Below we provide a non-adaptive algorithm for the rank testing problem under the sampling model with $\widetilde\cO(\frac{d^2}{\epsilon})$ queries when $\epsilon\le\frac{1}{e}$. Let $\eta\in(0,\frac12)$ be such that $\eta\log(\frac{1}{\eta})=\epsilon$ and let $m = \lceil\log (\frac{1}{\eta})\rceil$.
\begin{algorithm}[ht]
\caption{Robust non-adaptive testing of matrix rank}
\label{algorithm: non-adaptive robust testing of rank}
\begin{algorithmic}[1]
\State Choose $\cR_1,\dots,\cR_m$ and $\cC_1,\dots,\cC_m$ from $[n]$ uniformly at random such that
\begin{equation*}
\cR_1\subseteq \cdots \subseteq \cR_m,\qquad \cC_1\supseteq \cdots \supseteq \cC_m,
\end{equation*}
and
\begin{equation*}
|\cR_i|=c[\log d+\log\log(1/\eta)]d\log(1/\eta)2^i,\qquad |\cC_i|=c[\log d+\log\log(1/\eta)]d\log (1/\eta)/(2^i\eta),
\end{equation*}
where $c>0$ is an absolute constant. To impose containment for $\cR_i$'s, $\cR_i$ can be formed by appending to $\cR_{i-1}$ uniformly random $|\cR_i|-|\cR_{i-1}|$ rows. The containment for $\cC_i$'s can be imposed similarly.

\State  Query the entries in $\cQ=\bigcup_{i=1}^m (\cR_i\times \cC_i)$. Note that the entries in $(\cR_m\times \cC_1)\setminus \cQ$ are unobserved. The algorithm solves the following minimization problem by filling in those entries of $\A_{(\cR_m\times \cC_1)\setminus\cQ}$ given input $\A_{\cQ}$.
\begin{equation}
\label{equ: upper bound}
r:=\min_{\A_{(\cR_m\times \cC_1)\setminus \cQ}} \rank(\A_{\cR_m,\cC_1}).
\end{equation}
\State Output ``$\A$ is $\epsilon$-far from having rank $d$'' if $r>d$; otherwise, output ``$\A$ is of rank at most $d$''.
\end{algorithmic}
\end{algorithm}

We note that the number of entries that Algorithm \ref{algorithm: non-adaptive robust testing of rank} queries is
$$\cO(k\cdot [\log d+\log\log(1/\eta)]^2d^2\log^2(1/\eta)/\eta)=\widetilde\cO(d^2/\epsilon).$$
We now prove the correctness of Algorithm \ref{algorithm: non-adaptive robust testing of rank}. Before proceeding, we reproduce the definitions \emph{augment set} and \emph{augment pattern $i$} and relevant lemmata from~\cite{li2014improved} as follows.

\begin{definition}[Augment]
For $n\times n$ fixed matrix $\A$, we call $(r,c)$ an augment for $\cR\times \cC\subseteq [n]\times [n]$ if $r\in[n]\backslash \cR$, $c\in[n]\backslash \cC$ and $\rank(\A_{\cR\cup\{r\},\cC\cup\{c\}})>\rank(\A_{\cR,\cC})$. We denote by $\mathsf{aug}(\cR,\cC)$ the set of all the augments for $\cR\times \cC$, namely,
\begin{equation*}
\mathsf{aug}(\cR,\cC)=\{(r,c)\in([n]\backslash \cR)\times ([n]\backslash \cC)\mid \rank(\A_{\cR\cup\{r\},\cC\cup\{c\}})>\rank(\A_{\cR,\cC})\}.
\end{equation*}
\end{definition}

\begin{definition}[Augment Pattern]
For fixed $\cR$, $\cC$ and $\A$, define $\mathsf{count}_r$ (where $r\in[n]\backslash \cR$) to be the number of $c$'s such that $(r,c)\in \mathsf{aug}(\cR,\cC)$. Let $\{\mathsf{count}^\ast_i\}_{i\in [n-|\cR|]}$ the non-increasing reordering of the sequence $\{\mathsf{count}_i\}_{i\in [n]\backslash \cR}$, and $\mathsf{count}_{i}^\ast=0$ for $i>n-|\cR|$. We say that $(\cR,\cC)$ has augment pattern $i$ on $\A$ if and only if $\mathsf{count}_{n/2^i}^\ast\ge 2^{i-1}\eta n$.
\end{definition}

\begin{lemma}%[Lemma 4, \cite{li2014improved}]
\label{lemma: augment}
Let $\A_{\cR,\cC}$ be a $t\times t$ full-rank matrix. If $\A$ is $\epsilon$-far from having rank $d$ and $\rank(\A_{\cR,\cC})=t\le d$, then
\begin{equation*}
|\mathsf{aug}(\cR,\cC)|=\sum_{r\in[n]\backslash \cR} \mathsf{count}_r=\sum_{i=1}^{n-|\cR|}\mathsf{count}_i^*\ge \frac{\epsilon n^2}{3}.
\end{equation*}
\end{lemma}

\begin{proof}
Let $\cS$ be the set of entries $(r,c)$ in $\cR^c\times \cC^c$ such that $\rank(\A_{\cR\cup \{r\},\cC\cup \{c\}})>\rank(\A_{\cR,\cC})$, i.e., $\cS=\mathsf{aug}(\cR,\cC)$. We will show that $|\cS|\ge \epsilon n^2/3$.

Let $\cT$ be the complement of $\cS$ inside the set $\cR^c\times \cC^c$. For any $(r,c)\in\cS$, we discuss the following two cases.

\medskip
\noindent{\textbf{Case (i). There is $c'\in \cC^c$ such that $(r,c')\in \cT$ or $r'\in \cR^c$ such that $(r',c)\in \cT$}}

In the former case, the row vector $\A_{r,\cC\cup\{c'\}}$ is a linear combination of the rows of $\A_{\cR,\cC\cup\{c'\}}$. So we can change the value of $\A_{r,c}$ so that $\A_{r,\cC\cup\{c\}}$ is a linear combination of $\A_{\cR,\cC\cup\{c\}}$ with the same representation coefficients as that of $\A_{\cR,\cC\cup\{c'\}}$. Therefore, augmenting $\A_{R,C}$ by the pair $(r,c)$ would not increase $\rank(\A_{\cR,\cC})$. Similarly, if there is $r'\in \cR^c$ such that $(r',c)\in T$, we can change the value of $\A_{r,c}$ so that augmenting $\A_{\cR,\cC}$ by the pair $(r,c)$ would not increase $\rank(\A_{\cR,\cC})$. We change at most $|\cS|$ entries for both cases combined.

\medskip
\noindent{\textbf{Case (ii). $(r,c')\in \cS$ for all $c'\in \cC^c$ and $(r',c)\in \cS$ for all $r'\in \cR^c$}}

In this case, we can change the entire $r$-th row and $c$-th column of $\A$ so that $\rank(\A_{\cR,\cC})$ does not increase by augmenting it with any pair in $(\cR^c\times \{c\})\cup(\{r\}\times \cC^c)$. Recall that $n\ge 2d$ and $t\le d$. It follows that $n\le 2(n-t)$. Therefore, this specific pair $(r,c)$ would lead to the change of at most $2n\le 2(n-t)+2(n-t)\le 2(|\cR^c|+|\cC^c|)$ entries. For all such $(r,c)$'s, we change at most $2|\cS|$ entries in this case.

In summary, we can change at most $3|\cS|$ entries of $\A$ so that $\rank(\A_{\cR,\cC})$ cannot increase by augmenting $\A_{\cR,\cC}$ with any pair $(r,c)\in \cR^c\times \cC^c$. Since $\A$ is $\epsilon$-far from being rank $d$, we must have $3|\cS|\ge \epsilon n^2$. Namely, $|\mathsf{aug}(\cR,\cC)|=|\cS|\ge \epsilon n^2/3$.
\end{proof}

\begin{lemma}%[Lemma 5, \cite{li2014improved}]
\label{lemma: existence of augment pattern}
Let $\A_{\cR,\cC}$ be a $t\times t$ full-rank matrix. If $\A$ is $\epsilon$-far from being rank $d$ and $\rank(\A_{\cR,\cC})=t\le d$, then there exists $i$ such that $(\cR,\cC)$ has augment pattern $i$.
\end{lemma}

\begin{proof}
Suppose that $(\cR,\cC)$ does not have any augment pattern in $[\log(1/\eta)]$. That is
\begin{equation*}
\mathsf{count}_{n/2^i}^\ast<2^{i-1}\eta n,\quad i=1,2,...,\log(1/\eta).
\end{equation*}
Therefore,
\begin{equation*}
\begin{split}
\sum_i \mathsf{count}_i^*&=\sum_{i=\frac{n}{2}+1}^n \mathsf{count}_i^*+\sum_{i=\frac{n}{4}+1}^{\frac{n}{2}} \mathsf{count}_i^*+...+\sum_{i=\frac{n}{2^{\log(1/\eta)}}+1}^{\frac{n}{2^{\log(1/\eta)-1}}}\mathsf{count}_i^*+\sum_{i=1}^{\eta n}\mathsf{count}_i^*\\
&\le \frac{n}{2}\mathsf{count}_{\frac{n}{2}+1}^\ast+\frac{n}{4}\mathsf{count}_{\frac{n}{4}+1}^\ast+\cdots+\frac{n}{2^{\log(1/\eta)}}\mathsf{count}_{\frac{n}{2^{\log(1/\eta)}}+1}^\ast+\eta n\mathsf{count}_{1}^\ast\\
&< \frac{n}{2}\eta n+\frac{n}{4}2\eta n+...+\eta n2^{\log(1/\eta)-1}\eta n+\eta n^2\\
&=\frac{\eta n^2}{2}(\log(1/\eta)+2)\\
&\le \frac{\epsilon n^2}{3},
\end{split}
\end{equation*}
which leads to a contradiction with Lemma \ref{lemma: augment}.
\end{proof}

\begin{lemma}%[Lemma 6, \cite{li2014improved}]
\label{lemma: uniform sampling captures augment pattern}
For fixed $(\cR,\cC)$, suppose that $(\cR,\cC)$ has augment pattern $i$ on $\A$. Let $\cR',\cC'\subseteq [n]$ be uniformly random such that $|\cR'|=c2^i$, $|\cC'|=c/(2^i\eta)$. Then the probability that $(\cR',\cC')$ contains at least one augment of $(\cR,\cC)$ on $\A$ is at least $1-2e^{-c/2}$.
\end{lemma}

\begin{proof}
Since $(\cR,\cC)$ has augment pattern $i$ on matrix $\A$, the probability that $\cR'$ (and $\cC'$) does not hit row (and column) of any augment is $(1-2^{-i})^{c2^i}$ (and $(1-2^{i-1}\eta)^{c/(2^i\eta)}$). Therefore, the probability that $(\cR',\cC')$ hits at least one augment is given by
\[
\left(1-(1-2^{-i})^{c2^i}\right)\left(1-(1-2^{i-1}\eta)^{c/(2^i\eta)}\right)\ge 1-\frac{2}{e^{c/2}}.\qedhere
\]
\end{proof}

\subsubsection{Warm-Up: The Case of $d=1$}

Without loss of generality, we may permute the rows and columns of $\A$ and assume that $\cR_i = \{1,\dots,|\cR_i|\}$ and $\cC_i = \{1,\dots,|\cC_i|\}$ for all $i\leq \lceil\log\frac{1}{\eta}\rceil$.

\begin{theorem}
\label{theorem: correctness of algorithm for d=1}
Let $\epsilon\le 1/e$ and $d=1$. For any matrix $\A$, the probability that Algorithm \ref{algorithm: non-adaptive robust testing of rank} fails is at most $1/3$.
\end{theorem}

\begin{proof}
If $\A$ is of rank at most $d$, then the algorithm will never make mistake; so we assume that $\A$ is $\epsilon$-far from being rank $d$ in the proof below.

Lemma~\ref{lemma: existence of augment pattern} shows that $(\emptyset,\emptyset)$ has some augment pattern $s$ and by Lemma~\ref{lemma: uniform sampling captures augment pattern}, with probability at least $1-2e^{-c/2}$ there exists $(r,c)\in (\cR_s,\cC_s)$ such that $(r,c)\in \textsf{aug}(\emptyset,\emptyset)$, i.e., $\A_{(r,c)}\ne 0$. We now argue that the rank-$1$ submatrix $\A_{(r,c)}$ can be augmented to a rank-$2$ submatrix.

Again by Lemma~\ref{lemma: existence of augment pattern}, $(\{r\},\{c\})$ has an augment pattern $j$; otherwise, $\A$ is not $\epsilon$-far from being rank-$d$, and with probability at least $1-2e^{-c/2}$ there exists $(r',c')\in (\cR_j,\cC_j)$ such that $(r',c')\in \textsf{aug}(\{r\},\{c\})$. We now discuss three cases based on the position of $(r',c')$ in relation to $(r,c)$.

\medskip
\noindent{\textbf{Case (i). $(r',c')\in R_s\times C_s$.}}

By Lemma \ref{lemma: uniform sampling captures augment pattern}, with probability at least $1-2e^{-c/2}$, $\cR_j\times \cC_j$ contains an argument for $({r},{c})$, denoted by $(r',c')$. By construction of $\{\cR_j\}$ and $\{\cC_j\}$, $(r,c')$ and $(r',c)$ are also queried (See Figure \ref{figure: case i}). Thus we find a $2\times 2$ non-singular matrix. The algorithm answers correctly with probability at least $1-4e^{-c_0/4}>2/3$ in this case.
%\begin{equation*}
%\begin{bmatrix}
%\A_{r',c'} & \A_{r',c}\\
%\A_{r,c'} & \A_{r,c}
%\end{bmatrix}
%=
%\begin{bmatrix}
%0 & \ne 0\\
%\ne 0 & ?
%\end{bmatrix}.
%\end{equation*}

\medskip
\noindent{\textbf{Case (ii). $r'\not\in \cR_s$ or $c'\not\in \cC_s$.}}

In this case, we show that starting from $\A_{r,c}$, we can always find a path for the non-singular $1\times 1$ submatrix $\A_{*,*}$ such that the index $(*,*)$ always moves to the left or above, so we make progress towards case (i): we note that the non-zero element in the most upper left corner can always be augmented with three queried elements in the same augment pattern (i.e., Case (i)), because the uppermost left corner belongs to all $(\cR_i,\cC_i)$'s by construction. We now show how to find the path (Please refer to Figure \ref{figure: case ii} for the following proofs).
\begin{figure}
\centering
\subfigure[Case (i).]{
\includegraphics{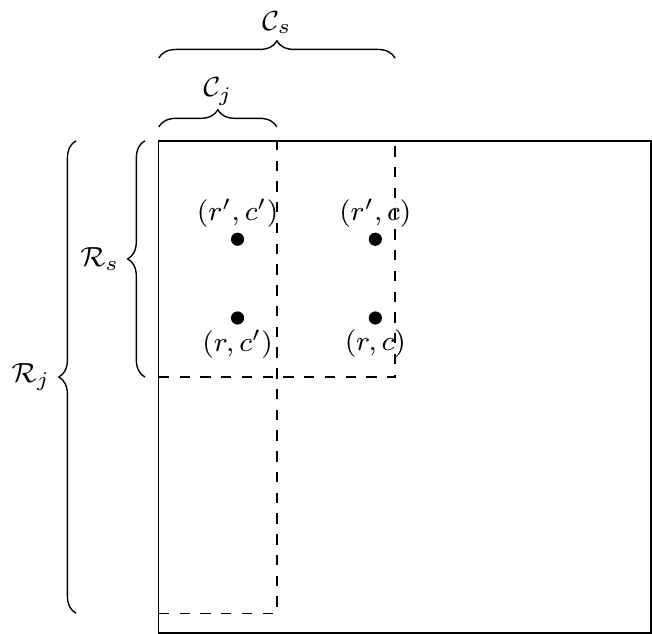}
\label{figure: case i}
}
\subfigure[Case (ii).]{
\includegraphics{caseii.pdf}
\label{figure: case ii}
}
\caption{Finding an augmentation path $(d=1)$, where the whole region is the $\cO(d/\epsilon)\times \cO(d/\epsilon)$ submatrix uniformly sampled from the original $n\times n$ matrix.}
\end{figure}

For index $(r,c)$ such that $\A_{r,c}\ne 0$, if $r'\not\in \cR_s$ or $c'\not\in \cC_s$ (say $r'\not\in\cR_s$ at the moment), then by Lemma \ref{lemma: uniform sampling captures augment pattern}, there exists an index $(r',c')\in (\cR_j,\cC_j)$ such that $(r,c)$ can be augmented by $(r',c')$. However, we cannot observe $\A_{r',c}$ so we do not find a $2\times 2$ submatrix at the moment. To make progress, we further discuss two cases.

\medskip
\noindent{\textbf{Case (ii.1). $\A_{r,c'}=0$ and $\A_{r',c'}=0$.}}

This case is impossible; otherwise, $(r',c')$ cannot be an augment of $(r,c)$.

\medskip
\noindent{\textbf{Case (ii.2). $\A_{r,c'}=0$ and $\A_{r',c'}\neq 0$.}}

Since $\A_{r,c}\ne 0$, $\A_{r',c'}\ne 0$ and $\A_{r,c'}=0$, no matter what $\A_{r,c'}$ is, the $2\times 2$ submatrix
\begin{equation*}
\begin{bmatrix}
\A_{r,c'} & \A_{r,c}\\
\A_{r',c'} & \A_{r',c}
\end{bmatrix}
=
\begin{bmatrix}
0 & \ne 0\\
\ne 0 & ?
\end{bmatrix}
\end{equation*}
is always non-singular (Denote by $?$ the entry which can be observed or unobserved, meaning that the specific value of the entry is unimportant for our purpose). So the algorithm answers correctly with probability at least $1-4e^{-c_0/4}>2/3$.

\medskip
\noindent{\textbf{Case (ii.3). $\A_{r,c'}\neq 0$.}}
Instead of augmenting $(r,c)$, we shall pick $(r, c')$ to be our new base entry ($1\times 1$ matrix) and try to augment it to a $2\times 2$ matrix. In this way, we have moved our base $1\times 1$ matrix towards the upper-left corner. We can repeat the preceding arguments of different cases.

If Case (i) happens for $(r,c')$, we immediately have a $2\times 2$ rank-$2$ submatrix and the algorithm answers correctly with a good probability. If Case (i) does not happen, we shall demonstrate that we can make further progress. Suppose that $(r'',c'')$ is an augment of $(\{r\},\{c'\})$ and $c''\not\in \cC_s\cup \cC_j$. We intend to look at the submatrix
\[
\begin{bmatrix}
\A_{r'',c'} & \A_{r'',c''}\\
\A_{r,c'} & \A_{r,c''}
\end{bmatrix}
\]
Here we cannot observe $\A_{r,c''}$. We know that $\A_{r'',c'}$ and $\A_{r'',c''}$ cannot be both $0$, otherwise $(r'',c'')$ would not be an augment for $(r,c')$. If $\A_{r'',c'} = 0$ and $\A_{r'',c''}\neq 0$, this $2\times 2$ matrix is nonsingular regardless of the value of $\A_{r,c''}$ and the algorithm will answer correctly. If $\A_{r'',c'} \neq 0$, we can rebase our $1\times 1$ base matrix to be $(r'',c')$ and try to augment it. Since $(r'',c')$ is above $(r',c)$, we have again moved towards the upper-left corner.

Note that there are at most $\log(1/\eta)$ different augment patterns and each time we rebase, $\A_{*,*}$ moves from one $(\cR_t,\cC_t)$ to another for some $t$. Hence, after repeating the argument above at most $2\log (1/\eta)$ times, the algorithm is guaranteed to observe a $2\times 2$ non-singular submatrix. Since the failure probability in each round is at most $4e^{-c_0/4}$, by union bound over $2\log (1/\eta)$ rounds, the overall failure probability is at most $8\log(1/\eta)e^{-c_0/4} \leq 1/3$, provided that $c_0=\cO(\log\log(\frac{1}{\eta}))$.

In summary, the overall probability is at least $2/3$ that the algorithm answers correctly in all cases by finding a submatrix of rank $2$, when $\A$ is $\epsilon$-far from being rank-$1$.
\end{proof}

\subsubsection{Extension to General Rank $d$}

\begin{theorem}
\label{theorem: correctness of algorithm for general d}
Let $\epsilon\le 1/e$ and $d\ge 1$. For any matrix $\A$, the probability that Algorithm \ref{algorithm: non-adaptive robust testing of rank} fails is at most $1/\mathsf{poly}(d\log (\frac{1}{\epsilon}))$.
\end{theorem}

\begin{proof}
If $\A$ is of rank at most $d$, then the algorithm will never make mistake, so we assume that $\A$ is $\epsilon$-far from being rank $d$ in the proof below.

The idea is that, we start with the base case of an empty matrix, and augment it to a full-rank $r\times r$ matrix in $r$ rounds, where in each round we increase the dimension of the matrix by exactly one. Each round may contain several steps in which we move the intermediate $j\times j$ matrix ($j\le r$) towards the upper-left corner without augmenting it; here, moving the matrix towards the upper-left corner means changing $\A_{\cR,\cC}$ to $\A_{\cR',\cC'}$, of the same rank, with $|\cR'|=|\cR|=|\cC'|=|\cC|=j$ and $\cR'\preceq \cR$ and $\cC'\preceq \cC$, where $\cR'\preceq \cR$ means that, suppose that $r_1' < r_2' < \cdots < r_j'$ are the (sorted) elements in $\cR'$ and $r_1 < r_2 < \cdots < r_j$ are the (sorted) elements in $\cR$, it holds that $r_i'\leq r_i$ for all $1\leq i\leq j$, and $\cC'\preceq \cC$ has a similar meaning.

The challenge is that those unobserved entries $?$'s may propagate as we augment the submatrix in each round. Our goal is to prove that starting from a \emph{structural} $(r-1)\times (r-1)$ full-rank submatrix which might have $?$'s as its entries, no matter what values of \emph{all} $?$'s are, with the augment operator we either (1) make progress for $(r-1)\times (r-1)$ submatrix, or (2) obtain an $r\times r$ full-rank submatrix \emph{with the same structure}. Let us first condition on the event that Lemma~\ref{lemma: uniform sampling captures augment pattern} holds true. Regarding the structure, we have the following claim.

\begin{claim}
\label{claim: structure for question mark}
There exists a searching path for $r\times r$ full-rank submatrices with non-decreasing $r$ which has the following lower triangular form modulo an elementary transformation
\begin{equation}
\label{equ: form}
\begin{bmatrix}
0 & 0 & \cdots & 0 & \cdots & 0 & \ne 0\\
0 & 0 & \cdots & 0 & \cdots & \ne 0 & ?\\
\vdots & \vdots &  & \vdots &  & \vdots & \vdots\\
0 & 0 & \cdots & \ne 0 & \cdots & ? & ?\\
\vdots & \vdots &  & \vdots &  & \vdots & \vdots\\
0 & \ne 0 & \cdots & ? & \cdots & ? & ?\\
\ne 0 & ? & \cdots & ? & \cdots & ? & ?\\
\end{bmatrix},
\end{equation}
where $\ne 0$ denotes the known entry which is non-zero, and $?$ denotes an entry which can be either observed or unobserved.
\end{claim}

\begin{proof}[Proof of Claim \ref{claim: structure for question mark}]
Without loss of generality, we assume that all $?$'s are unobserved, which is the most challenging case; otherwise, the proof degenerates to the discussion of \emph{central submatrix} in Case (iii) which we shall specify later.
We prove the claim by induction. The base case $r=1$ is true by Theorem \ref{theorem: correctness of algorithm for d=1}. Suppose the claims holds for $r-1$. We now argue the correctness for $r$.

Let $(p,q)$ be the augment. Denote the augment row by
\begin{equation*}
\begin{bmatrix}
y_1 & \cdots & y_b & \A_{p,q} & y_{b+2} & \cdots & y_{r}
\end{bmatrix},
\end{equation*}
and the augment column by
\begin{equation*}
\begin{bmatrix}
x_1 & \cdots & x_a & \A_{p,q} & x_{a+2} & \cdots & x_{r}
\end{bmatrix}^\top.
\end{equation*}
We now discuss three cases based on the relation between $a+b$ and $r$.

\medskip
\noindent{\textbf{Case (i). $a+b=r-1$ ($\A_{p,q}$ is on the antidiagonal of $r\times r$ submatrix)}.}

In this case, $y_{b+2},\dots,y_r$ and $x_{a+2},\dots,x_r$ are all $?$'s. We argue that $x_1=x_2=\cdots=x_a=0$ and $y_1=y_2=\cdots=y_b=0$; otherwise, we can make progress. First consider $y_i$ for $1\leq i\leq b$. If some $y_i\ne 0$, we can delete the $(r-i)$-th row in the $(r-1)\times (r-1)$ submatrix and insert the augment row (without the augment entry $\A_{p,q}$), which is above the deleted row.
Thus we obtain a new $(r-1)\times (r-1)$ submatrix towards the upper-left corner, and furthermore, the new submatrix exhibits the structure in \eqref{equ: form}. The same argument applies for $x_1,x_2,...,x_a$. Therefore, if no progress is made, it must hold that $x_1=x_2=...=x_a=0$ and $y_1=y_2=...=y_b=0$. In this case, $\A_{p,q}\ne 0$; otherwise, $(p,q)$ is not an augment. Therefore we obtain an $r\times r$ full-rank matrix of the form \eqref{equ: form}.

\medskip
\noindent{\textbf{Case (ii). $a+b< r-1$ ($\A_{p,q}$ is above the antidiagonal of $r\times r$ submatrix)}.}

In this case, $y_{r-a+1},\dots,y_r$ and $x_{r-b+1},\dots,x_r$ are all $?$'s. Similarly to Case (i), we shall argue that $x_1= \cdots =x_a=x_{a+2}= \cdots =x_{r-b}=0$ and $y_1=\cdots =y_b=y_{b+2}=\cdots=y_{r-a}=0$; otherwise, we can make progress. To see this, consider first $y_i$ for $1\leq i\leq b$ and then for $b+2\leq i\leq r-a$. If $y_i\ne 0$ for some $i\leq b$, we can delete the $(r-i)$-th row in the $(r-1)\times (r-1)$ submatrix  and insert the augment row (without the augment entry $\A_{p,q}$), which is above the deleted row, and so we make progress. Now assume that $y_1=\cdots=y_b=0$. If $y_i\ne 0$ for some $i$ such that $b+2\leq i\leq r-a$, we can delete the $(r-i+1)$-st row in the $(r-1)\times (r-1)$ submatrix of the last step and insert the augment row (without the augment entry $\A_{p,q}$), which is above the deleted row. So we make progress towards the most upper left corner. The same argument applies to $x_1,\dots,x_a,x_{a+2},\dots,x_{r-b}$. Therefore, $x_1=\cdots=x_a=x_{a+2}=...=x_{r-b}=0$ and $y_1=\cdots=y_b=y_{b+2}=\cdots=y_{r-a}=0$. In this case, $\A_{p,q}\ne 0$; otherwise, $(p,q)$ is not an augment since all possible choices of $?$'s cannot make the $r\times r$ submatrix non-singular. By exchanging the $(a+1)$-st row and the $(r-b)$-th row of the $r\times r$ submatrix or exchanging the $(b+1)$-st column and the $(r-a)$-th column, we obtain an $r\times r$ submatrix of the form \eqref{equ: form}.

\medskip
\noindent{\textbf{Case (iii). $a+b> r-1$ ($\A_{p,q}$ is below the antidiagonal of $r\times r$ submatrix)}.}

In this case, we argue that $x_i=y_j=0$ for all $i\leq r-b-1$ and $j\leq r-a-1$; otherwise we can make progress as Cases (i) and (ii) for $y_j$. To see this, let us discuss from $j=1$ to $r-a-1$. If $y_j\ne 0$ ($j=1,2,\dots,r-a-1$), we can delete the $(r-j)$-th row in the $(r-1)\times (r-1)$ submatrix and insert the augment row (without the augment entry $\A_{p,q}$), which is above the deleted row. So we make progress. The same argument applies to $x_1,\dots,x_{r-b-1}$. So $x_i=y_j=0$ for all $i\leq r-b-1$ and $j\leq r-a-1$.

Given that there is only one non-zero entry in the first $r-b-1$ rows and the first $r-a-1$ columns of the $r\times r$ submatrix (i.e., the Laplace expansion of the determinant), we only need to focus on a minor corresponding to a $\min\{a,b\}\times \min\{a,b\}$ \emph{central submatrix}, which decides whether the determinant of the $r\times r$ submatrix is zero and is fully-observed because the augment $(p,q)$ is at the lower right corner of the central submatrix (see the red part in Eqn.~\eqref{equ: central matrix}). Since it is fully-observed, the minor must be non-zero; otherwise, $(p,q)$ cannot be an augment for all choices of $?$'s. Therefore, we can do an elementary transformation to make the central submatrix a lower triangular matrix with non-zero antidiagonal entries. More importantly, such an elementary transformation also transforms the $r\times r$ matrix to a lower triangular matrix with non-zero antidiagonal entries, because all the entries to the left and above of the central matrix are $0$'s, and all the entries to the right and below of the central matrix are $?$'s. Hence any elementary transformation keeps $0$'s and $?$'s unchanged, and we obtain therefore an $r\times r$ submatrix of the form \eqref{equ: form}.

\begin{equation}
\label{equ: central matrix}
\begin{bmatrix}
0 & 0 & \cdots & \cdots & 0 & \cdots & 0 & \cdots & 0 & \ne 0\\
0 & 0 & \cdots & \cdots & 0 & \cdots & 0 & \cdots & \ne 0 & ?\\
\vdots & \vdots & & & \vdots & & \vdots & & \vdots & \vdots\\
\vdots & \vdots & & & \vdots & & \vdots & & \vdots & \vdots\\
0 & 0 & \cdots & \cdots & \red{\ne 0} & \red{\cdots} & \red{\textsf{known}} & \cdots & ? & ?\\
\vdots & \vdots & & & \red{\vdots} & & \red{\vdots} & & \vdots & \vdots\\
0 & 0 & \cdots & \cdots & \red{\textsf{known}} & \red{\cdots} & \red{\textsf{augment } \A_{p,q}} & \cdots & ? & ?\\
\vdots & \vdots & & & \vdots & & \vdots & & \vdots & \vdots\\
0 & \ne 0 & \cdots & \cdots & ? & \cdots & ? & \cdots & ? & ?\\
\ne 0 & ? & \cdots & \cdots & ? & \cdots & ? & \cdots & ? & ?
\end{bmatrix}.
\end{equation}
\end{proof}

Now we are ready to prove Theorem \ref{theorem: correctness of algorithm for general d}. Note that Lemma \ref{lemma: uniform sampling captures augment pattern} works only for \emph{fixed} $(\cR,\cC)$. To make the lemma applicable ``for all'' $(\cR,\cC)$ throughout the augmentation process, we shall take a union bound by choosing $|\cR|$ and $|\cC|$ large enough. Specifically, for each $i$, we divide $\cR_i=\bigcup_{k=1}^\ell \cR_i^{(k)}$ uniformly at random into $\ell = d+d\log(\frac{1}{\eta})$\footnote{In the number of parts $d+d\log(\frac{1}{\eta})$, the first term follows from the operation of augmenting $1\times 1$ submatrix to $d\times d$. The second term follows from moving the submatrix towards the upper left corner (from the lower-right corner in the worst case).} even parts $\cR_i^{(1)}, \cR_i^{(2)}, \dots, \cR_i^{(d)}$, where each $|\cR_i^{(k)}|=c[\log(d)+\log\log(\frac{1}{\eta})]2^i$, and divide $\cC_i=\bigcup_{k=1}^d \cC_i^{(k)}$ uniformly at random into $\ell$ even parts $\cC_i^{(1)}, \cC_i^{(2)}, \dots, \cC_i^{(\ell)}$, where each $|\cC_i^{(k)}|=c[\log(d)+\log\log(\frac{1}{\eta})]/(2^i\eta)$ for every $k$. We note that $\{\cR_i^{(k)}\}_k$ (and $\{\cC_i^{(k)}\}_k$) are independent of each other. It follows that
the event in Lemma \ref{lemma: uniform sampling captures augment pattern} holds
with probability at least $1-\frac{1}{\operatorname{poly}(d\log(1/\eta))}$. By a union bound over all $\ell^2 = \Theta(d^2\log^2(\frac{1}{\eta}))$ possible choices of $\{\cR_i^{(k)}\}\times \{\cC_i^{(k)}\}$ and Claim \ref{claim: structure for question mark}, with probability at least $1-1/\mathsf{poly}(d\log(\frac{1}{\epsilon}))$, Algorithm \ref{algorithm: non-adaptive robust testing of rank} answers correctly, when $\A$ is $\epsilon$-far from having rank $d$.
\end{proof}

\subsubsection{A Computationally Efficient Algorithm}

We now show how to implement Algorithm \ref{algorithm: non-adaptive robust testing of rank} efficiently, for which we only need to give a polynomial-time algorithm to solve the minimization problem \eqref{equ: upper bound} in Algorithm \ref{algorithm: a computationally efficient algorithm}. We have the following theorem.

\begin{algorithm}[ht]
\caption{Solving problem \eqref{equ: upper bound} in the polynomial time}
\label{algorithm: a computationally efficient algorithm}
\begin{algorithmic}[1]
\INPUT $\cR_1,\dots,\cR_k$ and $\cC_1,\dots,\cC_k$. Denote by $\cR^{(i)}$ the set of sampled indices in the $i$-th column of $\A$, and let $\cR^{(i)}_\perp = \cR_k \setminus \cR^{(i)}$.
\OUTPUT the solution $r$ to the minimization problem \eqref{equ: upper bound}.
\State $\S\leftarrow [\ ]$ is an empty matrix.
\State $r\leftarrow 0$.
\For{$i=1,\dots,|\cC_1|$}
	\If {there exists $\x$ such that $\S_{\cR^{(i)},:}\x = \A_{\cR^{(i)},i}$}
		\State $\A_{\cR_\perp^{(i)},i}\leftarrow \S_{\cR_\perp^{(i)},:}\x$.
	\Else
		\State $\A_{\cR_\perp^{(i)},i}\leftarrow \1$.
		\State $\S\leftarrow [\S,\A_{\cR_k,i}]$.
		\State $r\leftarrow r+1$.
	\EndIf
\EndFor
\State \Return $r$.
\end{algorithmic}
\end{algorithm}

\begin{theorem}
\label{theorem: time complexity of rank testing}
Algorithm \ref{algorithm: a computationally efficient algorithm} correctly solves the minimization problem \eqref{equ: upper bound} in $\mathsf{poly}(\frac{d}{\epsilon})$ time.
\end{theorem}

\begin{proof}
Without loss of generality, we may permute the rows and columns of $\A$ and assume that $\cR_i = \{1,\dots,|\cR_i|\}$ and $\cC_i = \{1,\dots,|\cC_i|\}$ for all $i\leq \lceil\log\frac{1}{\eta}\rceil$. Our goal is to complete the submatrix $\A_{(\cR_m\times \cC_1)}$ such that its rank is minimal. Denote by $\cR^{(i)}$ the set of sampled indices in the $i$-th column of $\A$. We start from an empty matrix $\S=[\ ]$. We will extend $\S$ as we process the columns of $\A$ that are not in the block $(\cR_k,\cC_k)$ from left to right. We will maintain the following two invariants:

\begin{itemize}
\item
The minimal rank of matrix completion is always equal to the number of columns of $\S$.
\item
%Denote by $\A_{\cR^{(i)},i}$ the restriction of the $i$-th column to its first $|\cR^{(i)}|$ coordinates.
After processing the $i$-th column of $\A$, the restricted column $\A_{\cR^{(i)},i}$ is in the column space of $\S_{\cR^{(i)},:}$.
\end{itemize}

Note that both invariants hold in the base case. For the $i$-th column $\A_{:,i}$ that we encounter, if $\A_{\cR^{(i)},i}$ is in the column space of $\S_{\cR^{(i)},:}$, then we use a linear combination on the first $|\cR^{(i)}|$ coordinates of vectors given by the columns of $\S_{\cR^{(i)},:}$ to extend $\A_{\cR^{(i)},i}$ from a vector in $|\cR^{(i)}|$ dimensions to a vector in $|\cR_m|$ dimensions, and we do not change $\S$. Notice that the two invariants are preserved in this case.

Otherwise, $\A_{\cR^{(i)},i}$ is not in the column space of $\S_{\cR^{(i)},:}$. If $\A_{\cR^{(i)},i}$ were in the column space of $\A_{\cR^{(i)},1:{(i-1)}}$, we would, by the second invariant above, have that $\A_{\cR^{(i)},i}$ is in the column space of $\S_{\cR^{(i)},:}$, a contradiction. Therefore, $\A_{\cR^{(i)},i}$ is not in the column space of $\A_{\cR^{(i)},[i-1]}$. In this case, $\A_{:,i}$ must be linearly independent of all previous columns $\A_{:,[i-1]}$. We can thus append to $\S$ on the right the vector $[\A_{\cR^{(i)},i};\1]$ (The vector $\1$ can be replaced with any $(|\cR_k|-|\cR^{(i)}|)$-dimensional vector), which increases the size and rank of $\S$ by $1$, and we maintain our two invariants.

The time complexity of Algorithm \ref{algorithm: a computationally efficient algorithm} is $\mathsf{poly}(\frac{d}{\epsilon})$.
\end{proof}

\subsection{Lower Bounds over Finite Fields in the Sampling Model}

According to Yao's minimax principle, it suffices to provide a distribution on $n\times n$ input matrices $\A$ for which any deterministic testing algorithm fails with significant probability over the choice of $\A$. Before proceeding, we first state a hardness result that we want to reduce from.

\begin{algorithm}[ht]
\caption{Decomposing edges $E$}
\label{algorithm: decompose E}
\begin{algorithmic}[1]
\INPUT A bipartite graph $G=(L\cup R,E)$.
\OUTPUT Partition of $E=E_1\cup \cdots \cup E_t$ and the set of pivot nodes $\{w_t\}$.
\State $t\leftarrow 0$.
\While {$E\ne\emptyset$}
	\State Find $v$ such that $1\le \mathsf{deg}(v)\le\gamma d$.
	\State $t\rightarrow t+1$.
	\State $E_t\leftarrow$ edges between $v$ and all its neighbours.
	\State $w_t\leftarrow v$.
	\State $E\leftarrow E \setminus E_t$.
\EndWhile
\State \Return $E=E_1\cup\cdots\cup E_t$ and $\{w_t\}$.
\end{algorithmic}
\end{algorithm}

\begin{lemma}
\label{lemma: partition of graph}
Let $G=(L\cup R,E)$ be a bipartite graph such that $|L|=|R|=n$ and $|E|<\gamma^2 d^2$ for $d\leq n/\gamma$. Then Algorithm \ref{algorithm: decompose E} returns a partition $E=E_1\cup E_2\cup\cdots\cup E_t$, where $t\le \gamma^2 d^2$ and $|E_i|\le \gamma d$ for all $i$.
\end{lemma}

\begin{proof}
We first show that Algorithm \ref{algorithm: decompose E} can be executed correctly, that is, whenever $E\ne \emptyset$ there always exists $v$ such that $1\le \mathsf{deg}(v)\le\gamma d$. We note that $1\le \mathsf{deg}(v)$ is obvious because $E\ne \emptyset$. If all vertices with non-zero degree have degree at least $\gamma d$, the total number of edges would be at least $\gamma d n\geq d^2 \gamma^2$, contradicting our assumption on the size of $E$. When the algorithm terminates, it is clear that each $E_i$ generates at most $\gamma d$ edges and the $E_i$'s are disjoint and so $t\le \gamma^2 d^2$.
\end{proof}

\begin{lemma}
\label{Lemma: independence}
Suppose that there are $t$ groups of (fixed) vectors $\{\v_1^{(k)},\dots,\v_{s_k}^{(k)}\}_{k\in [t]}\subset\mathbb{F}^d$ such that the vectors in each group are linearly independent (denoted by $\perp$). Let $\w_1,\dots,\w_r$ be random vectors in $\bbF^d$ such that each $\w_i$ is chosen uniformly at random from some set $\cS_i\subseteq \bbF^d$ with $|\cS_i|\geq |\bbF|^{(1-\gamma)d}$. Let $s = \max_k {s_k}$. When $s+r\le \gamma d$ for all $k$ and $t\le \gamma^2 d^2$, it holds that
\begin{equation*}
\Pr_{\w_1,\dots,\w_r} \left\{\v_1^{(k)}\perp \dots \perp\v_{s_k}^{(k)}\perp\w_1\perp \cdots \perp\w_r \text{ for all } k\in [t]\right \}\ge 1-\frac{\gamma^3 d^3}{|\bbF|^{(1-2\gamma)d}}.
\end{equation*}
\end{lemma}

\begin{proof}
For fixed $\w_1,\dots,\w_{i-1}$ such that $\v_1^{(k)},\dots,\v_{s_k}^{(k)},\w_1,\dots,\w_{i-1}$ are linearly independent for all $k\in[t]$, the probability that $\w_i\in \cS_i$ is linearly independent of $\v_1^{(k)},\dots,\v_{s_k}^{(k)},\w_1,\dots,\w_{i-1}$ for all $k\in[t]$ is at least
$
1-\frac{t|\bbF|^{s_k+i-1}}{|\cS_i|} \geq 1-\frac{t|\bbF|^{s_k+i-1}}{|\bbF|^{(1-\gamma)d}} = 1-\frac{t}{|\bbF|^{(1-\gamma)d-(s_k+i-1)}} \geq 1 - \frac{t}{|\bbF|^{(1-\gamma)d-(s+i-1)}}.
$
Therefore, for all $k\in[t]$ with $t\le\gamma^2 d^2$, we have
\begin{align*}
& \quad\ \Pr_{\w_1,\dots,\w_r}\{\v_1^{(k)}\perp\cdots\perp\v_{s_k}^{(k)}\perp\w_1\perp\cdots\perp\w_r\text{ for all }k\}\\
&=  \prod_{i=2}^{r} \Pr_{\w_i}\{\v_1^{(k)}\perp\cdots\perp\v_{s_k}^{(k)}\perp \w_1 \perp \cdots \perp \w_i \ \text{ for all }k \mid \v_1^{(k)}\perp\cdots\perp\v_{s_k}^{(k)}\perp \w_1 \perp \cdots \perp \w_{i-1}\text{ for all }k\}
\\
&\qquad \times \Pr_{\w_1}\{\v_1^{(k)}\perp \cdots \perp\v_{s_k}^{(k)}\perp\w_1\text{ for all }k\}\\
&\ge \prod_{i=1}^r \left(1-\frac{t}{|\bbF|^{(1-\gamma)d-(s+i-1)}}\right)\\
&\ge \prod_{i=1}^r \left(1-\frac{\gamma^2 d^2}{|\bbF|^{(1-2\gamma)d}}\right)\quad\text{(by $s+i-1\le\gamma d$ and $t\le \gamma^2 d^2$)}\\
&\ge 1- \frac{r\gamma^2 d^2}{|\bbF|^{(1-2\gamma)d}}\quad\text{(by $(1-x)^t\ge 1-tx$ for $x\in(0,1)$)}\\
&\ge 1-\frac{\gamma^3 d^3}{|\bbF|^{(1-2\gamma)d}}.\quad\text{(since $r\le\gamma d$)}. \qedhere
\end{align*}
\end{proof}

When $|\cS_i| = |\bbF|^{d-d_i}$ for $d_i\leq \gamma d$, it follows from Lemma \ref{Lemma: independence} that the number of choices of the event
\begin{equation}
\label{equ: fraction}
\begin{split}
&\left|\left\{(\w_1,\dots,\w_r)\in \cS_1\times\cdots\times \cS_r: \v_1^{(k)}\perp \cdots \perp\v_{s_k}^{(k)}\perp\w_1 \perp \cdots \perp \w_r \text{ for all } k\right\}\right|\\
=&\Pr_{\w_1,\dots,\w_r}\left\{\v_1^{(k)}\perp \cdots \perp\v_{s_k}^{(k)}\perp\w_1 \perp \cdots \perp \w_r \text{ for all } k\right\}\cdot \prod_{i=1}^r|\cS_i|\\
\ge & \left(1-\frac{\gamma^3 d^3}{|\bbF|^{(1-2\gamma)d}}\right)\cdot \prod_{i=1}^r |\cS_i|\quad\text{(by Lemma \ref{Lemma: independence})}\\
= & \left(1-\frac{\gamma^3 d^3}{|\bbF|^{(1-2\gamma)d}}\right)|\bbF|^{rd-\sum_{i=1}^r d_i}.\quad\text{(recall that $|\cS_i|=|\bbF|^{d-d_i}$)}
\end{split}
\end{equation}

Based on this result, we have the following lemma.
\begin{algorithm}[ht]
\caption{Path for assigning subspace $H_v$ and random vector $\x_v$ to each node $v$}
\label{algorithm: path for assigning subspace and random vector for each node}
\begin{algorithmic}[1]
\INPUT Bipartite graph $G=(L\cup R,E)$, partition $E=E_1\cup\cdots \cup E_t$ and pivot nodes $\{w_t\}$ by Algorithm \ref{algorithm: decompose E}, observed entries $\x|_{E}$.
\OUTPUT An affine space $H_v$ of vectors for every node $v$ and a vector $\x_v\in H_v$ for every node $v$.
\State $H_v\gets \bbF^d$ for all $v$.
\State Set all nodes $v$ unassigned.
\For {$i\leftarrow t$ \textbf{down to} $1$}
	\State Let $v_1^{(i)},\dots,v_{|E_i|}^{(i)}$ be the non-pivot nodes in $E_i$ (i.e., the edges in $E_i$ are $(w_i, v_j^{(i)})$).
	\For {$j\leftarrow 1$\textbf{ to }$|E_i|$}
%		\State $r_i\leftarrow 0$
		\If{$v_j^{(i)}$ is unassigned} \label{alg:line:v_unassigned}
%			\State $r_i\leftarrow r_i+1$
			\State $W_{v_j^{(i)}}\hspace{-0.1cm} \gets\hspace{-0.1cm} H_{v_j^{(i)}}\setminus\bigcup_{k\leq i: v_{j}^{(i)}\neq w_k} \hspace{-0.2cm} \operatorname{span}\{\x_{v_1}^{(i)}, \dots, \x_{v_{j-1}}^{(i)}, \text{previously assigned non-pivot nodes in }E_k\}$.
			\State Choose $w_{v_j^{(i)}}$ uniformly at random from $H_{v_j^{(i)}}$. \label{alg:line:nonpivot_vector_assignment}
			\If {$w_{v_j^{(i)}}\not\in W_{v_j^{(i)}}$}
				\State \textbf{abort}.
			\EndIf
			\State Set $v_j^{(i)}$ to be assigned.
		\EndIf
	\EndFor
	\State Let $H_{w_i}$ be the solution set to the linear system (w.r.t. $\x_{w_i}$): $\x_{w_i}^\top [\x_{v_1}^{(i)},\cdots ,\x_{v_{|E_i|}}^{(i)}]=(\x|_{E_i})^\top$.
\EndFor
\State Choose $\x_{w_s}$ uniformly from $H_{w_s}$ of dimension $d-|E_s|$ for all $s\in\cS_0=\{p\in[t]\mid w_p\text{ is
 unassigned}\}$\label{alg:line:pivot_vector_assignment}.
\State \Return $\{H_v\}$ and $\{\x_v\}$.
\end{algorithmic}
\end{algorithm}

\begin{lemma}
Let $\U,\V\sim \mathcal{U}_\bbF(n,d)$, where $\cU_{\mathbb{F}}(m,n)$ represents $m\times n$ i.i.d. uniform matrix over a finite field $\mathbb{F}$. Denote by $\cS$ any subset of $[n]\times [n]$ such that $|\cS|<\gamma^2 d^2$ for $\gamma\in (0,1/4)$ and $d\leq n/\gamma$. It holds that for any $\x\in\mathbb{F}^{|\cS|}$,
\begin{equation*}
\Pr[(\U\V^T)|_\cS=\x]-\frac{1}{|\bbF|^{|\cS|}}\ge -\frac{\gamma^5 d^5}{|\bbF|^{(1-2\gamma)d+|\cS|}}.
\end{equation*}
\end{lemma}

\begin{proof}
Consider a bipartite graph $G=(L\cup R,E)$ where $|L|=|R|=n$ and $(i,j)\in E$ if and only if $(i,j)\in\cS$. We run Algorithm \ref{algorithm: decompose E} on graph $G$. By Lemma \ref{lemma: partition of graph}, we obtain a sequence of edge sets $E_1\dots,E_t$ with $w_1,\dots,w_t$ (called \emph{pivot nodes}), such that
\begin{itemize}
\item[1.]
$\{E_i,\dots,E_t\}$ forms a partition of $E$;
\item[2.]
$|E_i|\le \gamma d$ for all $i$.
\end{itemize}
Since there is a one-by-one correspondence between the edges and the entries in $\cS$, we will not distinguish edges and entries in the rest of the proof.

We associate each node $v$ of $G$ with an affine space $H_v\subseteq \bbF^d$ and a random vector $\x_v\in H_v$ as in Algorithm \ref{algorithm: path for assigning subspace and random vector for each node}. Basically, Algorithm \ref{algorithm: path for assigning subspace and random vector for each node} first assigns the non-pivot nodes (to determine the affine subspace $H_{w_i}$) from the $E_t$ down to the $E_1$), and in the end assigns all unassigned pivot nodes.

In the following argument, we number the for-loop iterations in Algorithm \ref{algorithm: path for assigning subspace and random vector for each node} backwards, i.e., the for-loop starts with the $t$-th iteration and goes down to the first iteration. In the $i$-th iteration, let $r_i$ denote the number of nodes $v_j^{(i)}$ that are unassigned at the runtime of Line~\ref{alg:line:v_unassigned} and let $\#\cE_i$ denote the number of good choices (which do not trigger abortion) of Step~\ref{alg:line:nonpivot_vector_assignment}  over all $r_i$ nodes to be assigned.
%, given the running of the $(i+1)$-st, ..., $t$-th iterations, such that the assignments in the $i$-th iteration which are also non-pivot node of $E_k$ are linearly independent of all previously assigned nodes as the non-pivot nodes of $E_k$ for all $k<i$.
Let $\#\cG$ be the number of possible choices of Step~\ref{alg:line:pivot_vector_assignment} of Algorithm \ref{algorithm: path for assigning subspace and random vector for each node} and $s_0=|\cS_0|$ be the number of assigned pivot nodes by Step~\ref{alg:line:pivot_vector_assignment}. Note that by the construction of Algorithm \ref{algorithm: decompose E}, the non-pivot nodes of $E_i$ cannot be the pivot nodes of $E_j$ for $j<i$. So Algorithm \ref{algorithm: path for assigning subspace and random vector for each node}, if terminated successfully, can find an assignment such that $(\U\V^\top)|_\cS=\x$. We now lower bound the success probability.

Let $d_j^{(i)}=d-\mathsf{dim}(H_{v_j^{(i)}})$, which is either $0$ or $|E_k|$ for some $k>i$. For any given realization $\x$, we have the following:

\begin{align*}
&\quad\ \Pr\{(\U \V^\top)|_\cS=\x\}\\
&\ge \frac{\#\cE_t\cdot \#\cE_{t-1} \cdots \#\cE_1}{|\bbF|^{d (r_t + \cdots + r_1)}}\frac{\#\cG}{|\bbF|^{d s_0}}\quad\mbox{(by rule of product and definition of $\#\cE_i$)}\\
&\ge \prod_{i=1}^{t} \frac{1}{|\bbF|^{d_1^{(i)}+ \cdots +d_{r_i}^{(i)}}}\left(1-\frac{\gamma^3 d^3}{|\bbF|^{(1-2\gamma)d}}\right)^t\cdot \frac{\#\cG}{|\bbF|^{d s_0}}\quad \mbox{(by Eqn. \eqref{equ: fraction})}\\
&\ge \frac{1}{|\bbF|^{\sum_{i=1}^t\sum_{j=1}^{r_i} d_j^{(i)}}}\left(1-\frac{\gamma^5 d^5}{|\bbF|^{(1-2\gamma)d}}\right)\cdot \frac{\#\cG}{|\bbF|^{d\times s_0}}\quad\mbox{(by $(1-x)^t\ge 1-tx$ for $x\in(0,1)$ and $t\le \gamma^2 d^2$)}\\
&=\frac{1}{|\bbF|^{\sum_{i=1}^t\sum_{j=1}^{r_i} d_j^{(i)}}}\left(1-\frac{\gamma^5 d^5}{|\bbF|^{(1-2\gamma)d}}\right)\cdot \frac{1}{|\bbF|^{\sum_{s\in\cS_0} |E_s|}}\quad\text{(by definition of $\#\cG$)}\\
&\ge \frac{1}{|\bbF|^{|E_1|+\cdots+|E_t|}}\left(1-\frac{\gamma^5 d^5}{|\bbF|^{(1-2\gamma)d}}\right),
\end{align*}
where the last inequality holds because $|E_1|+\dots+|E_t|=\sum_{j=1}^t\sum_{i=1}^{r_j} d_i^{(j)}+\sum_{s\in\cS_0}|E_s|$ as every pivot and non-pivot node must be assigned exactly once by Algorithm \ref{algorithm: path for assigning subspace and random vector for each node} upon successful termination. (Recall that $d_i^{(j)}$ is either equal to $0$ when $v_j^{(i)}$ is non-pivotal, or equal to $|E_k|$ when $v_j^{(i)}=w_k$.)
\end{proof}

Denote by $\cS\subset[n]\times [n]$ a set of indices of an $n\times n$ matrix. For any distribution $\cL$ over $\bbF^{n\times n}$, define $\cL(\cS)$ on $\bbF^{|\cS|}$ as the marginal distribution of $\cL$ on the entries of $\cS$, namely,
\begin{equation*}
(\X_{p_1,q_1},\X_{p_2,q_2},\dots,\X_{p_{|\cS|},q_{|\cS|}})\sim\cL(\cS),\qquad \X\sim\cL.
\end{equation*}
Now we are ready to show a lower bound of robust testing problem over any finite field.
\begin{theorem}
\label{theorem: total variation distance of sampled matrices for rank testing lower bound over finite field}
Suppose that $\bbF$ is a finite field and $\gamma\in(0,1/4)$ is an absolute constant. Let $\U,\V\sim\cU_{\bbF}(n,d)$ and $\W\sim\cU_{\bbF}(n,n)$, where $\cU_{\mathbb{F}}(m,n)$ represents $m\times n$ i.i.d. uniform matrix over a finite field $\mathbb{F}$. Consider two distributions $\cL_1$ and $\cL_2$ over $\bbF^{n\times n}$ defined by $\U\V^\top$ and $\W$, respectively. Let $\cS\subset [n]\times [n]$. When $|\cS|<\gamma^2 d^2$, it holds that
\begin{equation*}
d_{TV}(\cL_1(\cS),\cL_2(\cS))\le Cd^5|\bbF|^{-cd},
\end{equation*}
where $C,c>0$ are constants depending on $\gamma$, and $d_{TV}(\cdot,\cdot)$ represents the total variation distance between two distributions.
\end{theorem}

\begin{proof}
Let
\begin{equation*}
\cX=\left\{\x\in\mathbb{F}^{|\cS|} \left|\ \Pr\left[(\U\V^\top)|_{\cS}=\x\right]<\frac{1}{|\bbF|^{|\cS|}}\right.\right\}.
\end{equation*}
It follows from the definition of total variation distance that
\begin{equation*}
d_{TV}(\cL_1(\cS),\cL_2(\cS))=\sum_{\x\in\cX}\left[\frac{1}{|\bbF|^{|\cS|}}-\Pr[(\U\V^\top)|_{\cS}=\x]\right]\le \sum_{\x\in\cX}\frac{\gamma^5 d^5}{|\bbF|^{(1-2\gamma)d}}\frac{1}{|\bbF|^{|\cS|}}\le \frac{\gamma^5 d^5}{|\bbF|^{(1-2\gamma)d}},
\end{equation*}
where the last inequality holds since $|\cX|\le |\bbF|^{|\cS|}$.
\end{proof}

Based on the above theorem, we have the following lower bound for the rank testing problem over finite field.
\begin{theorem}
\label{theorem: lower bound of rank testing over finite fields under sampling model}
Let $d\le \sqrt{\epsilon}n$. Any non-adaptive algorithm for Problem \ref{problem: problem of property testing of rank} over any finite field $\bbF$ requires $\Omega(d^2/\epsilon)$ queries.
\end{theorem}
\begin{proof}
We first show that for constant $\epsilon$, any non-adaptive algorithm for Problem \ref{problem: problem of property testing of rank} over finite field $\bbF$ requires $\Omega(d^2)$ queries. Note that $\W\sim\cU_{\bbF}(n,n)$ is $\epsilon$-far from having rank less than $d$. It follows immediately from the preceding theorem that any algorithm which solves the matrix rank testing problem over a finite field must read $\Omega(d^2)$ entries; otherwise when $d$ is large enough, it will hold that $d_{TV}(\cL_1(\cS),\cL_2(\cS))<1/4$, contradicting the correctness of the algorithm on distinguishing $\cL_1$ from $\cL_2$.

We now prove the case for arbitrary $\epsilon$. Denote by $\A$ and $\B$ the two hard instances in Theorem \ref{theorem: total variation distance of sampled matrices for rank testing lower bound over finite field}. We construct two hard instances $\C$ and $\D$ by uniformly at random planting the above-mentioned hard instances $\A$ and $\B$ of dimension $\sqrt{\epsilon}n\times \sqrt{\epsilon}n$, respectively, and padding zeros everywhere else. Note that $\D$ being $\epsilon$-far from rank $d$ is equivalent to $\B$ being constant-far from rank $d$. Suppose that we can request $cd^2/\epsilon$ queries with a small absolute constant $c$ to distinguish the ranks of the hard instances $\C$ and $\D$, then in expectation (and with high probability by a Markov bound) we can request $cd^2$ queries of the hard instances $\A$ and $\B$ to distinguish their ranks, which leads to a contradiction.
\end{proof}

\subsection{Lower Bounds over Real Field under the Sampling Model}
The \emph{rigidity} of a matrix $\A$ over a field $\F$, denoted by $\cR_{\A}^{\F}(r)$, is the least number of entries of $\A$ that must be changed in order to reduce the rank of $\A$ to a value at most $r$:
$
\cR_{\A}^{\F}(r):=\min\{\|\C\|_0\mid\rank_{\F}(\A+\C)\le r\}.
$
We first cite the following lemma and theorem.

\begin{lemma}[Matrix Rigidity, Theorem 6.4, \cite{valiant1977graph}]
\label{lemma: matrix rigidity over reals}
The real $n\times n$ i.i.d. Gaussian matrix $\G$ is of rigidity $\cR_{\G}^{\R}(r)=\Omega((n-r)^2)$ with probability $1$.
\end{lemma}

\begin{theorem}[Theorem 3.5, \cite{li2014sketching}]
\label{theorem: hardness of distinguishing Gaussian noise}
Let $\U,\V\sim \cG(n,d)$ and $\G\sim \cG(n,n)$. Consider two distributions $\cL_1$ and $\cL_2$ over $\R^{n\times n}$ defined by $\U\V^\top$ and $\U\V^\top+n^{-14}\G$, respectively. Let $\cS\subset [n]\times [n]$. Whenever $|\cS|\le d^2$, it holds that
\begin{equation*}
d_{TV}(\cL_1(\cS),\cL_2(\cS))\le C|S|(n^{-2}+dc^d),
\end{equation*}
where $C>0$ and $0<c<1$ are absolute constants.
\end{theorem}

Now we are ready to prove the sample complexity lower bound of rank testing over the reals in the sampling model.
\begin{theorem}
\label{theorem: lower bound of rank testing over reals}
Let $d\le \sqrt{\epsilon}n$. Any non-adaptive algorithm for Problem \ref{problem: problem of property testing of rank} over $\R$ requires $\Omega(d^2/\epsilon)$ queries.
\end{theorem}

\begin{proof}
We first show that for constant $\epsilon$, any non-adaptive algorithm for Problem \ref{problem: problem of property testing of rank} over $\R$ requires $\Omega(d^2)$ queries. Note that Theorem \ref{theorem: hardness of distinguishing Gaussian noise} provides two hard instances for distinguishing a rank-$d$ matrix (of the form $\A=\U\V^\top$) from a rank-$n$ matrix (of the form $\B=\U\V^\top+n^{-14}\G$), where $\U,\V\sim\cG(n,d)$ and $\G\sim\cG(n,n)$. For our purpose, we only need to show that the rank-$n$ matrix $\B=\U\V^\top+n^{-14}\G$ has rigidity $\cR_{\B}^{\R}(d)=\Omega(n^2)$.
Denote by $\rank_{\ell}(\B)=\min_{\|\S\|_0=\ell} \rank(\B+\S)$. We note that
\begin{equation*}
\begin{split}
d&\ge \rank_{\cR_{\B}^{\R}(d)}(\B)\\
&=\min_{\|\S\|_0=\cR_{\B}^{\R}(d)}\rank(\U\V^\top+n^{-14}\G+\S)\\
&\ge \min_{\|\S\|_0=\cR_{\B}^{\R}(d)}\rank(n^{-14}\G+\S)-\rank(\U\V^\top)\\
&\ge \min_{\|\S\|_0=\cR_{\B}^{\R}(d)}\rank(n^{-14}\G+\S)-d.
\end{split}
\end{equation*}
Therefore, $\min_{\|\S\|_0=\cR_{\B}^{\R}(d)}\rank(n^{-14}\G+\S)\le 2d$, i.e., $\cR_{\B}^{\R}(d)\ge \cR_{n^{-14}\G}^{\R}(2d)$. By Lemma \ref{lemma: matrix rigidity over reals}, we have $\cR_{n^{-14}\G}^{\R}(2d)=\Omega(n-2d)^2=\Omega(n^2)$. So $\cR_{\B}^{\R}(d)=\Omega(n^2)$.

We now prove the case for arbitrary $\epsilon$. We construct two hard instances $\C$ and $\D$ by uniformly at random planting the above-mentioned hard instances $\A$ and $\B$ of dimension $\sqrt{\epsilon}n\times \sqrt{\epsilon}n$, respectively, and padding zeros everywhere else. Note that $\D$ being $\epsilon$-far from rank $d$ is equivalent to $\B$ being constant-far from rank $d$. Suppose that we can request $cd^2/\epsilon$ queries with a small absolute constant $c$ to distinguish the ranks of the hard instances $\C$ and $\D$, then in expectation (and with high probability by a Markov bound) we can request $cd^2$ queries of the hard instances $\A$ and $\B$ to distinguish their ranks, which leads to a contradiction.
\end{proof}

\subsection{Lower Bounds over Finite Fields in the Sensing Model}
\label{section: Lower Bounds over Finite Fields under the Sensing Model}

In this section, we provide a lower bound for the rank testing problem in the sensing model over any finite field $\mathbb{F}$. The sensing problem can query the underlying matrix $\A$ in the form of $\langle \A,\X_{i} \rangle$ for any sequence of (randomized or deterministic) sensing matrices $\{\X_i\}$. The algorithms for querying entries of $\A$ are a special case of matrix sensing problem if we set $\X_i=\e_p\e_q^\top$ for some $(p,q)$. The problem can be stated more formally as follows:
\begin{problem}[Rank Testing with Parameter $(n,d,\epsilon)$ in the Sensing Model]
\label{problem: problem of property testing of rank in the sensing model}
Given a field $\mathbb{F}$ and a matrix $\A\in\mathbb{F}^{n\times n}$ which has one of promised properties:
\begin{itemize}
\item[$\mathsf{H0.}$]
$\A$ has rank at most $d$;
\item[$\mathsf{H1.}$]
$\A$ is $\epsilon$-far from having rank at most $d$, meaning that $\A$ requires changing at least an $\epsilon$-fraction of its entries to have rank at most $d$.
\end{itemize}
The problem is to design a property testing algorithm that outputs $\mathsf{H0}$ with probability $1$ if $\A\in \mathsf{H0}$, and output $\mathsf{H1}$ with probability at least 0.99 if $\A\in\mathsf{H1}$, with the least number of queries of the form $\langle\A,\X_i\rangle$, where $\{\X_i\}$ is a sequence of sensing matrices.
\end{problem}

\begin{definition}[Ruzsa-Szemer$\acute{\text{e}}$di Graph]
A graph $G$ is an $(r,t)$-\emph{Ruzsa-Szemer$\acute{\text{e}}$di graph (RS graph for short)}, if and only if the set of edges of $G$ consists of $t$ pairwise disjoint induced matchings $M_1,\dots,M_t$, each of which is of size $r$.
\end{definition}

\begin{definition}[Boolean Hidden Hypermatching, $\BHH_{n,p}$]
The \emph{Boolean Hidden Hypermatching} problem is a one-way communication problem where Alice is given a boolean vector $\x\in\{0,1\}^n$ such that $n=2kp$ for some integer $k\ge 1$, and Bob is given a boolean vector $\w$ of length $n/p$ and a perfect $p$-hypermatching $\cM$ on $n$ vertices such that each hyperedge contains $p$ vertices. Denote by $\cM\x$ the length $n/p$ boolean vector $(\bigoplus_{1\le i\le p}\x_{\cM_{1,i}},\bigoplus_{1\le i\le p}\x_{\cM_{2,i}},\dots,\bigoplus_{1\le i\le p}\x_{\cM_{n/p,i}})$ where $\{\cM_{1,1},\dots,\cM_{1,p}\},\dots,\{\cM_{n/p,1},\dots,\cM_{n/p,p}\}$ are the hyperedges of $\cM$. It is promised that either $\cM\x=\w$ or $\cM\x=\overline\w$. The goal of the problem is for Bob to output \textup{\textbf{YES}} when $\cM\x=\w$ and \textup{\textbf{NO}} when $\cM\x=\overline\w$ ($\oplus$ stands for addition modulo 2).
\end{definition}

For our purpose, it is more convenient to focus on a special case of Boolean Hidden Hypermatching problem, namely, $\BHH_{n,p}^0$ where the vector $\w=\0^{n/p}$ ($p$ is an even integer) and Bob's task is to output \textbf{YES} if $\cM\x=\0^{n/p}$ and output \textbf{NO} if $\cM\x=\1^{n/p}$. It is known that we can reduce any instance of $\BHH_{n,p}$ to an instance of $\BHH_{2n,p}^0$ deterministically without any communication between Alice and Bob~\cite{bury2015sublinear,li2016approximating,verbin2011streaming}, by the following reduction.

\medskip
\noindent{\textbf{Reduction from $\BHH_{n,p}$ to $\BHH_{2n,p}^0$.}}
We reduce any instance of $\BHH_{n,p}$ to an instance of $\BHH_{2n,p}^0$ ($n=2kp$ for some integer $k$). Let $\cM$ be a perfect $p$-hypermatching and $\x\in\{0,1\}^n$ in $\BHH_{n,p}$. Denote by $\x'=[\x;\overline{\x}]$ the concatenation of $\x$ and $\overline\x$, where $\overline\x$ is the bitwise negation of $\x$. Let $\cM'$ be the $p$-hypermatching in $\BHH_{2n,p}^0$.
Denote by $\{\x_1,\dots,\x_p\}\in\cM$ the $l$-th hyperedge of $\cM$ ($l\in[n/p]$). We add two hyperedges to $\cM'$ as follows. If $\w_l=0$, we add $\{\x_1,\x_2,\dots,\x_p\}$ and $\{\overline{\x_1},\overline{\x_2},\dots,\overline{\x_p}\}$ to $\cM'$; Otherwise, we add $\{\overline{\x_1},\x_2,\dots,\x_p\}$ and $\{\x_1,\overline{\x_2},\dots,\overline{\x_p}\}$ to $\cM'$. Note that we flip an even number of bits when $\w_l=0$ and an odd number of bits when $\w_l=1$. This does not change the parity of each set as $p$ is even. Thus $\cM\x=\w$ implies $\cM'\x'=\0^{2n/p}$, and $\cM\x=\overline{\w}$ implies $\cM'\x'=\1^{2n/p}$.

\medskip
Previous papers
\cite{assadi2017estimating,bury2015sublinear,esfandiari2014streaming} used the $\BHH_{n,p}^0$ problem to prove lower bounds for estimating matching size in the data stream: given an instance $(\x,\cM)$ in $\BHH_{n,p}^0$ (Denote by $\cD_{\BHH}$ the hard distribution of $\BHH_{n,p}^0$), we create a graph $G(V\cup W,E)$ with $|V|=|W|=n$ via the following algorithm.

\begin{algorithm}[ht]
\caption{Reduction from $\BHH_{n,p}^0$ to the problem of estimating matching size in the data stream}
\label{algorithm: reduction from BHH to matching size}
\begin{algorithmic}[1]
\INPUT An instance from $\BHH_{n,p}^0$.
\OUTPUT A graph $G=(V\cup W,E)$.
\State For any $\x_i=1$, Alice adds an edge between $v_i$ and $w_i$ to $E$.
\State Bob adds to $E$ a clique between the vertices $w_i$ that belongs to the same hyperedge $\e$ in the $p$-hypermatching $\cM$.
\end{algorithmic}
\end{algorithm}

We shall use the graph created by Algorithm \ref{algorithm: reduction from BHH to matching size} to build a hard distribution for our \emph{rank testing} problem. The following claim guarantees the correctness of this reduction from $\BHH_{n,p}^0$ to the problem of estimating matrix rank in a data stream.

\begin{lemma}
\label{lemma: reduction from BHH to rank testing}
Let $G(V\cup W,E)$ be the graph derived from an instance $(\x,\cM)$ of $\BHH_{n,p}^0$ (for even integers $p$ and $n$) with the property that $\|\x\|_0=n/2$ (see Algorithm \ref{algorithm: reduction from BHH to matching size}). Denote by $\A$ the $2n\times 2n$ adjacency matrix of $G$. Then with probability at least $1-e^{-n/p^4}$, we have
\begin{itemize}
\item
if $\cM\x=\0^{n/p}$ (i.e., \textup{$\mathsf{YES}$} case), then $\rank(\A)\ge \frac{3n}{2}-\frac{n}{2p^2}$;
\item
if $\cM\x=\1^{n/p}$ (i.e., \textup{$\mathsf{NO}$} case), then $\rank(\A)\le \frac{3n}{2}-\frac{3n}{2p^2}$.
\end{itemize}
\end{lemma}

\begin{proof}
According to Algorithm \ref{algorithm: reduction from BHH to matching size}, the graph consists of $n$ vertices $v_1,\dots,v_n$ and $n/p$ cliques, together with edges which connect $v_i$'s with the cliques according to $\x\in\{0,1\}^n$. We call these latter edges `tentacles'.

Let $\A$ be the adjacency matrix of $G$ where both the rows and columns are indexed by the nodes in $G$. The diagonals of $\A$ are all zeros. For each pair $w$, $u$ of clique nodes in $G$, we have $\A_{w,u}=1$. For each `tentacle' pair $(v,w)$, $\A_{v,w}=1$. All other entries of $\A$ are zeros. Then $\A$ is an $n\times n$ block diagonal matrix, where each block $\A_{q_i}$ ($\A_{q_i}$ represents the block with $q_i$ `tentacles') is of the following form modulo permutations of rows and columns (The red rows and columns represent `tentacles'):
\begin{equation*}
\A_{q_i}=
\begin{bmatrix}
0 & 1 & \cdots & 1 & 1 & \red{1} & \red{0} & \red{\cdots}\\
1 & 0 & \cdots & 1 & 1 & \red{0} & \red{1} & \red{\cdots}\\
\vdots & \vdots & & \vdots & \vdots & \red{\vdots} & \red{\vdots}\\
1 & 1 & \cdots & 0 & 1 & \red{0} & \red{0} & \red{\cdots}\\
1 & 1 & \cdots & 1 & 0 & \red{0} & \red{0} & \red{\cdots}\\
\red{1} & \red{0} & \red{\cdots} & \red{0} & \red{0} & \red{0} & \red{0} & \red{\cdots}\\
\red{0} & \red{1} & \red{\cdots} & \red{0} & \red{0} & \red{0} & \red{0} & \red{\cdots}\\
\red{\vdots} & \red{\vdots} &  & \red{\vdots} & \red{\vdots} & \red{\vdots} & \red{\vdots}\\
\end{bmatrix}
\end{equation*}

According to the reduction from $\BHH_{n/2,p}$ to $\BHH_{n,p}^0$, the hypermatching in the hard distribution of $\BHH_{n,p}^0$ can be divided into $n/(2p)$ groups. Each group consists of two hyperedges such that the sum of the number of `tentacles' connecting to these two hyperedges is $p$ for every group, i.e., $(q_i,p-q_i)$ where $q_i$ is the number of `tentacles' connecting to one of hyperedges, which is either even ($\mathsf{YES}$ case) or odd ($\mathsf{NO}$ case) according to the promise. Moreover, the $q_i$'s are independent across the $n/p$ groups, because we can process each group one by one and after processing each group, the number of remaining `tentacles' decreases by $p$.

Let $r_{q_i}=\rank(\A_{q_i})$. Denote by $A=\mathbb{E}_{\mathsf{YES}}(r_{q_i}+r_{p-q_i})$ and $B=\mathbb{E}_{\mathsf{NO}}(r_{q_i}+r_{p-q_i})$, where $A$ and $B$ will be calculated later. Summing up $n/(2p)$ \emph{independent} groups and by the Chernoff bound, with probability at least $1-e^{-\delta^2\frac{n}{2p}A/2}$ and $1-e^{-\delta^2\frac{n}{2p}B/3}$, respectively, $\rank(\A)\ge (1-\delta)\frac{n}{2p}A$ in the even case and $\rank(\A)\le (1+\delta)\frac{n}{2p}B$ in the odd case, where $\delta>0$ is an absolute constant. We note that $A=3p$ and $B=3p-4/p$. Therefore, $\rank(\A)\ge (1-\delta)3n/2$ in the even case and $\rank(\A)\le (1+\delta)(3n/2-2n/p^2)$ in the odd case. Choosing $\delta=\frac{1}{3p^2}$ finishes the proof.
\end{proof}

In the following we shall set $\epsilon=\Theta(1/\log n)$ and $p=\Theta(\log n)$. Denote by $\mathsf{Matching}_{n,k,\epsilon}$ the \emph{$k$-player simultaneous communication} problem of estimating the size of maximum matching up to a factor of $(1\pm\epsilon)$, where the edges of an $n$-vertex input graph are partitioned across the $k$ players and the referee. For our purpose, we reduce from the problem of $\mathsf{Matching}_{n,k,\epsilon}$ to our problem of \emph{rank testing}. We use the hard distribution $\cD_{\mathsf{M}}$ in Algorithm \ref{algorithm: hard distribution for matching} for $\mathsf{Matching}_{n,k,\epsilon}$. Notice that the hard instance of $\BHH_{r,p}^0$ in Step \ref{alg: BHH^0} is reduced from that of $\BHH_{r/2,p}$ as we did before in this section.

\begin{algorithm}[ht]
\caption{A construction of a hard distribution $\cD_{\mathsf{M}}$ for $\mathsf{Matching}_{n,k,\epsilon}$}
\label{algorithm: hard distribution for matching}
\begin{algorithmic}[1]
\INPUT $r=N^{1-o(1)}$, $t=\frac{\binom{N}{2}-o(N^2)}{r}$, $k=\frac{N}{\epsilon r}$, $n=N+2r(k-1)$, and $p=\lfloor\frac{1}{8\epsilon}\rfloor$.
\State Fix an $(r,t)$-RS graph $G^{\RS}$ on $N$ vertices.
\State Pick $j^*\in[t]$ uniformly at random and draw a $\BHH_{r,p}^0$ instance $(\x^{(j^*)},\cM)$ from the distribution $\cD_{\BHH}$. \label{alg: BHH^0}
\For{each player $P^{(i)}$ independently}
	\State (a) Let $G_i$ be the input graph of $P^{(i)}$, initialized by a copy of $G^{\RS}$ with vertices $V_i=[N]$.
	\State (b) Let $V_i^*$ be the set of vertices matched in the $j^*$-th induced matching of $G_i$. Change the induced matching $M_{J^*}^{\RS}$ of $G_i$ to $M_{j^*}:=M_{j^*}^{\RS}|_{\x^{(j^*)}}$.
	\State (c) For any $j\in[t]\backslash\{j^*\}$, draw a vector $\x^{(i,j)}\in\{0,1\}^r$ from the distribution $\cD_{\BHH}$ for $\BHH_{r,p}^0$, and change the induced matching $M_j^{\RS}$ of $G_i$ to $M_j:=M_j^{\RS}|_{\x^{(i,j)}}$.
	\State (d) Create the family of $p$-cliques of $\cM$ on the vertices $R(M_{j^*}^{\RS})$, and give the edges of the $p$-clique family to the referee.
\EndFor
\State Choose a random permutation $\sigma$ of $[n]$. For each player $P^{(i)}$, relabel $v$ to $\sigma(j)$ for each vertex $v$ in $V_i\backslash V_i^*$ with label $j\in [N-2r]$. Enumerate the vertices in $V_i^*$ from the one with the smallest label to the one with the largest label, and relabel the $j$-th vertex to $\sigma(N+(i-2)2r+j)$. Finally, let the vertices with the same label correspond to the same vertex.
\end{algorithmic}
\end{algorithm}

\begin{claim}
\label{claim: BHH and rank}
Let $I_{\BHH}$ be the embedded $\BHH_{r,p}^0$ instance $(\x^{(i)},\cM)$ in Algorithm \ref{algorithm: hard distribution for matching}. The adjacency matrix $\A\in\mathbb{F}^{n\times n}$ of the graph that is drawn from distribution $\cD_{\mathsf{M}}$ (Algorithm \ref{algorithm: hard distribution for matching}) obeys
\begin{itemize}
\item[1.]
If $I_{\BHH}$ is a $\mathsf{YES}$ instance, then $\rank(\A)\ge k(\frac{3r}{2}-\frac{r}{2p^2})$;
\item[2.]
If $I_{\BHH}$ is a $\mathsf{NO}$ instance, then $\rank(\A)\le k(\frac{3r}{2}-\frac{3r}{2p^2})+N-2r$,
\end{itemize}
with probability at least $1-ke^{-n/p^4}$.
\end{claim}
\begin{proof}
Note that by construction, the adjacency matrix of the graph drawn from $\cD_{\mathsf{M}}$ is a $k$-block-diagonal matrix together with some `junk' (area of size $(N-2r)\times n$ union $n\times (N-2r)$) outside the block area such that each block is an independent sample of the matrix $\A$ in Lemma \ref{lemma: reduction from BHH to rank testing}. The claim then is a straightforward result of Lemma \ref{lemma: reduction from BHH to rank testing}.
\end{proof}

\medskip
\noindent{\textbf{Reduction from $\mathsf{Matching}_{n,k,\epsilon}$ to Problem \ref{problem: problem of property testing of rank in the sensing model}.}} Given a hard graph instance $G$ of $\mathsf{Matching}_{n,k,\epsilon}$, we can estimate the maximum matching size of $G$ by testing the rank of the adjacency matrix $\A_G$ of $G$: If we can distinguish out $\rank(\A_G)\ge k(\frac{3r}{2}-\frac{r}{2p^2})$, we ouput that the matching size is strictly larger than $\frac{3N}{\epsilon}$; If we can distinguish out $\rank(\A_G)\le k(\frac{3r}{2}-\frac{3r}{2p^2})+N-2r$, we output that the matching size is smaller than $\frac{3N}{\epsilon}-3N$. The correctness for the reduction follows from Claim \ref{claim: BHH and rank}, the construction that the hard distributions of $\mathsf{Matching}_{n,k,\epsilon}$ and Problem \ref{problem: problem of property testing of rank in the sensing model} are derived from the same graph, and the fact that the matching size is strictly larger than $\frac{3N}{\epsilon}$ when $I_\BHH$ is a $\YES$ instance and is smaller than $\frac{3N}{\epsilon}-3N$ when $I_\BHH$ is a $\NO$ instance (see Claim 6.3, \cite{assadi2017estimating}).

\medskip
The hardness of $\mathsf{Matching}_{n,k,\epsilon}$ by the construction in Algorithm \ref{algorithm: hard distribution for matching} was proved in \cite{assadi2017estimating}.
\begin{theorem}[Theorem 10, \cite{assadi2017estimating}]
\label{theorem: hardness of maximum matching}
For any sufficiently large $n$ and sufficiently small $\epsilon<\frac{1}{2}$, there exists some $k=n^{o(1)}$ such that the distribution $\cD_{\mathsf{M}}$ for $\mathsf{Matching}_{n,k,\epsilon}$ in Algorithm \ref{algorithm: hard distribution for matching} satisfies
\begin{equation*}
\mathsf{IC}_{\mathsf{SMP},\cD_{\mathsf{M}}}^\delta(\mathsf{Matching}_{n,k,\epsilon})=n^{2-\cO(\epsilon)},
\end{equation*}
where $\mathsf{IC}_{\mathsf{SMP},\cD_{\mathsf{M}}}^\delta(\mathsf{Matching}_{n,k,\epsilon})$ is the information complexity of $\mathsf{Matching_{n,k,\epsilon}}$ in the multi-party number-in-hand simultaneous message passing model (SMP).
\end{theorem}

The following theorem summarizes the results in this section, providing a lower bound for Problem \ref{problem: problem of property testing of rank in the sensing model}.

\begin{theorem}
\label{theorem: lower bound over finite fields in general basis}
Any non-adaptive algorithm for Problem \ref{problem: problem of property testing of rank in the sensing model} over $\mathsf{GF}(p)$ requires $\Omega(d^2/\log p)$ queries.
\end{theorem}

\begin{proof}
We first discuss the case when $d=\Omega(n)$, where we will give an $\Omega(n^2)$ lower bound. Let $\A_G$ be the hard instance given by Algorithm \ref{algorithm: hard distribution for matching}. We want to find an $n\times n$ random matrix $\H'$ such that: (1) $\rank(\H'\A)=\rank(\A)$ (Multiplying $\H'$ does not change the rank of $\A$ so that testing $\A$ is equivalent to testing $\H'\A$); (2) $\M=\H'\A$ is rigid (Multiplying $\H'$ makes matrix $\A$ rigid). We now show how to do this. Let $\B$ be a random matrix such that we want to distinguish rank $n$ v.s. rank $n-n/\log^2 n$ for matrix $\A:=\A_G+\B$. Let $k = \rank(\A)$, $\H$ be a $3nk/\delta\times n$ uniformly sampled matrix over $\mathsf{GF}(p)^{3nk/\delta\times n}$ and $\H'$ be the first $n$ rows of $\H$. One can see that any subset of at most $n$ rows of $\H$ has full rank with a large probability.

\medskip
\noindent{\textbf{Proof of (1).}} We note that $\rank(\H'\A)\le k$. We will show that $\H'\A$ has rank $k$ with probability at least $1-\delta$. We will use the following lemma.
\begin{lemma}[Lemma 5.3, \cite{clarkson2009numerical}]
If $\cL\subseteq \mathsf{GF}(p)^n$ is a $j$-dimensional linear subspace, and $\A$ has rank $k\ge j$, then the dimension of $\cL_{\A}:=\{\w\in\mathsf{GF}(p)^n\mid \w^\top\A\in\cL\}$ is at most $n-k+j$.
\end{lemma}
For $j<n$, consider the linear subspace $\cL_j$ spanned by the first $j$ rows of $\H\A$. By the above lemma, the dimension of the subspace $\cL_j':=\{\w\in\R^n\mid \w^\top\A\in\cL_j\}$ is at most $n-k+j$. Given that the rows of $\H$ are linearly independent with high probability, at most $n-k+j$ of them can be in $\cL_j'$. Thus the probability that $\H_{(j+1),:}'\A$ is not in $\cL_j$ is at least $1-(n-k+j)/(3nk/\delta-j)$, and the probability that all such events hold, for $j=0,\dots,k-1$, is at least
\begin{equation*}
\left(1-\frac{n}{3nk/\delta-k}\right)^k= \left(1-\frac{1}{k}\frac{\delta/3}{1-\delta/(3n)}\right)^k\ge 1-\frac{\delta}{2}
\end{equation*}
for small $\delta$. All such independence events occur if and only if $\H'_{1:k,:}\A$ has rank $k$. Therefore, the probability that $\H'\A$ is of rank $k$ is at least $1-\delta/2$.

\medskip
\noindent{\textbf{Proof of (2).}} We need the following result on matrix rigidity.

\begin{lemma}[Matrix Rigidity, Theorem 6.4, \cite{valiant1977graph}]
\label{lemma: matrix rigidity over GF(p)}
The fraction of matrices over $\mathsf{GF}(p)^{n\times n}$ with matrix rigidity $\cR^{\mathsf{GF}(p)}(r)=\Omega((n-r)^2/\log_p n)$ is at least $0.99$, for $r<n-\sqrt{2n\log_p 2+\log n}$.
\end{lemma}

For uniform matrix $\H'$, we note that $\H'\A$ is uniform as well: for any given matrix $\T$ in $\mathsf{GF}(p)^{n\times n}$
\begin{equation*}
\Pr_{\H'\sim\mathsf{Unif}}[\H'\A=\T]=\Pr_{\H'\sim\mathsf{Unif}}[\H'=\T\A^{-1}]=\left(\frac{1}{p}\right)^{kn}.
\end{equation*}
Then by Lemma \ref{lemma: matrix rigidity over GF(p)}, $\cR_{\H'\A}^{\mathsf{GF}(p)}(n-n/\log n)=\Omega(n^2/\log^2 n)$ with high probability.

Now we are ready to prove the hardness of Problem \ref{problem: problem of property testing of rank in the sensing model} with parameter $(n,n-\frac{n}{\log^2(n)},\frac{1}{\log^4 (n)\log_p(n)})$. For any non-adaptive algorithm $\cA_{test}$ for Problem \ref{problem: problem of property testing of rank in the sensing model} with $\epsilon=1/\log n$ and $d=n-n/\log n$, assume that the required number of queries is $q$. We use such algorithms to estimate the maximum matching size by our reduction. Given a graph $G$ with maximum matching size $\ge 3N/\epsilon$ v.s. $\le 3N/\epsilon-3N$. We know that the rank of $\A:=\A_G+\B$ is of rank $n$ v.s. $n-n/\log^2 n$. By left multiplying matrix $\A$ with above-mentioned $\H'$, the rank of resulting matrix $\H'\A$ remains the same and is of rigidity $\cR_{\B}^{\mathsf{GF}(p)}(n-n/\log^2 n)=\Omega(n^2/(\log^4 (n)\log_p (n)))$ according to properties (1) and (2) that we have proven. By assumption, $\cA_{test}$ can distinguish rank $n$ from rank $n-n/\log^2 n$ for matrix $\H'\A$ in $q$ queries with high probability. So $\cA_{test}$ can be used to compute the maximum matching size with $(1\pm 1/\log n)$-approximation rate with $\cO(q\log p)$ bits of communication. By Theorem \ref{theorem: hardness of maximum matching}, we have $q\log p=\Omega(n^2)$, which implies that $q=\Omega(n^2/\log p)$.

We now prove the lower bound for arbitrary $d$. Let $\1\in\mathsf{GF}(p)^{\frac{n(1-1/\log d)}{d}\times \frac{n(1-1/\log d)}{d}}$ be the all-ones matrix and $\A\in\mathsf{GF}(p)^{\frac{d}{1-1/\log d}\times \frac{d}{1-1/\log d}}$ be the above hard instance. We do the Kronecker product to generate matrix $\C=\1\otimes\A\in\mathsf{GF}(p)^{n\times n}$. If there exists a non-adaptive algorithm $\cA_{test}$ that can correctly test whether $\C$ has rank at most $d$ or is far from having rank $d$ with $cd^2/\log p$ queries and high probability for an absolute constant $c$, the algorithm $\cA_{test}$ can also test whether $\A$ has rank at most $d$ or is far from having rank $d$ with $cd^2/\log p$ queries by outputting the same result as testing $\C$. This leads to a contradiction.
\end{proof}

Theorem \ref{theorem: lower bound over finite fields in general basis} is tight up to a logarithmic factor. Indeed, there is an $\cO(d^2)$ upper bound for every field, independent of $\epsilon$, as follows. If $\A$ is an (unknown) $n\times n$ matrix and has rank at least $d+1$, the matrix $\S\A\T$ will have rank at least $d+1$ with high probability for random $\S$ of $d+1$ rows and $\T$ of $d+1$ columns; furthermore, this matrix product can be computed in the matrix sensing model because $(\S\A\T)_{i,j}$ can be written as ${\langle \A, \S_{i,:}\T_{:,j}\rangle}_{i,j}$, which is in the form of matrix sensing. Computing $\S\A\T$ uses only $(d+1)^2$ measurements instead of the $d^2/\epsilon$ we need for reading entries.

\section{Non-Adaptive Stable Rank Testing}

In this section and onwards, we study the problem of non-adaptively testing numerical properties of real-valued matrices. They can be studied under a unified framework in this section.

Roughly, our analytical framework reduces the testing problem to a sequence of estimation problems \emph{without involving $\mathsf{poly}(n)$ in the sample complexity}. Our framework consists of two levels of estimation: (1) a constant-factor approximation to the statistic $X$ of interest (e.g., stable rank), and (2) a more accurate $(1\pm\tau)$-approximation to $X$.

\begin{definition}[Stable Rank]
The stable rank of $\A$ is defined by $\srank(\A)=\|\A\|_F^2/\|\A\|^2$, where $\|\A\|_F$ is the Frobenius norm and $\|\A\|$ the spectral norm (largest singular value).
\end{definition}

\begin{problem}[Stable Rank Testing in the Entry Model]
Let $\A\in\mathbb{R}^{n\times n}$ be a matrix which satisfies $\|\A\|_\infty\leq 1$ and has one of promised properties:
\begin{itemize}
\item[$\mathsf{H0.}$]
$\A$ has stable rank at most $d$;
\item[$\mathsf{H1.}$]
$\A$ is $\epsilon/d$-far from having stable rank at most $d$, meaning that $\A$ requires changing at least an $\epsilon/d$-fraction of its entries to have stable rank at most $d$.
\end{itemize}
The problem is to design a property testing algorithm that outputs $\mathsf{H0}$ with probability at least $0.99$ if $\A\in \mathsf{H0}$, and output $\mathsf{H1}$ with probability at least 0.99 if $\A\in\mathsf{H1}$, with the least number of queried entries.
\end{problem}

\subsection{Upper Bounds}

\comment{
\begin{lemma}
\label{lemma: f and 2 norm when far from stable rank d}
If matrix $\A$ is $\epsilon/d$-far from being stable rank at most $d$ and $\|\A\|_\infty=1$, then
\begin{equation*}
\|\A\|_F^2> \left(\frac{\epsilon n}{d}-1\right)(d-1),
\end{equation*}
and
\begin{equation*}
\|\A\|^2< \left[1+\left(1-\frac{1}{d}\right)\frac{1}{n}\left(\frac{\epsilon n}{d}-1\right)\right]\frac{\|\A\|_F^2}{d}+\left(\frac{n}{d}-1\right)\left(\frac{\epsilon n}{d}-1\right).
\end{equation*}
\end{lemma}

\begin{proof}
Suppose that $\A$ is $\epsilon/d$-far from being stable rank at most $d$. Let $\x\in \mathbb{S}^{n-1}$ be the unit vector such that $\|\A\|=\|\A\x\|_2$, i.e., $\x$ is the right singular vector corresponding to the largest singular value. Without loss of generality, we assume that $\langle\A_{1,:},\x\rangle^2\le \langle\A_{2,:},\x\rangle^2\le...\le \langle\A_{n,:},\x\rangle^2$. Let $m=\lceil\epsilon n/d\rceil-1$. Replacing each $A_{i,:}$, $i\in [m]$ with vector $\x^\top/\|\x\|_\infty$ forms a new matrix $\B$. Since $\A$ is $\epsilon/d$-far from being stable rank at most $d$ and $\|\A\|_\infty=1$, it must hold that $\srank(\B)>d$ and $\|\B\|_\infty=1$.

By our construction of $\B$, we have $\|\B\|^2\ge \frac{m}{\|\x\|_\infty^2}$ and $\|\B\|_F^2\le \|\A\|_F^2+\frac{m}{\|\x\|_\infty^2}$. Therefore,
\begin{equation*}
d<\srank(\B)=\frac{\|\B\|_F^2}{\|\B\|^2}\le \frac{\|\A\|_F^2+\frac{m}{\|\x\|_\infty^2}}{\frac{m}{\|\x\|_\infty^2}},
\end{equation*}
which implies that
\begin{equation*}
\|\A\|_F^2> \frac{m(d-1)}{\|\x\|_\infty^2}\ge \left(\frac{\epsilon n}{d}-1\right)(d-1).
\end{equation*}
Here the last inequality holds because $\|\x\|_\infty\le 1$.

We now prove the second conclusion. Notice that
\begin{equation*}
S_m:=\sum_{i=1}^m\langle\A_{i,:},\x\rangle^2 \le \frac{m}{n} \sum_{i=1}^n\langle\A_{i,:},\x\rangle^2=\frac{m}{n}\|\A\|^2< \frac{m}{n}\frac{\|\A\|_F^2}{d},
\end{equation*}
where the last inequality holds because $\srank(\A)=\frac{\|\A\|_F^2}{\|\A\|^2}>d$ by assumption. We also observe that
\begin{equation*}
\|\B\|_F^2\le \|\A\|_F^2-\sum_{i=1}^m\|\A_{i,:}\|_2^2+\frac{m}{\|\x\|_\infty^2}\le \|\A\|_F^2-S_m+\frac{m}{\|\x\|_\infty^2},
\end{equation*}
where the last inequality holds since $S_m=\sum_{i=1}^m\langle\A_{i,:},\x\rangle^2\le \sum_{i=1}^m \|\A_{i,:}\|_2^2$ by the Cauchy-Schwartz inequality.
Because
\begin{equation*}
\|\B\|^2\ge \|\B\x\|_2^2= \sum_{i=m+1}^n \langle\A_{i,:},\x\rangle^2+m\left\langle\frac{\x}{\|\x\|_\infty},\x\right\rangle^2=\|\A\|^2-S_m+\frac{m}{\|\x\|_\infty^2},
\end{equation*}
it follows that
\begin{equation*}
d<\srank(\B)=\frac{\|\B\|_F^2}{\|\B\|^2}\le \frac{\|\A\|_F^2-S_m+\frac{m}{\|\x\|_\infty^2}}{\|\A\|^2-S_m+\frac{m}{\|\x\|_\infty^2}}.
\end{equation*}
Thus
\begin{equation*}
\|\A\|^2< \left[1+\left(1-\frac{1}{d}\right)\frac{1}{n}\left(\frac{\epsilon n}{d}-1\right)\right]\frac{\|\A\|_F^2}{d}+\left(\frac{n}{d}-1\right)\left(\frac{\epsilon n}{d}-1\right).
\end{equation*}
\end{proof}
}

\comment{
\begin{lemma}
\label{lemma: concentration}
Let $\widetilde\A\in\R^{n\times n}$ be the matrix which uniformly samples $q$ rows of $\A$. Denote by $X=\frac{n}{q}\|\widetilde\A\|_F^2$. Then it holds that $|X-\|\A\|_F^2|\le \frac{1}{8}(1-\frac{1}{d})^2\epsilon n^2$ with probability at least $0.9$.
\end{lemma}

\begin{proof}
By the Chernoff bound, sampling $q$ rows uniformly gives
\begin{equation*}
\Pr\left[\left|X-\|\A\|_F^2\right|\le \frac{1}{8}\left(1-\frac{1}{d}\right)^2\epsilon n^2\right]\ge \Pr\left[\left|X-\|\A\|_F^2\right|\le \frac{1}{8}\left(1-\frac{1}{d}\right)^2\epsilon n^2\right].
\end{equation*}
\end{proof}
}

\begin{lemma}[{\cite[Theorem 1.8]{rudelson2007sampling}}]
\label{Lemma: Vershynin}
Let $\A$ be an $n\times n$ matrix. Let $\cQ$ be a uniformly random subset of $\{1,2,\dots,n\}$ of expected cardinality $q$ with replacement.
Then
\begin{equation*}
\bbE\left\|\A|_{\cQ}\right\|\lesssim \sqrt{\frac{q}{n}}\|\A\|+\sqrt{\log q}\|\A\|_{(n/q)},
\end{equation*}
where $\A|_{\cQ}=(\A_{i,j})_{i\in\cQ,j\le n}$ is a random row-submatrix of $\A$, and $\|\A\|_{(n/q)}$ is the average of $n/q$ biggest Euclidean lengths of the columns of $\A$.
\end{lemma}

\comment{
\begin{proof}
Denote by $\delta_1,\delta_2,\dots,\delta_n$ the $\{0,1\}$-valued independent random variables with $\bbE\delta_i=p_i$. Thus our random set is $Q=\{j\mid \delta_j=1\}$.

Denote by $\x_1,\x_2,\dots,\x_n$ the columns of $\A$. So
\begin{equation*}
\A=\sum_{j=1}^n \e_j\otimes \x_j,\qquad \A|_Q=\sum_{j=1}^n \delta_j \e_j\otimes\x_j.
\end{equation*}
The spectral norm of $\A$ is
\begin{equation*}
\|\A\|=\|\A\A^\top\|^{1/2}=\left\|\sum_{j=1}^n \x_j\otimes\x_j\right\|^{1/2},
\end{equation*}
and the spectral norm of $\A|_Q$ is
\begin{equation*}
\|\A|_Q\|=\left\|\sum_{j=1}^n\delta_j\x_j\otimes \x_j\right\|^{1/2}.
\end{equation*}
Let
\begin{equation*}
E=\bbE\|\A|_Q\|.
\end{equation*}
The symmetrization argument yields
\begin{equation*}
\begin{split}
E&=\bbE\left\|\sum_{j=1}^n(\delta_j-\bbE\delta_j+\bbE\delta_j)\x_j\otimes \x_j\right\|^{1/2}\\
&\le\mathbb{E}\left\|\sum_{j=1}^n(\delta_j-\bbE\delta_j)\x_j\otimes \x_j\right\|^{1/2}+\left\|\sum_{j=1}^n \bbE\delta_j\x_j\otimes \x_j\right\|^{1/2}\\
&\le \bbE_\delta\left(\bbE_\epsilon\left\|\sum_{j=1}^n\epsilon_j\delta_j\x_j\otimes\x_j\right\|\right)^{1/2}+\frac{}{}
\end{split}
\end{equation*}
\end{proof}
}

\begin{lemma}
\label{Lemma: infty norm of uniform sampling from unit sphere}
Let $\x\sim\mathsf{Unif}(\mathbb{S}^{n-1})$. Then with probability at least $1-n^{-2}$, we have $\|\x\|_\infty\le \sqrt{\frac{2\log n}{n}}$.
\end{lemma}

\comment{
\begin{theorem}
Suppose that $d/\epsilon\ge 2$. Then Algorithm \ref{algorithm: algorithm for the stable rank testing problem} is a correct algorithm for the stable rank testing problem with failure probability at most $\frac{1}{3}$ in the entry model. It reads $\widetilde\cO(\frac{d^6}{\epsilon^6})$ entries.
\end{theorem}

\begin{proof}
When $\A$ is far from being stable rank at most $d$, we claim that $\|\A\|_F^2\ge \epsilon n^2(1-\frac{1}{d})$. Otherwise, changing the first $\frac{\epsilon n}{d}$ rows of $\A$ to be all-one vectors $\1^\top$ leads to a new matrix $\B$ such that $\|\B\|^2\ge \frac{\epsilon n^2}{d}$ and $\|\B\|_F^2\le \epsilon n^2(1-\frac{1}{d})+\frac{\epsilon n^2}{d}=\epsilon n^2$, leading to $\srank(\B)\le d$, a contradiction. By the Chernoff bound, it holds that with probability at least $1-e^{-\tau^2\|\A\|_F^2/3}\ge 1-e^{-\tau^2\epsilon n^2(1-1/d)/3}$, by sampling $q$ rows from $\A$, the resulting matrix $\A_{\mathsf{row}}$ satisfies
\begin{equation}
\label{equ: F norm}
(1-\tau)\|\A\|_F^2\le \frac{n}{q}\|\A_{\mathsf{row}}\|_F^2\le (1+\tau)\|\A\|_F^2.
\end{equation}
Conditioning on this event, by sampling $q$ columns from $\A_{\mathsf{row}}$ and using Chernoff bound again, with probability at least $1-e^{-\tau^2\epsilon n^2(1-1/d)/3}-e^{-\tau^2\|\A_{\mathsf{row}}\|^2/3}\ge 1-2e^{-\tau^2\epsilon n^2(1-1/d)q(1-\tau)/(3n)}$ we have
\begin{equation}
\label{equ: F norm entry level}
(1-\tau)^2\|\A\|_F^2\le (1-\tau)\frac{n}{q}\|\A_{\mathsf{row}}\|_F^2\le X\le (1+\tau)\frac{n}{q}\|\A_{\mathsf{row}}\|_F^2\le (1+\tau)^2\|\A\|_F^2.
\end{equation}

\comment{
Conditioned on the event $|X-\|\A\|_F^2|\le \frac{1}{8}(1-\frac{1}{d})^2\epsilon n$ in Lemma \ref{lemma: concentration}, it holds that
\begin{equation*}
X\ge \|\A\|_F^2-\frac{1}{8}\left(1-\frac{1}{d}\right)^2\epsilon n\ge \frac{9}{10}\epsilon n^2\left(1-\frac{1}{d}\right).
\end{equation*}
Therefore, Algorithm \ref{algorithm: algorithm for the stable rank testing problem} is correct on Line 5.

To prove the correctness of Lines 8 and 10, we only need to discuss the case when $X\ge \frac{9}{10}\epsilon n^2\left(1-\frac{1}{d}\right)$. So
\begin{equation*}
\|\A\|_F^2\ge X-\frac{1}{8}\left(1-\frac{1}{d}\right)^2\epsilon n\ge \frac{8}{10}\epsilon n^2\left(1-\frac{1}{d}\right).
\end{equation*}
Let $\eta:=\frac{c\epsilon n}{\|\A\|_F^2}\le \frac{10cd}{8(d-1)}=:\eta'$. We have
\begin{equation}
\left(1-\frac{\tau\eta'}{c}\right)\|\A\|_F^2\le X\le \left(1+\frac{\tau\eta'}{c}\right)\|\A\|_F^2.
\end{equation}
}

Let $c_1>1$ be a value to be specified later. We discuss two separate cases.

\medskip
\noindent{\textbf{Case (i). $\srank(\A)> c_1 d$.}}

Let $\U$ be a uniformly random $n\times n$ orthogonal matrix. Note that $\|\A_{\mathsf{row}}\|=\|\A_{\mathsf{row}}\U\|$, and $(\A\U)_{i,:}=\A_{i,:}\U$ uniformly distributes on $\|\A_{i,:}\|_2\cdot\mathbb{S}^{n-1}$. So $\|\A_{i,:}\U\|_\infty^2\le 2\|\A_{i,:}\U\|_2^2\log (n)/n$ for all $i$ with probability at least $1-1/n^2$ by Lemma \ref{Lemma: infty norm of uniform sampling from unit sphere}. Therefore, with probability at least $1-1/n$ by union bound over all rows, $\|\A\U\|_{\mathsf{col}}^2\le 2\|\A\|_F^2\log (n)/n$, where $\|\A\|_{\mathsf{col}}$ represents the maximum norm among all columns of $\A$. By Lemma \ref{Lemma: Vershynin},
\begin{equation*}
\bbE\|\A_{\mathsf{row}}\|\le C_1\sqrt{\frac{q}{n}}\|\A\|+C_2\sqrt{\log q}\sqrt{\frac{\log n}{n}}\|\A\|_F
\end{equation*}
for absolute constants $C_1$ and $C_2$, and by the Markov bound, with probability at least $0.9$,
\begin{equation*}
\begin{split}
\|\A_{\mathsf{row}}\|&\le C_1\sqrt{\frac{q}{n}}\|\A\|+C_2\sqrt{\log q}\sqrt{\frac{\log n}{n}}\|\A\|_F\\
&\le C_1\sqrt{\frac{q}{n}}\frac{\|\A\|_F}{\sqrt{c_1 d}}+C_2\sqrt{\log q}\sqrt{\frac{\log n}{n}}\|\A\|_F\\
&\le \frac{1}{\sqrt{1-\tau}}\left(C_1\sqrt{\frac{q}{c_1d}}+C_2\sqrt{\log q\log n}\right)\sqrt{\frac{X}{n}}.
\end{split}
\end{equation*}
Conditioning on this event, by applying the same argument on the column sampling of $\A_{\mathsf{row}}$, we have
\begin{equation*}
\begin{split}
\|\widetilde\A\|&\le C_1\sqrt{\frac{q}{n}}\|\A_{\mathsf{row}}\|+C_2\sqrt{\log q}\sqrt{\frac{\log n}{n}}\|\A_{\mathsf{row}}\|_F\\
&\le C_1\sqrt{\frac{q}{n}} \frac{1}{\sqrt{1-\tau}}\left(C_1\sqrt{\frac{q}{c_1d}}+C_2\sqrt{\log q\log n}\right)\sqrt{\frac{X}{n}}+C_2\sqrt{\log q}\sqrt{\frac{\log n}{n}}\sqrt{\frac{q}{1-\tau}}\sqrt{\frac{X}{n}},
\end{split}
\end{equation*}
where $\widetilde\A$ is the matrix after column sampling of $\A_{\mathsf{row}}$.
On the other hand, when $\srank(\A)\le d$, by setting $\beta=\|\A\|_F^2/n^2$ and applying Corollary \ref{lemma: Ismail}, we have with probability at least $0.9$ that
\begin{equation}
\label{equ: row operator norm}
\|\A_{\mathsf{row}}\|\ge \sqrt{1-\tau}\sqrt{\frac{q}{n}}\|\A\|\ge \sqrt{1-\tau}\sqrt{\frac{q}{n}}\frac{\|\A\|_F}{\sqrt{d}}\ge \sqrt{\frac{1-\tau}{1+\tau}}\sqrt{\frac{q}{n}}\sqrt{\frac{X}{d}}.
\end{equation}
Conditioning on this event and applying Corollary \ref{lemma: Ismail} to column sampling of $\A_{\mathsf{row}}$, we have
\begin{equation*}
\|\widetilde\A\|\ge \sqrt{1-\tau}\sqrt{\frac{q}{n}}\|\A_{\mathsf{row}}\|\ge \frac{1-\tau}{\sqrt{1+\tau}}\frac{q}{n}\sqrt{\frac{X}{d}},
\end{equation*}
where we have used the fact that (by Eqns. \eqref{equ: F norm} and \eqref{equ: row operator norm})
\begin{equation}
\label{equ: stable rank of row sampling}
\srank(\A_{\mathsf{row}})=\frac{\|\A_{\mathsf{row}}\|_F^2}{\|\A_{\mathsf{row}}\|^2}\le \frac{1+\tau}{1-\tau}\srank(\A)\le \frac{1+\tau}{1-\tau}d.
\end{equation}
By setting $c_1$ larger than a constant times $(1+\tau)/(1-\tau)^3$, we have
\begin{equation*}
\frac{1-\tau}{\sqrt{1+\tau}}\frac{q}{n}\sqrt{\frac{X}{d}} > C_1\sqrt{\frac{q}{n}} \frac{1}{\sqrt{1-\tau}}\left(C_1\sqrt{\frac{q}{c_1d}}+C_2\sqrt{\log q\log n}\right)\sqrt{\frac{X}{n}}+C_2\sqrt{\log q}\sqrt{\frac{\log n}{n}}\sqrt{\frac{q}{1-\tau}}\sqrt{\frac{X}{n}}.
\end{equation*}
Thus we can distinguish the two cases of $\srank(\A)\le d$ and $\srank(\A)$ far from being at most $d$.

\medskip
\noindent{\textbf{Case (ii). $\srank(\A)\le c_1 d$.}}

Note that in this case, $\|\A\|^2\ge \frac{\|\A\|_F^2}{c_1 d}\ge \frac{\epsilon n^2(1-\frac{1}{d})}{c_1d}$. Let $\beta=\|\A\|_F^2/n^2$. So $\beta \|\A_{i,:}\|_2^2/\|\A\|_F^2\le 1/n$ for all $i$'s, namely, uniform sampling obeys the condition of Corollary \ref{lemma: Ismail}. It then follows from Corollary \ref{lemma: Ismail} that with probability at least $0.9$,
\begin{equation*}
(1-\tau)\|\A\|^2\le \frac{n}{q}\|\A_{\mathsf{row}}\|^2\le (1+\tau)\|\A\|^2.
\end{equation*}
Conditioning on this event, by applying Corollary \ref{lemma: Ismail} again to the column sampling of $\A_{\mathsf{row}}$, we have with high probability,
\begin{equation}
\label{equ: small stable rank, operator norm}
(1-\tau)^2\|\A\|^2\le (1-\tau)\frac{n}{q}\|\A_{\mathsf{row}}\|^2\le \frac{n^2}{q^2}\|\widetilde\A\|^2\le (1+\tau)\frac{n}{q}\|\A_{\mathsf{row}}\|^2\le (1+\tau)^2\|\A\|^2,
\end{equation}
where we have used the fact that $\srank(\A_{\mathsf{row}})\le \cO(d)$ by Eqn. \eqref{equ: stable rank of row sampling}. Combining Eqn. \eqref{equ: small stable rank, operator norm} with Eqn. \eqref{equ: F norm entry level}, we obtain an estimation of stable rank of $\A$ from $\srank(\widetilde\A)$:
\begin{equation}
\label{equ: stable rank estimator}
(1-\Theta(\tau))\cdot\srank(\A)\le \srank(\widetilde\A)\le (1+\Theta(\tau))\cdot\srank(\A).
\end{equation}

We now show that we can distinguish the two cases of $\srank(\A)\le d$ and $\srank(\A)$ far from being at most $d$ by the above estimator $\srank(\widetilde\A)$. Let $\x\in \mathbb{S}^{n-1}$ be the unit vector such that $\|\A\|=\|\A\x\|_2$, i.e., $\x$ is the right singular vector corresponding to the largest singular value. Decompose $\x$ as $\x_H+\x_L$, where $\x_H$ represents the ``heavy'' part of $\x$, i.e., zeroing out those entries whose values are smaller than a hard threshold $x_0=\Theta(\sqrt{\frac{d}{n}})$, and $\x_L$ represents the ``light'' part of $\x$ defined by $\x_L=\x-\x_H$. We note that $\x_H$ is sparse, i.e., $\|\x_H\|_0\le 1/x_0^2$, because $\|\x\|^2=1$. So
\begin{equation*}
\begin{split}
\|\A\|&=\|\A\x\|_2\\
&\le \|\A\x_H\|_2+\|\A\x_L\|_2\\
&\le \sqrt{\sum_{i\in H} \langle\A_{i,:},\x_H\rangle^2}+\|\A\x_L\|_2\\
&\le \frac{\sqrt{n}}{x_0}+\|\A\x_L\|_2\\
&\le \frac{\sqrt{n}}{x_0}+\|\A\|\|\x_L\|_2.
\end{split}
\end{equation*}
Therefore, $\|\x_L\|_2\ge 1-\frac{\sqrt{n}}{x_0\|\A\|}\ge 1-\frac{\sqrt{n}c_1d}{x_0\epsilon n^2(1-1/d)}=\Omega(1)$.
Suppose that $\A$ is $\epsilon/d$-far from being stable rank at most $d$. Without loss of generality, we assume that $\langle\A_{1,:},\x\rangle^2\le \langle\A_{2,:},\x\rangle^2\le...\le \langle\A_{n,:},\x\rangle^2$. Let $m=\lceil\epsilon n/d\rceil-1$. Replacing each $\A_{i,:}$, $i\in [m]$ with vector $\x_L^\top/x_0$ forms a new matrix $\B$. Since $\A$ is $\epsilon/d$-far from being stable rank at most $d$, it must hold that $\srank(\B)>d$ and $\|\B\|_\infty\le 1$.

Notice that
\begin{equation}
\label{equ: S_m}
S_m:=\sum_{i=1}^m\langle\A_{i,:},\x\rangle^2 \le \frac{m}{n} \sum_{i=1}^n\langle\A_{i,:},\x\rangle^2=\frac{m}{n}\|\A\|^2< \frac{m}{n}\frac{\|\A\|_F^2}{d}\le \frac{2\epsilon\|\A\|_F^2}{d^2},
\end{equation}
where the penultimate inequality holds because $\srank(\A)=\frac{\|\A\|_F^2}{\|\A\|^2}>d$ by assumption. We also observe that
\begin{equation}
\label{equ: F norm of B}
\|\B\|_F^2\le \|\A\|_F^2-\sum_{i=1}^m\|\A_{i,:}\|_2^2+\frac{m\|\x_L\|_2^2}{x_0^2}\le \|\A\|_F^2-S_m+\frac{m}{x_0^2},
\end{equation}
where the last inequality holds since $S_m=\sum_{i=1}^m\langle\A_{i,:},\x\rangle^2\le \sum_{i=1}^m \|\A_{i,:}\|_2^2$ by the Cauchy-Schwartz inequality.
Because
\begin{equation}
\label{equ: operator norm of B}
\|\B\|^2\ge \|\B\x\|_2^2= \sum_{i=m+1}^n \langle\A_{i,:},\x\rangle^2+m\left\langle\frac{\x_L}{x_0},\x\right\rangle^2= \|\A\|^2-S_m+\frac{m}{x_0^2}\|\x_L\|_2^2,
\end{equation}
it follows that
\begin{equation*}
\begin{split}
d&<\srank(\B)=\frac{\|\B\|_F^2}{\|\B\|^2}\\
&\le \frac{\|\A\|_F^2-S_m+\frac{m}{x_0^2}}{\|\A\|^2-S_m+\frac{m}{x_0^2}\|\x_L\|_2^2}\qquad (\text{by Eqns. } \eqref{equ: F norm of B} \text{ and } \eqref{equ: operator norm of B})\\
&\le \frac{\|\A\|_F^2+\frac{m}{x_0^2}}{\|\A\|^2-S_m+\frac{m}{x_0^2}\|\x_L\|_2^2}\\
&\le \frac{\|\A\|_F^2+\frac{2\epsilon n}{x_0^2d}}{\|\A\|^2-\frac{2\epsilon}{d}\frac{\|\A\|_F^2}{d}+\frac{2\epsilon n}{2x_0^2d}}.\qquad (\mbox{by Eqn. } \eqref{equ: S_m}\mbox{, }\|\x_L\|_2=\Omega(1) \mbox{ and } m=\Theta(\epsilon n/d))
\end{split}
\end{equation*}
Thus
\begin{equation*}
\begin{split}
\|\A\|^2&< \frac{\|\A\|_F^2}{d}+\frac{2\epsilon n}{x_0^2d^2}+\frac{2\epsilon}{d}\frac{\|\A\|_F^2}{d}-\frac{\epsilon n}{x_0^2d}\\
&\le \frac{\|\A\|_F^2}{d}+\frac{2\epsilon}{d}\frac{\|\A\|_F^2}{d}-\frac{\epsilon n}{x_0^2d}\left(1-\frac{2}{d}\right)\\
&\le \frac{\|\A\|_F^2}{d}+\frac{2\epsilon}{d}\frac{\|\A\|_F^2}{d}-\frac{\epsilon n}{2x_0^2d}\qquad \left(\mbox{since } \frac{1}{2}\le 1-\frac{2}{d}\right)\\
&\le \frac{\|\A\|_F^2}{d}+\frac{2\epsilon}{d}\frac{\|\A\|_F^2}{d}-\frac{\epsilon}{2x_0^2n}\frac{\|\A\|_F^2}{d}\qquad \left(\mbox{since }\frac{\|\A\|_F^2}{n}\le n\right)\\
&=\frac{\|\A\|_F^2}{d}\left(1-\left(\frac{\epsilon}{2x_0^2n}-\frac{2\epsilon}{d}\right)\right)\\
&=\frac{\|\A\|_F^2}{d}\left(1-\Theta\left(\frac{\epsilon}{d}\right)\right),\qquad \left(\mbox{since } x_0^2=\Theta\left(\frac{d}{n}\right)\right)
\end{split}
\end{equation*}
namely, $(1+\Theta(\epsilon/d))\cdot\srank(\A)>d$.
Setting $\tau$ in Eqn. \eqref{equ: stable rank estimator} as $\Theta(\frac{\epsilon}{d})$, we can distinguish the two cases of ``$\srank(\A)<d$'' from ``$\srank(\A)$ being far from stable rank $\le d$'' with sample complexity $\widetilde\cO(\frac{d^2}{\beta^2 \tau^4})=\widetilde\cO(\frac{d^6}{\epsilon^6})$.

\end{proof}

\red{Improved Proof:}
\begin{theorem}
Suppose that $d/\epsilon\ge 2$ and $d\ge 2$. Then Algorithm \ref{algorithm: algorithm for the stable rank testing problem} is a correct algorithm for the stable rank testing problem with failure probability at most $\frac{1}{3}$ in the entry model. It reads $\widetilde\cO(\frac{d^2}{\epsilon^6})$ entries.
\end{theorem}

\begin{proof}
When $\A$ is far from being stable rank at most $d$, we claim that $\|\A\|_F^2\ge \epsilon n^2(1-\frac{1}{d})$. Otherwise, changing the first $\frac{\epsilon n}{d}$ rows of $\A$ to be all-one vectors $\1^\top$ leads to a new matrix $\B$ such that $\|\B\|^2\ge \frac{\epsilon n^2}{d}$ and $\|\B\|_F^2\le \epsilon n^2(1-\frac{1}{d})+\frac{\epsilon n^2}{d}=\epsilon n^2$, leading to $\srank(\B)\le d$, a contradiction. By the Chernoff bound, it holds that with probability at least $1-e^{-\tau^2\|\A\|_F^2/3}\ge 1-e^{-\tau^2\epsilon n^2(1-1/d)/3}$, by sampling $q$ rows from $\A$, the resulting matrix $\A_{\mathsf{row}}$ satisfies
\begin{equation}
\label{equ: F norm}
(1-\tau)\|\A\|_F^2\le \frac{n}{q}\|\A_{\mathsf{row}}\|_F^2\le (1+\tau)\|\A\|_F^2.
\end{equation}
Conditioning on this event, by sampling $q$ columns from $\A_{\mathsf{row}}$ and using Chernoff bound again, with probability at least $1-e^{-\tau^2\epsilon n^2(1-1/d)/3}-e^{-\tau^2\|\A_{\mathsf{row}}\|^2/3}\ge 1-2e^{-\tau^2\epsilon n^2(1-1/d)q(1-\tau)/(3n)}$ we have
\begin{equation}
\label{equ: F norm entry level}
(1-\tau)^2\|\A\|_F^2\le (1-\tau)\frac{n}{q}\|\A_{\mathsf{row}}\|_F^2\le X\le (1+\tau)\frac{n}{q}\|\A_{\mathsf{row}}\|_F^2\le (1+\tau)^2\|\A\|_F^2.
\end{equation}

\comment{
Conditioned on the event $|X-\|\A\|_F^2|\le \frac{1}{8}(1-\frac{1}{d})^2\epsilon n$ in Lemma \ref{lemma: concentration}, it holds that
\begin{equation*}
X\ge \|\A\|_F^2-\frac{1}{8}\left(1-\frac{1}{d}\right)^2\epsilon n\ge \frac{9}{10}\epsilon n^2\left(1-\frac{1}{d}\right).
\end{equation*}
Therefore, Algorithm \ref{algorithm: algorithm for the stable rank testing problem} is correct on Line 5.

To prove the correctness of Lines 8 and 10, we only need to discuss the case when $X\ge \frac{9}{10}\epsilon n^2\left(1-\frac{1}{d}\right)$. So
\begin{equation*}
\|\A\|_F^2\ge X-\frac{1}{8}\left(1-\frac{1}{d}\right)^2\epsilon n\ge \frac{8}{10}\epsilon n^2\left(1-\frac{1}{d}\right).
\end{equation*}
Let $\eta:=\frac{c\epsilon n}{\|\A\|_F^2}\le \frac{10cd}{8(d-1)}=:\eta'$. We have
\begin{equation}
\left(1-\frac{\tau\eta'}{c}\right)\|\A\|_F^2\le X\le \left(1+\frac{\tau\eta'}{c}\right)\|\A\|_F^2.
\end{equation}
}

Let $c_1>1$ be a value to be specified later. We discuss two separate cases.

\medskip
\noindent{\textbf{Case (i). $\srank(\A)> c_1 d$.}}

Let $\U$ be a uniformly random $n\times n$ orthogonal matrix. Note that $\|\A_{\mathsf{row}}\|=\|\A_{\mathsf{row}}\U\|$, and $(\A\U)_{i,:}=\A_{i,:}\U$ uniformly distributes on $\|\A_{i,:}\|_2\cdot\mathbb{S}^{n-1}$. So $\|\A_{i,:}\U\|_\infty^2\le 2\|\A_{i,:}\U\|_2^2\log (n)/n$ for all $i$ with probability at least $1-1/n^2$ by Lemma \ref{Lemma: infty norm of uniform sampling from unit sphere}. Therefore, with probability at least $1-1/n$ by union bound over all rows, $\|\A\U\|_{\mathsf{col}}^2\le 2\|\A\|_F^2\log (n)/n$, where $\|\A\|_{\mathsf{col}}$ represents the maximum norm among all columns of $\A$. By Lemma \ref{Lemma: Vershynin},
\begin{equation*}
\bbE\|\A_{\mathsf{row}}\|\le C_1\sqrt{\frac{q}{n}}\|\A\|+C_2\sqrt{\log q}\sqrt{\frac{\log n}{n}}\|\A\|_F
\end{equation*}
for absolute constants $C_1$ and $C_2$, and by the Markov bound, with probability at least $0.9$,
\begin{equation*}
\begin{split}
\|\A_{\mathsf{row}}\|&\le C_1\sqrt{\frac{q}{n}}\|\A\|+C_2\sqrt{\log q}\sqrt{\frac{\log n}{n}}\|\A\|_F\\
&\le C_1\sqrt{\frac{q}{n}}\frac{\|\A\|_F}{\sqrt{c_1 d}}+C_2\sqrt{\log q}\sqrt{\frac{\log n}{n}}\|\A\|_F\\
&\le \frac{1}{\sqrt{1-\tau}}\left(C_1\sqrt{\frac{q}{c_1d}}+C_2\sqrt{\log q\log n}\right)\sqrt{\frac{X}{n}}.
\end{split}
\end{equation*}
Conditioning on this event, by applying the same argument on the column sampling of $\A_{\mathsf{row}}$, we have
\begin{equation*}
\begin{split}
\|\widetilde\A\|&\le C_1\sqrt{\frac{q}{n}}\|\A_{\mathsf{row}}\|+C_2\sqrt{\log q}\sqrt{\frac{\log n}{n}}\|\A_{\mathsf{row}}\|_F\\
&\le C_1\sqrt{\frac{q}{n}} \frac{1}{\sqrt{1-\tau}}\left(C_1\sqrt{\frac{q}{c_1d}}+C_2\sqrt{\log q\log n}\right)\sqrt{\frac{X}{n}}+C_2\sqrt{\log q}\sqrt{\frac{\log n}{n}}\sqrt{\frac{q}{1-\tau}}\sqrt{\frac{X}{n}},
\end{split}
\end{equation*}
where $\widetilde\A$ is the matrix after column sampling of $\A_{\mathsf{row}}$.
On the other hand, when $\srank(\A)\le d$, by setting $\beta=\|\A\|_F^2/n^2$ and applying Corollary \ref{lemma: Ismail}, we have with probability at least $0.9$ that
\begin{equation}
\label{equ: row operator norm}
\|\A_{\mathsf{row}}\|\ge \sqrt{1-\tau}\sqrt{\frac{q}{n}}\|\A\|\ge \sqrt{1-\tau}\sqrt{\frac{q}{n}}\frac{\|\A\|_F}{\sqrt{d}}\ge \sqrt{\frac{1-\tau}{1+\tau}}\sqrt{\frac{q}{n}}\sqrt{\frac{X}{d}}.
\end{equation}
Conditioning on this event and applying Corollary \ref{lemma: Ismail} to column sampling of $\A_{\mathsf{row}}$, we have
\begin{equation*}
\|\widetilde\A\|\ge \sqrt{1-\tau}\sqrt{\frac{q}{n}}\|\A_{\mathsf{row}}\|\ge \frac{1-\tau}{\sqrt{1+\tau}}\frac{q}{n}\sqrt{\frac{X}{d}},
\end{equation*}
where we have used the fact that (by Eqns. \eqref{equ: F norm} and \eqref{equ: row operator norm})
\begin{equation}
\label{equ: stable rank of row sampling}
\srank(\A_{\mathsf{row}})=\frac{\|\A_{\mathsf{row}}\|_F^2}{\|\A_{\mathsf{row}}\|^2}\le \frac{1+\tau}{1-\tau}\srank(\A)\le \frac{1+\tau}{1-\tau}d.
\end{equation}
By setting $c_1$ larger than a constant times $(1+\tau)/(1-\tau)^3$, we have
\begin{equation*}
\frac{1-\tau}{\sqrt{1+\tau}}\frac{q}{n}\sqrt{\frac{X}{d}} > C_1\sqrt{\frac{q}{n}} \frac{1}{\sqrt{1-\tau}}\left(C_1\sqrt{\frac{q}{c_1d}}+C_2\sqrt{\log q\log n}\right)\sqrt{\frac{X}{n}}+C_2\sqrt{\log q}\sqrt{\frac{\log n}{n}}\sqrt{\frac{q}{1-\tau}}\sqrt{\frac{X}{n}}.
\end{equation*}
Thus we can distinguish the two cases of $\srank(\A)\le d$ and $\srank(\A)$ far from being at most $d$.

\medskip
\noindent{\textbf{Case (ii). $\srank(\A)\le c_1 d$.}}

Let $\beta=\|\A\|_F^2/n^2$. So $\beta \|\A_{i,:}\|_2^2/\|\A\|_F^2\le 1/n$ for all $i$'s, namely, uniform sampling obeys the condition of Corollary \ref{lemma: Ismail}. It then follows from Corollary \ref{lemma: Ismail} that with probability at least $0.9$,
\begin{equation*}
(1-\tau)\|\A\|^2\le \frac{n}{q}\|\A_{\mathsf{row}}\|^2\le (1+\tau)\|\A\|^2.
\end{equation*}
Conditioning on this event, by applying Corollary \ref{lemma: Ismail} again to the column sampling of $\A_{\mathsf{row}}$, we have with high probability,
\begin{equation}
\label{equ: small stable rank, operator norm}
(1-\tau)^2\|\A\|^2\le (1-\tau)\frac{n}{q}\|\A_{\mathsf{row}}\|^2\le \frac{n^2}{q^2}\|\widetilde\A\|^2\le (1+\tau)\frac{n}{q}\|\A_{\mathsf{row}}\|^2\le (1+\tau)^2\|\A\|^2,
\end{equation}
where we have used the fact that $\srank(\A_{\mathsf{row}})\le \cO(d)$ by Eqn. \eqref{equ: stable rank of row sampling}. Combining Eqn. \eqref{equ: small stable rank, operator norm} with Eqn. \eqref{equ: F norm entry level}, we obtain an estimation of stable rank of $\A$ from $\srank(\widetilde\A)$:
\begin{equation}
\label{equ: stable rank estimator}
(1-\Theta(\tau))\cdot\srank(\A)\le \srank(\widetilde\A)\le (1+\Theta(\tau))\cdot\srank(\A).
\end{equation}

We now show that we can distinguish the two cases of $\srank(\A)\le d$ and $\srank(\A)$ far from being at most $d$ by the above estimator $\srank(\widetilde\A)$. Let $\x\in \mathbb{S}^{n-1}$ be the unit vector such that $\|\A\|=\|\A\x\|_2$, i.e., $\x$ is the right singular vector corresponding to the largest singular value. Let $\x^{(i)}$ be the projection of $\x$ on those coordinates with absolute value in $(1/2^i,1/2^{i-1})$, so that $\x=\sum_{i=0}^\infty\x^{(i)}$. Then $\|\A\x\|_2\le \sum_{i=0}^\infty\|\A\x^{(i)}\|_2$ by the triangle inequality. Also, for all $i$'s, we have
\begin{equation*}
\|\A\x^{(i)}\|_2\le \|\A\| \|\x^{(i)}\|_2\le \frac{\sqrt{n}\|\A\|}{2^{i-1}}\le \frac{\sqrt{n}\|\A\|_F}{2^{i-1}}\le \frac{n^{3/2}}{2^{i-1}},
\end{equation*}
which implies that
\begin{equation*}
\sum_{i>2\log n} \|\A\x^{(i)}\|_2\le \sum_{i>2\log n} \frac{n^{3/2}}{2^{i-1}}\le \frac{2}{\sqrt{n}}.
\end{equation*}
On the other hand, notice that $\|\A\|\ge \|\A\|_F/\sqrt{n}\ge \sqrt{\epsilon n}\sqrt{1-\frac{1}{d}}$. Therefore,
\begin{equation*}
\sum_{i=0}^{2\log n}\|\A\x^{(i)}\|_2 \ge \|\A\|-\sum_{i=2\log n}^{\infty}\|\A\x^{(i)}\|_2\ge \frac{\|\A\|}{2}.
\end{equation*}
So by averaging, there must exist a $k\in[2\log n]$ such that $\|\A\x^{(k)}\|_2\ge \|\A\|/(4\log n)$. Denote by $\x_L:=\x^{(k)}$ and $x_0:=1/2^{k-1}$, where $\|\x_L\|_\infty\le x_0$. Note that $\|\x_L\|_2\le \|\x\|_2\le 1$.

Suppose that $\A$ is $\epsilon/d$-far from being stable rank at most $d$. Without loss of generality, we assume that $\langle\A_{1,:},\x\rangle^2\le \langle\A_{2,:},\x\rangle^2\le...\le \langle\A_{n,:},\x\rangle^2$. Let $m=\frac{\epsilon' n^2}{d \|\x_L\|_0}$, where $\epsilon'=\frac{\epsilon(1-\frac{1}{d})}{512c_1 \log^4 (n)}<\epsilon$. Replacing each $\A_{i,\mathsf{supp}(\x_L)}$, $i\in [m]$ with $\x_L^\top/x_0$ forms a new matrix $\B$. Since $\A$ is $\epsilon/d$-far from being stable rank at most $d$, it must hold that $\srank(\B)>d$ and $\|\B\|_\infty\le 1$.
Let $\mathsf{supp}(\x_L)$ be the support set of $\x_L$. We also have
\begin{equation*}
\begin{split}
\|\A\|^2&\le 16\log^2 (n)\|\A\x_L\|_2^2\\
&=16\log^2(n) \sum_{i=1}^n \langle\A_{i,:},\x_L\rangle^2\\
&\le 16\log^2(n) \sum_{i=1}^n \|\A_{i,\mathsf{supp}(\x_L)}\|_2^2 \|\x_L\|_2^2\qquad (\text{by Cauchy-Schwartz})\\
&\le 16\log^2(n) n\|\x_L\|_0\qquad (\text{since } \|\x_L\|_2\le 1,\ \|\A_{i,\mathsf{supp}(\x_L)}\|_2^2\le \|\x_L\|_0),
\end{split}
\end{equation*}
which implies that
\begin{equation*}
c_1d\ge \srank(\A)=\frac{\|\A\|_F^2}{\|\A\|^2}\ge \frac{\epsilon n(1-\frac{1}{d})}{16\log^2(n) \|\x_L\|_0},
\end{equation*}
i.e., $\|\x_L\|_0\ge \frac{\epsilon n(1-\frac{1}{d})}{16c_1 d\log^2 (n)}=\frac{32\epsilon' n \log^2 (n)}{d}$.

Notice that
\begin{equation}
\label{equ: S_m}
S_m:=\sum_{i=1}^m\langle\A_{i,:},\x_L\rangle^2 \le \frac{m}{n} \sum_{i=1}^n\langle\A_{i,:},\x_L\rangle^2=\frac{m}{n}\|\A\x_L\|_2^2\le \frac{m}{n}\|\A\|^2,
\end{equation}
where the penultimate inequality holds because $\srank(\A)=\frac{\|\A\|_F^2}{\|\A\|^2}>d$ by assumption. We also observe that
\begin{equation}
\label{equ: F norm of B}
\|\B\|_F^2\le \|\A\|_F^2+\frac{m\|\x_L\|_2^2}{x_0^2}.
\end{equation}
Because
\begin{equation}
\label{equ: operator norm of B}
\begin{split}
\|\B\|^2&=\max_{\|\y\|_2\le 1} \|\B\y\|_2^2\\
&\ge \|\B\x_L\|_2^2\qquad (\text{because } \|\x_L\|_2\le 1)\\
&= \sum_{i=m+1}^n \langle\A_{i,:},\x_L\rangle^2+m\left\langle\frac{\x_L}{x_0},\x_L\right\rangle^2\\
&=\|\A\x_L\|_2^2-S_m+\frac{m}{x_0^2}\|\x_L\|_2^2\\
&\ge \frac{\|\A\|^2}{16\log^2 n}-S_m+\frac{m}{4}\|\x_L\|_0\quad \left(\text{since non-zero entries of } \x_L \text{ in }\left[\frac{x_0}{2},x_0\right] \text{ so } \|\x_L\|_2^2\ge \frac{\|\x_L\|_0x_0^2}{4}\right)\\
&\ge \frac{\|\A\|^2}{16\log^2 n}-\frac{m}{n}\|\A\|^2+\frac{m}{4}\|\x_L\|_0\qquad (\text{by Eqn. \eqref{equ: S_m}})\\
&\ge \frac{\|\A\|^2}{32\log^2 n}+\frac{m}{4}\|\x_L\|_0,\qquad \left(\text{since }\frac{\|\A\|^2}{32\log ^2n}\ge \frac{m}{n}\|\A\|^2\right)
\end{split}
\end{equation}
and the fact that the number of changed entries is $\frac{\epsilon' n^2}{d}<\frac{\epsilon n^2}{d}$, it follows from the definition of being $\epsilon/d$-far from having stable rank at most $d$ that
\begin{equation*}
\begin{split}
d&<\srank(\B)=\frac{\|\B\|_F^2}{\|\B\|^2}\\
&\le \frac{\|\A\|_F^2+\frac{m\|\x_L\|_2^2}{x_0^2}}{\frac{\|\A\|^2}{16\log^2 n}+\frac{m}{4}\|\x_L\|_0}\qquad (\text{by Eqns. } \eqref{equ: F norm of B} \text{ and } \eqref{equ: operator norm of B})\\
&\le \frac{\|\A\|_F^2+\frac{\epsilon' n^2\|\x_L\|_2^2}{x_0^2d\|\x_L\|_0}}{\frac{\|\A\|^2}{16\log^2 n}+\frac{\epsilon' n^2}{4d}}.\qquad (\mbox{by Eqn. } \eqref{equ: S_m} \mbox{ and } m=\epsilon' n^2/(\|\x_L\|_0d))
\end{split}
\end{equation*}
Thus
\begin{equation*}
\begin{split}
\frac{\|\A\|^2}{16\log^2 n}&< \frac{\|\A\|_F^2}{d}+\frac{\epsilon' n^2\|\x_L\|_2^2}{x_0^2d^2\|\x_L\|_0}-\frac{\epsilon' n^2}{4d}\\
&\le \frac{\|\A\|_F^2}{d}+\frac{\epsilon' n^2}{d^2}-\frac{\epsilon' n^2}{4d}\quad (\text{because } \|\x_L\|_2^2\le x_0^2\|\x_L\|_0)\\
&\le \frac{\|\A\|_F^2}{d}-\frac{\epsilon' n^2}{8d}\\
&\le \frac{\|\A\|_F^2}{d}-\frac{\epsilon' \|\A\|_F^2}{8d}\quad (\text{because } \|\A\|_F^2\le n^2)\\
&\le \frac{\|\A\|_F^2}{d}\left(1-\frac{\epsilon'}{8}\right).
\end{split}
\end{equation*}
That is, $(1+\Theta(\epsilon'))\cdot\srank(\A)>d$.
Setting $\tau$ in Eqn. \eqref{equ: stable rank estimator} as $\Theta(\epsilon')$, we can distinguish the two cases of $\srank(\A)<d$ and $\srank(\A)$ being far from stable rank $\le d$ with sample complexity $\widetilde\cO(\frac{d^2}{\beta^2 \tau^4})=\widetilde\cO(\frac{d^2}{\epsilon^6})$, where we have used the fact that $\beta=\|\A\|_F^2/n^2\ge \epsilon (1-\frac{1}{d})\ge \frac{\epsilon}{2}$.

\comment{
To get $d^2/\epsilon^2$, we need $\log^2(1/\epsilon)$ algorithms. In particular, we run Ismail with each possible $\beta$ a power of $2$ between $1$ and $1/\epsilon$, and do this for both the row and then the column sampling. Then the idea is that if $\beta$ is large, e.g., $1/\epsilon$ for rows and $1$ for columns, this means the operator norm is concentrated on an $\epsilon n \times n$ submatrix, and so you need to change a $1/d$ fraction of entries on the submatrix, not an $\epsilon/d$, and so the second part of our algorithm can be run with constant eps on that submatrix. So in this case we only get in total $d^2/\epsilon$. But if $\beta=1$ originally then the second part of the algorithm needs changing $\epsilon/d$ fraction of entries but since $\beta=1$ this gives $d^2/\epsilon^2$ in this case.
}
\end{proof}
}

\begin{algorithm}
\caption{Algorithm for stable rank testing under sampling/sensing model}
\label{algorithm: algorithm for the stable rank testing problem}
\begin{algorithmic}[1]
\LineComment Lines 1-2 estimate the Frobenius norm of $\A$.
\State Uniformly sample $q_0=\cO(\frac{\sqrt{d}}{\epsilon^{2.5}})$ entries $\A$, forming vector $\y$.
\State $X\leftarrow \frac{n^2}{q_0}\|\y\|_2^2$.\Comment{$X$ is an estimator of $\|\A\|_F^2$.}
\If{$X\le \frac{9}{10}(1-\frac{1}{d})\epsilon n^2$}
	\State Output ``stable rank $\le d$''.  \label{alg:F norm constant approximation}
\Else
	\State Uniformly sample a $q\times q$ submatrix $\widetilde\A'$ with $q=\cO(\frac{d\log n}{\epsilon})$.
	\If {$\|\widetilde\A'\|\le C_0\frac{\sqrt{X}}{\sqrt{c_1d}}\frac{q}{n}$}
		\State Output ``$\epsilon/d$-far from being stable rank $\le d$''.
	\Else
		\State Run Algorithm \ref{algorithm: the sampling algorithm for even p} (with $\tau=\Theta(\epsilon/d^{1/4})$) for the sampling model or Algorithm \ref{algorithm: the sketching algorithm for even p} (with $\tau=\Theta(\epsilon/(d^{1/4}\sqrt{\log n}))$) for the sensing model to obtain an operator norm estimate $Z$.
		\If{$Z^2\ge \frac{X}{d}$}
			\State Output ``stable rank $\le d$''.
		\Else
			\State Output ``$\epsilon/d$-far from being stable rank $\le d$''.
		\EndIf
	\EndIf
\EndIf
\end{algorithmic}
\end{algorithm}

\begin{theorem}
\label{theorem: stable rank testing upper bound}
Suppose that $d=\Omega((1/\epsilon)^{1/3})$. Then (a) Algorithm \ref{algorithm: algorithm for the stable rank testing problem} is a correct algorithm for the stable rank testing problem with failure probability at most $1/3$ under the sampling model, and it reads $\cO(\frac{d^3}{\epsilon^4}\log^2 n)$ entries; (b) Algorithm \ref{algorithm: algorithm for the stable rank testing problem} is a correct algorithm for the stable rank testing problem with failure probability at most $1/3$ under the sensing model, and it makes $\cO(\frac{d^{2.5}}{\epsilon^2}\log n)$ sensing queries.
\end{theorem}

\begin{proof}
When $\A$ is $\epsilon/d$-far from being stable rank at most $d$, we claim that $\|\A\|_F^2\ge \epsilon n^2(1-\frac{1}{d})$. Otherwise, replacing any $\frac{\epsilon n}{d}$ rows of $\A$ with all-one row vectors $\1^\top$'s results in a new matrix $\B$ such that $\|\B\|^2\ge \frac{\epsilon n^2}{d}$ and $\|\B\|_F^2= \|\A\|_F^2+\frac{\epsilon n^2}{d}\le \epsilon n^2(1-\frac{1}{d})+\frac{\epsilon n^2}{d}=\epsilon n^2$, leading to $\srank(\B)\le d$, a contradiction. We note that by sampling $q_0$ entries from $\A$ and stacking them as vector $\y$, the resulting estimator $X=\frac{n^2}{q_0^2}\|\y\|_2^2$ satisfies $\bbE[X]=\|\A\|_F^2$ and $\mathsf{Var}[X]\le n^2(n^4/q_0^2)(q_0/n^2)=n^4/q_0$. So by the Chebyshev's inequality, when $q_0=\cO(n^4/(\tau^2\|\A\|_F^4))$ we have
\begin{equation*}
\Pr\left[|X-\|\A\|_F^2|>\tau\|\A\|_F^2\right]\le \frac{n^4/q_0}{\tau^2\|\A\|_F^4}\le \frac{1}{3}.
\end{equation*}
Thus we have
\begin{equation}
\label{equ: F norm entry level}
(1-\tau)\|\A\|_F^2\le X\le (1+\tau)\|\A\|_F^2,
\end{equation}
where $\tau$ will be specified later multiple times.
So the algorithm is correct in Step \ref{alg:F norm constant approximation} with constant $\tau$ (although we over-sample entries here for later purpose).

\comment{
Conditioned on the event $|X-\|\A\|_F^2|\le \frac{1}{8}(1-\frac{1}{d})^2\epsilon n$ in Lemma \ref{lemma: concentration}, it holds that
\begin{equation*}
X\ge \|\A\|_F^2-\frac{1}{8}\left(1-\frac{1}{d}\right)^2\epsilon n\ge \frac{9}{10}\epsilon n^2\left(1-\frac{1}{d}\right).
\end{equation*}
Therefore, Algorithm \ref{algorithm: algorithm for the stable rank testing problem} is correct on Line 5.

To prove the correctness of Lines 8 and 10, we only need to discuss the case when $X\ge \frac{9}{10}\epsilon n^2\left(1-\frac{1}{d}\right)$. So
\begin{equation*}
\|\A\|_F^2\ge X-\frac{1}{8}\left(1-\frac{1}{d}\right)^2\epsilon n\ge \frac{8}{10}\epsilon n^2\left(1-\frac{1}{d}\right).
\end{equation*}
Let $\eta:=\frac{c\epsilon n}{\|\A\|_F^2}\le \frac{10cd}{8(d-1)}=:\eta'$. We have
\begin{equation}
\left(1-\frac{\tau\eta'}{c}\right)\|\A\|_F^2\le X\le \left(1+\frac{\tau\eta'}{c}\right)\|\A\|_F^2.
\end{equation}
}

Let $c_1>1$ be an absolute constant to be specified later. We discuss two separate cases.

\medskip
\noindent{\textbf{Case (i). $\srank(\A)> c_1 d$ when $\A$ is far from $\srank(\A)\le d$.}}

We first discuss the case when $\A$ is far from $\srank(\A)\le d$. Let $\U$ be a uniformly random $n\times n$ orthogonal matrix and let $\A_{\mathsf{row}}'$ be the matrix after uniform row sampling of $\A$ of expected cardinality $q$. Note that $\|\A_{\mathsf{row}}'\|=\|\A_{\mathsf{row}}'\U\|$, and $(\A\U)_{i,:}=\A_{i,:}\U$ uniformly distributes on $\|\A_{i,:}\|_2\cdot\mathbb{S}^{n-1}$. So $\|\A_{i,:}\U\|_\infty^2\le 2\|\A_{i,:}\U\|_2^2\log (n)/n$ for any fixed $i$ with probability at least $1-1/n^2$ by Lemma \ref{Lemma: infty norm of uniform sampling from unit sphere}. Therefore, with probability at least $1-1/n$ by a union bound over all rows, $\|\A\U\|_{\mathsf{col}}^2\le 2\|\A\|_F^2\log (n)/n$, where $\|\A\|_{\mathsf{col}}$ represents the maximum $\ell_2$ norm among all columns of $\A$. By Lemma \ref{Lemma: Vershynin},
\begin{equation*}
\bbE\|\A_{\mathsf{row}}'\|\le C_1'\sqrt{\frac{q}{n}}\|\A\|+C_2'\sqrt{\log q}\sqrt{\frac{\log n}{n}}\|\A\|_F
\end{equation*}
for absolute constants $C_1'$ and $C_2'$, and by the Markov bound, with probability at least $0.9$,
\begin{equation*}
\begin{split}
\|\A_{\mathsf{row}}'\|&\le C_1\sqrt{\frac{q}{n}}\|\A\|+C_2\sqrt{\log q}\sqrt{\frac{\log n}{n}}\|\A\|_F\\
&\le C_1\sqrt{\frac{q}{n}}\frac{\|\A\|_F}{\sqrt{c_1 d}}+C_2\sqrt{\log q}\sqrt{\frac{\log n}{n}}\|\A\|_F\quad(\text{since }\srank(\A)>c_1d)\\
&\le \frac{1}{\sqrt{1-\tau}}\left(C_1\sqrt{\frac{q}{c_1d}}+C_2\sqrt{\log q\log n}\right)\sqrt{\frac{X}{n}}\quad(\text{by Eqn. \eqref{equ: F norm entry level}})
\end{split}
\end{equation*}
for absolute constants $C_1$ and $C_2$. By the Markov bound, we also have with constant probability that
\begin{equation*}
\|\A_{\mathsf{row}}'\|_F^2\le c'\frac{q}{n}\|\A\|_F^2\le c\frac{q}{n}X.
\end{equation*}
Conditioning on this event, by applying the same argument on the column sampling of $\A_{\mathsf{row}}'$, we have
\begin{equation*}
\begin{split}
\|\widetilde\A'\|&\le C_1\sqrt{\frac{q}{n}}\|\A_{\mathsf{row}}'\|+C_2\sqrt{\log q}\sqrt{\frac{\log n}{n}}\|\A_{\mathsf{row}}'\|_F\\
&\le C_1\sqrt{\frac{q}{n}} \frac{1}{\sqrt{1-\tau}}\left(C_1\sqrt{\frac{q}{c_1d}}+C_2\sqrt{\log q\log n}\right)\sqrt{\frac{X}{n}}+C_2\sqrt{\log q}\sqrt{\frac{\log n}{n}}\sqrt{qc}\sqrt{\frac{X}{n}}\\
&\le C_0\frac{1}{\sqrt{1-\tau}}\frac{\sqrt{X}}{\sqrt{c_1d}}\frac{q}{n}\qquad(\text{because the first term dominates as } q\gg d)\\
&\le C_0'\sqrt{\frac{1+\tau}{1-\tau}}\frac{\|\A\|_F}{\sqrt{c_1d}}\frac{q}{n},\qquad(\text{by Eqn. }\eqref{equ: F norm entry level})
\end{split}
\end{equation*}
where $\widetilde\A'$ is the matrix after the column sampling of $\A_{\mathsf{row}}'$, and $C_0,C_0'$ are absolute constants.

On the other hand, when $\srank(\A)\le d$ and $q=\cO(\frac{d\log n}{\epsilon})$, we have with high probability that
\begin{equation*}
\|\widetilde\A'\|\ge C\frac{q}{n}\|\A\|= C\frac{q}{n}\frac{\|\A\|_F}{\sqrt{\srank(\A)}}\ge C\frac{q}{n}\frac{\|\A\|_F}{\sqrt{d}},
\end{equation*}
where the first inequality holds by applying Lemma \ref{lemma: Ismail} twice on row and column sampling (set $\beta=\Theta(\epsilon)$ and $\tau=\Theta(1)$ there).
\comment{
by setting $\beta=\|\A\|_F^2/n^2$ and applying Corollary \ref{lemma: Ismail}, we have with probability at least $0.9$ that
\begin{equation}
\label{equ: row operator norm}
\|\A_{\mathsf{row}}\|\ge \sqrt{1-\tau}\sqrt{\frac{q}{n}}\|\A\|\ge \sqrt{1-\tau}\sqrt{\frac{q}{n}}\frac{\|\A\|_F}{\sqrt{d}}\ge \sqrt{\frac{1-\tau}{1+\tau}}\sqrt{\frac{q}{n}}\sqrt{\frac{X}{d}}.
\end{equation}
Conditioning on this event and applying Corollary \ref{lemma: Ismail} to the column sampling of $\A_{\mathsf{row}}$, we have
\begin{equation*}
\|\widetilde\A\|\ge \sqrt{1-\tau}\sqrt{\frac{q}{n}}\|\A_{\mathsf{row}}\|\ge \frac{1-\tau}{\sqrt{1+\tau}}\frac{q}{n}\sqrt{\frac{X}{d}},
\end{equation*}
where we have used the fact that (by Eqns. \eqref{equ: F norm} and \eqref{equ: row operator norm})
\begin{equation}
\label{equ: stable rank of row sampling}
\srank(\A_{\mathsf{row}})=\frac{\|\A_{\mathsf{row}}\|_F^2}{\|\A_{\mathsf{row}}\|^2}\le \frac{1+\tau}{1-\tau}\srank(\A)\le \frac{1+\tau}{1-\tau}d.
\end{equation}
}
By setting $c_1$ as a large absolute constant, we have
\begin{equation*}
C_0'\sqrt{\frac{1+\tau}{1-\tau}}\frac{\|\A\|_F}{\sqrt{c_1d}}\frac{q}{n} < C\frac{q}{n}\frac{\|\A\|_F}{\sqrt{d}}.
\end{equation*}
Thus we can distinguish (a) $\srank(\A)\le d$ from (b) $\srank(\A)$ $\epsilon/d$-far from being at most $d$ by checking $\|\widetilde\A'\|$ in Case (i).

\medskip
\noindent{\textbf{Case (ii). $\srank(\A)\le c_1 d$ when $\A$ is far from $\srank(\A)\le d$.}}

We now show that we can distinguish the two cases of $\srank(\A)\le d$ from $\srank(\A)$ being $\epsilon/d$-far from at most $d$, suppose we have an accurate estimator to estimate the stable rank.

Let $\u\in \mathbb{S}^{n-1}$ be a unit vector such that $\|\A\|=\|\A \u\|_2$, i.e., $\u$ is a right singular vector corresponding to the largest singular value. First we claim that we can drop off coordinates in $\u$ that are at most $\theta/\sqrt{n}$ for some small constant $\theta$ without affecting $\|\A \u\|_2$ by too much.

Let $\u'$ be the vector obtained from $\u$ by zeroing out the coordinates of $\u$ which are \textit{at least} $\theta/\sqrt{n}$, then
\[
\|\A \u'\|_2^2 \leq \|\A\|^2 \|\u'\|_2^2 \leq \|\A\|^2 n \left(\frac{\theta}{\sqrt n}\right)^2 \leq \theta^2\|\A\|^2,
\]
and thus
\[
\|\A(\u-\u')\|_2 \geq \|\A \u\|_2 - \|\A \u'\|_2 \geq (1-\theta)\|\A\|.
\]
Let $\u'' = \u-\u'$ and $\v = \A\u''/\|\A\u''\|_2$, then $(1-\theta)\|\A\| \leq \langle \A\u'',\v\rangle$. Next we show similarly that we can drop off coordinates in $\v$ that are at most $\theta/\sqrt{n}$. Similarly we let $\v'$ be the vector obtained from $\v$ by zeroing out the coordinates of $\v$ which are at least $\theta/\sqrt{n}$, then
$\|\v'\|_2 \leq \theta$, hence
\[
\langle \A\u'',\v-\v'\rangle \geq (1-\theta)\|\A\| - \langle \A\u'',\v'\rangle
\geq (1-\theta)\|\A\| - \|\A\u''\|_2\|\v'\|_2
\geq (1-2\theta)\|\A\|.
\]
Let $\v'' = \v-\v'$. Observe that
\[
(1-2\theta)\frac{\|\A\|_F}{\sqrt{c_1 d}}\leq (1-2\theta)\|\A\| \leq \langle \A\u'',\v''\rangle \leq \|\A\u''\|_\infty \|\v''\|_1 \leq \|\u''\|_1\|\v''\|_1,
\]
where we used the fact that $|\A_{ij}|\leq 1$ in the last inequality. This implies that at least one of $\|\u''\|_1$ and $\|\v''\|_1$ is at least $\sqrt{(1-2\theta)\|\A\|_F}/(c_1 d)^{1/4}= c\sqrt{\|\A\|_F}/d^{1/4}$ for some constant $c = \sqrt{1-2\theta}/c_1^{1/4}$.

Without loss of generality, assume that $\|\u''\|_1 \geq c\sqrt{\|\A\|_F}/d^{1/4}$. Next we shall argue that we can drop large coordinates from $\u''$ by affecting $\|\u''\|_1$ by at most a constant factor. To see this, let $I = \{i: |\u''_i|\geq \kappa\}$ for some $\kappa$ to be determined later. It follows that $|I| \leq 1/\kappa^2$ and
\[
\|\u''_I\|_1 \leq \sqrt{|I|}\ \|\u''_I\|_2 \leq \frac{1}{\kappa} = \frac{c}{2} \frac{\sqrt{\|\A\|_F}}{d^{1/4}},
\]
provided that
\[
\kappa = \frac{2d^{1/4}}{c\sqrt{\|\A\|_F}}.
\]

Let $\hat \x = \u'' - \u''_I$, we see that $\|\hat \x\|_1\geq \frac{1}{2}\|\u''\|_1$. For notational simplicity let $S = \operatorname{supp}(\hat \x)$.
%
%Let $\hat \u = \u'' - \u''_I$, we see that $\|\hat \u\|\geq \frac{1}{2}\|\u''\|_1$ and
%\[
%\frac{\max_i |\hat \u_i|}{\min_{i\in \operatorname{supp}(\hat \u)}|\hat \u_i|} \leq \frac{\kappa}{\theta/\sqrt{n}} \leq \frac{2}{\theta c}\cdot d^{\frac14}\sqrt{\frac{n}{\|\A\|_F}}\lesssim_{\theta,c_1} \left(\frac{d}{\epsilon}\right)^\frac{1}{4},
%\]
%where we used the fact that $\|\A\|_F^2\geq \epsilon n^2(1-\frac1d)$ in the last inequality. This implies that we can split $\hat \u$ into $s = \Theta(\log(d/\epsilon))$ bands as $\hat \u = \hat \u_1 + \cdots + \hat \u_s$ such that in each band $\hat \u_i$ the nonzero coordinates' magnitudes are within a factor of $2$ from each other. There exists $i^\ast$ such that
%\[
%\|\hat \u_{i^\ast}\|_1 \geq \frac{\|\hat \u\|_1}{s}\geq \frac{1}{2s}\|\u''\|_1\geq \frac{c}{2s}\frac{\sqrt{\|\A\|_F}}{d^{1/4}}.
%\]
%
%For notational simplicity we denote $\hat\x = \hat \u_{i^\ast}$ and $S = \operatorname{supp}(\hat \u_j)$ in the rest of the proof.

Suppose that $\A$ is $\epsilon/d$-far from being stable rank at most $d$, and we reorder the rows of $\A$ such that $|\langle \A_{1,:},\u\rangle| \leq |\langle \A_{2,:},\u\rangle| \leq \cdots \leq |\langle \A_{n,:},\u\rangle|$. Let $m=\frac{\epsilon n^2}{d |S|}$ (we shall verify that $m\leq n$ later). For $i=1,\dots, m$, change $\A_{i,j}$ to $\operatorname{sgn}(\hat\x_j)$ for all $j\in S$ if $\langle \A_{i,S^c}, \u\rangle\geq 0$, and change $\A_{i,j}$ to $-\operatorname{sgn}(\hat\x_j)$ for all $j\in S$ if $\langle \A_{i,S^c}, \u\rangle < 0$, yielding a matrix $\B$ and we know that $\operatorname{srank}(\B) > d$.

%Now we verify that $m\leq n$ so that the aforementioned change is valid. It is clear that $|S|\geq \|\hat \x\|_1/\kappa$, and so
%\[
%m\leq \frac{\epsilon n^2}{d\cdot \|\hat\x\|_1/\kappa} \leq \frac{\epsilon n^2}{d}\cdot \frac{4 s \sqrt{d}}{ c^2 \|\A\|_F} \leq \frac{4\epsilon s}{c^2\sqrt d} n < n,
%\]
%provided that
%\[
%\sqrt{d}\geq C \epsilon\log\frac{d}{\epsilon}.
%\]
%for some absolute constant $C$ large enough, which is satisfied whenever $\epsilon\leq \epsilon_0$ for some absolute constant $\epsilon_0$ small enough.

Now we verify that $m\leq n$ so that the aforementioned change is valid. It is clear that $|S|\geq \|\hat \x\|_1/\kappa$, and so
\[
m\leq \frac{\epsilon n^2}{d\cdot \|\hat\x\|_1/\kappa} \leq \frac{\epsilon n^2}{d}\cdot \frac{4 \sqrt{d}}{ c^2 \|\A\|_F} \leq \frac{8\sqrt{\epsilon}}{c^2\sqrt d} n < n,
\]
provided that $\epsilon\leq \epsilon_0$ for some absolute constant $\epsilon_0$ small enough.

We observe that
\begin{equation}
\label{equ: F norm of B}
\|\B\|_F^2\le \|\A\|_F^2 + m|S|,
\end{equation}
and
\begin{equation}
\label{equ: operator norm of B}
\begin{split}
\|\B\|^2  \ge \|\B\u\|_2^2 &\geq \sum_{i=m+1}^n \langle\A_{i,:},\u\rangle^2+ m\|\hat\x\|_1^2\\
&\geq \left(1-\frac{m}{n}\right)\|\A\u\|_2^2 + m\|\hat\x\|_1^2\\
&= \left(1-\frac{m}{n}\right)\|\A\|^2 + m\|\hat\x\|_1^2.
\end{split}
\end{equation}
It follows from $\srank(\B) > d$ that
\[
d < \srank(\B)=\frac{\|\B\|_F^2}{\|\B\|^2} \le \frac{\|\A\|_F^2 + m|S|}{(1-\frac{m}{n})\|\A\|^2+ m\|\hat \x\|_1^2},
\]
or,
\begin{equation}\label{eqn:B_stable_rank}
d\left(1-\frac{m}{n}\right)\|\A\|^2 < \|\A\|_F^2 \left( 1 - \frac{m (d\|\hat\x\|_1^2-|S|)}{\|\A\|_F^2}\right).
\end{equation}

Next we claim that it holds under certain assumptions
\begin{equation}\label{eqn:subclaim1}
d \|\hat \x\|_1^2 \geq \frac{1}{\eta_1}|S|.
\end{equation}
Observe that
\begin{equation}\label{eqn:subclaim1_finer}
\frac{d \|\hat \x\|_1^2}{|S|} = d \|\hat \x\|_1 \frac{\|\hat \x\|_1}{|S|}\geq d\cdot\frac{\|\u''\|_1}{2}\cdot \frac{\theta}{\sqrt n} \geq d^{\frac{3}{4}}c\theta\sqrt{\frac{\|\A\|_F}{n}},
\end{equation}
which is at least $1/\eta_1$, provided that
\begin{equation}\label{eqn:condition_on_|A|_F}
\|\A\|_F \geq \frac{n}{d^{\frac{3}{2}}c^2\theta^2\eta_1^2}.
\end{equation}
This holds when $d=\Omega((1/\epsilon)^{1/3})$ since we know that $\|\A\|_F^2=\Omega(\epsilon n^2)$.

Hence under the assumption~\eqref{eqn:condition_on_|A|_F} it follows from~\eqref{eqn:B_stable_rank} that
\begin{equation}\label{eqn:B_stable_rank_2}
d\left(1-\frac{m}{n}\right)\|\A\|^2 < \|\A\|_F^2 \left( 1 - (1-\eta_1)\frac{m d\|\hat\x\|_1^2}{\|\A\|_F^2}\right).
\end{equation}
Note that
\[
(1-\eta_1)\frac{m d\|\hat\x\|_1^2}{\|\A\|_F^2} \geq (1-\eta_1)m d\cdot \frac{\frac{c^2}{4}\cdot \frac{\|\A\|_F}{\sqrt d}}{\|\A\|_F^2} = \frac{(1-\eta_1)c^2}{4}m \cdot \frac{\sqrt{d}}{\|\A\|_F} \geq \frac{1}{\eta_2}\cdot \frac{m}{n},
\]
provided that
\begin{equation}\label{eqn:condition2_on_|A|_F}
\|\A\|_F \leq \frac{\eta_2(1-\eta_1)c^2}{4}n\sqrt{d}.
\end{equation}
Combining \eqref{eqn:condition_on_|A|_F} and \eqref{eqn:condition2_on_|A|_F} leads to that
\begin{equation}\label{eqn:condition_on_d}
d \geq \frac{2}{c^2\theta\eta_1\sqrt{(1-\eta_1)\eta_2}} = \frac{2\sqrt{c_1}}{\theta(1-2\theta)\eta_1\sqrt{(1-\eta_1)\eta_2}}.
\end{equation}

Now, under both assumptions~\eqref{eqn:condition_on_|A|_F} and~\eqref{eqn:condition_on_d}, it follows from~\eqref{eqn:B_stable_rank_2} that
\begin{align*}
\frac{\|\A\|_F^2}{d\|\A\|^2} &\geq 1 + (1-\eta_1-\eta_2)\frac{m d\|\hat\x\|_1^2}{\|\A\|_F^2}
\\
&\geq 1 + (1-\eta_1-\eta_2)\frac{\epsilon n^2}{d\|\A\|_F^2}\cdot \frac{d\|\hat\x\|_1^2}{|S|}\\
&\geq 1 + (1-\eta_1-\eta_2)c\theta \frac{\epsilon n^{3/2}}{d^{1/4}\|\A\|_F^{3/2}},\qquad (\text{by }\eqref{eqn:subclaim1_finer})
\end{align*}
Choosing $\theta=1/4$, $\eta_1+\eta_2<1$, we see from \eqref{eqn:condition_on_d} that we shall need $d = \Omega(\sqrt{c_1})$. It is also easy to verify that \eqref{eqn:condition2_on_|A|_F} is satisfied for such $d$. Overall, we see that we shall need $\tau=\Theta\left(\frac{\epsilon n^{3/2}}{d^{1/4}\|\A\|_F^{3/2}}\right)$ in \eqref{equ: F norm entry level}.

\medskip
We are now ready to prove Theorem \ref{theorem: stable rank testing upper bound}.

\noindent{\textbf{Result (a):}} In fact, We have an accurate estimator to estimate the stable rank by reading an $\cO(d^{1.5}\log (n)/\epsilon^2)\times \cO(d^{1.5}\log (n)/\epsilon^2)$ submatrix:
combining Theorem \ref{theorem: operator norm estimator by Ismail} with Eqn. \eqref{equ: F norm entry level} yields an accurate estimator of the stable rank of $\A$:
\begin{equation*}
\label{equ: stable rank estimator}
(1-\Theta(\tau))\cdot\srank(\A)\le \frac{X}{\|\widetilde\A\|^2}\le (1+\Theta(\tau))\cdot\srank(\A).
\end{equation*}
Setting $\tau$ as $\Theta(\frac{\epsilon}{d^{1/4}})$ gives the claimed result immediately.

\medskip
\noindent{\textbf{Result (b):}} It follows from setting $\tau = \Theta(\frac{\epsilon}{d^{1/4}})$ in Theorem \ref{theorem: operator norm estimator under sensing model} on sketching complexity.
\end{proof}

\subsection{Lower Bounds}

\begin{lemma}[Corollary 5.35, \cite{vershynin2010introduction}]
\label{lemma: vershynin concentration of Gaussian operator norm}
Let $\A$ be an $m\times n$ ($m>n$) matrix whose entries are independent standard normal random variables. Then for every $t\ge 0$ and fixed $\v\in\mathbb{R}^n$, it holds with probability at least $1-2\exp(-t^2/2)$ that
\[
\sqrt{m}-\sqrt{n}-t\le \sigma_{\min}(\A) \leq \sigma_{\max}(\A) \le \sqrt{m}+\sqrt{n}+t.
\]
\end{lemma}

\begin{lemma}[Lemma 1, \cite{laurent2000adaptive}]
\label{lemma: tail bound of chi-square}
Let $X\sim \chi^2(k)$. Then we have the tail bound
\begin{equation*}
\Pr[k-2\sqrt{kx}\le X\le k+2\sqrt{kx}+2x] \ge 1-2e^{-x}.
\end{equation*}
\end{lemma}

\begin{lemma}[Theorem 4, \cite{li2016tight}]
\label{lemma: total variation distance of L1 and L2}
Let $\u_1,...,\u_r$ be i.i.d. $\cN(\0,\I_m)$ vectors and $\v_1,...,\v_r$ be i.i.d. $\cN(\0,\I_n)$ vectors and further suppose that $\{\u_i\}$ and $\{\v_i\}$ are independent. Let $\cD_1=\cG(m,n)$ and $\cD_2=\cG(m,n)+\sum_{i=1}^r s_i\u_i\v_i^\top$, where $\s=[s_1,...,s_r]^\top$ and $\cG(m,n)$ represents $m\times n$ i.i.d. standard Gaussian matrix over $\R$. Denote by $\cL_1$ and $\cL_2$ the corresponding distribution of the linear sketch of size $k$ on $\cD_1$ and $\cD_2$. Then
there exists an absolute constant $c>0$ such that $d_{TV}(\cL_1,\cL_2)\le 1/10$ whenever $k\le c/\|\s\|_2^4$, where $d_{TV}(\cdot,\cdot)$ represents the total variation distance between two distributions.
\end{lemma}

\begin{theorem}
\label{theorem: stable rank lower bound}
Let $\epsilon\in (0,1/3)$ and let $d\ge 4$. For $\A\in\R^{(d/\epsilon^2)\times d}$, any algorithm that distinguishes ``$\srank(\A)\le d_0$'' from ``$\A$ being $\epsilon_0/d_0$-far from stable rank $\le d_0$'' with error probability at most $1/6$ requires measurements $\Omega(d^2/(\epsilon^2\log(d/\epsilon)))$ for any linear sketch, where $d_0=\frac{d}{1+\Theta(\epsilon)}$ and $\epsilon_0=\Theta(\frac{\epsilon}{\log^2(d/\epsilon)})$.
\end{theorem}

\begin{remark}
Theorem \ref{theorem: stable rank lower bound} can be generalized to the $(d/\epsilon^2)\times (d/\epsilon^2)$ matrix by concatenating the columns of two hard instances in Theorem \ref{theorem: stable rank lower bound} $(1/\epsilon^2)$ times. This scales up all singular values in our $(d/\epsilon^2)\times d$ hard instances by a factor of $1/\epsilon$, and thus the stable rank remains the same. Observe that the bounds on $\|\G\|$, $\|\G\|_F$ and $\|\S\|_F$ in \eqref{eqn:upper_bound_perturbed} and \eqref{eqn:lower_bound_perturbed} in the proof below are also scaled up by $1/\epsilon$. The concatenated matrix is therefore $\epsilon_0/d_0$-far from having stable rank at most $d_0$ following the same argument.
\end{remark}

\begin{proof}
Let $m=d/\epsilon^2$ and $n=d$. We will apply Lemma \ref{lemma: total variation distance of L1 and L2} with $r=1$. To this end, we need to justify that
\begin{equation*}
\frac{C}{\log(d/\epsilon)}\G \qquad \text{and} \qquad \frac{C}{\log(d/\epsilon)}(\G_0+s_1\u\v^\top)
\end{equation*}
differ in the stable rank (i.e., $\srank(\G)> d_0\ge \srank(\G_0+s_1\u\v^\top)$) and that $\G$ is rigid (i.e., changing $\epsilon_0/d_0$-fraction of entries of $\G$ would not change the stable rank of $\G$ to be less than $d_0$), where the multiplicative factor $C/\log(d/\epsilon)$ is to keep the maximum absolute value of entries in the two hard instances less than $1$, $\G,\G_0\sim\cG(m,n)$, $\u\sim\cN(\0,\I_m)$, $\v\sim\cN(\0,\I_n)$, and $s_1=3\sqrt{\epsilon/d}$. Note that by Lemma \ref{lemma: total variation distance of L1 and L2}, we cannot distinguish $\G$ from $\G':=\G_0+s_1\u\v^\top$ with $\Omega(d^2/\epsilon^2)$ samples. So if the stable ranks of $\G$ and $\G'$ have a gap, we cannot detect the gap either.

For the operator norm, on one hand, it follows from Lemma \ref{lemma: vershynin concentration of Gaussian operator norm} that with probability at least $1-2\exp(-d/2)$,
\begin{equation*}
(1-1.1\epsilon)\frac{\sqrt{d}}{\epsilon}\le \sigma_{\min}(\G) \leq \sigma_{\max}(\G)\le (1+1.1\epsilon)\frac{\sqrt{d}}{\epsilon},
\end{equation*}
for an absolute constant $C_0>0$.
On the other hand, with probability $\ge 1-\exp(-\Omega(d))$ we have
\begin{equation*}
\begin{split}
\left\|\G_0+s_1\u\v^\top\right\|^2 &=\sup_{\x\in \mathbb{S}^{n-1}} \left\|\G_0\x+s_1\u\v^\top\x\right\|_2^2\\
&\ge \frac{\left\|\G_0\v+s_1\u\v^\top\v\right\|_2^2}{\|\v\|_2^2}\\
&=\frac{\|\G_0\v\|_2^2}{\|\v\|_2^2}+s_1^2\|\u\|_2^2\|\v\|_2^2+2\langle\G_0\v,s_1\u\rangle\\
&\ge \left((1-1.1\epsilon)\sqrt{\frac{d}{\epsilon^2}}\right)^2+0.9^2 s_1^2 \frac{d^2}{\epsilon^2}-\cO\left(\frac{d}{\sqrt{\epsilon}}\right)\\
&\ge \left((1-1.1\epsilon)\sqrt{\frac{d}{\epsilon^2}}\right)^2+0.9^2 s_1^2 \frac{d^2}{\epsilon^2}-\cO\left(\frac{d}{\sqrt{\epsilon}}\right)\\
&\ge ((1-1.1\epsilon)^2+7.29\epsilon)\frac{d}{\epsilon^2}-\cO\left(\frac{d}{\sqrt{\epsilon}}\right)\\
&\ge ((1-1.1\epsilon)^2+7.29\epsilon)\frac{d}{\epsilon^2}-\cO\left(\frac{d}{\sqrt{\epsilon}}\right)\\
&\ge (1+2\epsilon)^2\frac{d}{\epsilon^2},
\end{split}
\end{equation*}
where the second inequality (line 4) follows from the concentration of the quadratic form (see, e.g., \cite{RV:hanson-wright})
\begin{equation}\label{equ: inner product}
\Pr_{\u,\v}\{|\v^T \G_0 \u|>t\}\leq 2\exp\left(-c\min\left\{ \frac{t}{\|\G_0\|}, \frac{t^2}{\|\G_0\|_F^2} \right\}\right)
\end{equation}
for fixed $\G_0$; since $\|\G_0\|\simeq \sqrt{d}/\epsilon$ and $\|\G_0\|_F^2\simeq d^2/\epsilon^2$ with high probability, we can take $t = \Theta(d^{3/2}/\epsilon)$.
%where the second inequality holds due to the Cheybshev inequality and the fact
%\begin{equation}
%\label{equ: inner product}
%\begin{split}
%\mathbb{E}|\langle \G_0,\u\v^\top\rangle|^2&=\mathbb{E}|\u^\top\G_0\v|^2\\
%&=\mathbb{E} \left(\sum_{i,j}{\G_0}_{ij}\u_i\v_j\right)^2\\
%&=\mathbb{E} \sum_{i,j}\sum_{i',j'} {\G_0}_{ij}{\G_0}_{i'j'}\u_i\u_{i'}\v_j\v_{j'}\\
%&=\mathbb{E} \sum_{i,j} {\G_0}_{ij}^2\u_i^2\v_j^2\\
%&=  \sum_{i,j} \mathbb{E} {\G_0}_{ij}^2\ \mathbb{E}\u_i^2\ \mathbb{E}\v_j^2\\
%&= \frac{d^2}{\epsilon^2}.
%\end{split}
%\end{equation}
For the Frobenius norm, we note that
\begin{equation*}
\begin{split}
\left\|\G_0+3\sqrt{\frac{\epsilon}{d}}\u\v^\top\right\|_F^2=\|\G_0\|_F^2+9\frac{\epsilon}{d}\|\u\v^\top\|_F^2+6\sqrt{\frac{\epsilon}{d}}\langle \G_0,\u\v^\top\rangle.
\end{split}
\end{equation*}
Observe that $\|\G_0\|_F^2\sim \chi^2(\frac{d^2}{\epsilon^2})$ so $\|\G_0\|_F^2=(1\pm \Theta(\frac{\epsilon}{d}))\frac{d^2}{\epsilon^2}$ with probability $\ge 0.9$ by Lemma \ref{lemma: tail bound of chi-square}, and $9\frac{\epsilon}{d}\|\u\v^\top\|_F^2=9\frac{\epsilon}{d}\|\u\|_2^2\|\v\|_2^2=\Theta(\frac{d}{\epsilon})$ with high probability. And also, setting $t = \Theta(d/\epsilon)$ in \eqref{equ: inner product}, we have with probability at least $0.9$ that $|\langle \G_0,\u\v^\top\rangle| = \cO(\frac{d}{\epsilon})$ and thus $6\sqrt{\frac{\epsilon}{d}}|\langle\G_0,\u\v^\top\rangle|=\cO(\sqrt{\frac{d}{\epsilon}})$.
Therefore,
\begin{equation*}
\left(1-\Theta\left(\frac{\epsilon}{d}\right)\right)\|\G_0\|_F^2\le \left\|\G_0+3\sqrt{\frac{\epsilon}{d}}\u\v^\top\right\|_F^2\le \left(1+\Theta\left(\frac{\epsilon}{d}\right)\right)\|\G_0\|_F^2.
\end{equation*}
As a result,
\begin{equation*}
\srank(\G)=\frac{\|\G\|_F^2}{\|\G\|^2}\ge \frac{d}{(1+1.2\epsilon)^2},
\end{equation*}
and
\begin{equation*}
\srank(\G')=\frac{\|\G'\|_F^2}{\|\G'\|^2}\le \frac{d}{(1+1.9\epsilon)^2}.
\end{equation*}
%namely,
%\begin{equation*}
%(1-\epsilon)\cdot\srank(\G)>\srank(\G').
%\end{equation*}
By Lemma \ref{lemma: total variation distance of L1 and L2}, it is therefore hard to distinguish
\begin{equation*}
\srank\left(\frac{C}{\log(d/\epsilon)}\G'\right)=\srank(\G')=\frac{\|\G'\|_F^2}{\|\G'\|^2}\le \frac{d}{(1+1.9\epsilon)^2}\triangleq d_0
\end{equation*}
from
\begin{equation}
\label{equ: stable rank of G}
\srank\left(\frac{C}{\log(d/\epsilon)}\G\right)=\srank(\G)=\frac{\|\G\|_F^2}{\|\G\|^2}\ge \frac{d}{(1+1.2\epsilon)^2}\geq (1+1.3\epsilon)d_0,
\end{equation}
with sample size $\cO(d^2/\epsilon^2)$, provided that $\epsilon$ is sufficiently small.

We now show that $\frac{C}{\log(d/\epsilon)}\G$ is rigid, i.e., changing $\epsilon_0/d_0$-fraction of entries of $\frac{C}{\log(d/\epsilon)}\G$ will not make $\srank\left(\frac{C}{\log(\frac{d}{\epsilon})}\G\right)\le d_0$. For any $\S\in\R^{(d/\epsilon^2)\times d}$ such that $\|\S\|_0=\frac{\theta d}{\epsilon\log^2(\frac{d}{\epsilon})}$ (where $0<\theta<1$) and $\left\|\frac{C}{\log(\frac{d}{\epsilon})}\G+\S\right\|_\infty\le 1$ (thus $\S$ has an $\frac{\epsilon_0}{d_0}$-fraction of non-zero entries and $\|\S\|_\infty\le 2$), we have
\begin{equation*}
\begin{split}
&\quad\ \left\|\frac{C}{\log(d/\epsilon)}\G+\S\right\|^2\\
&=\sup_{\|\u\|_2=1,\|\v\|_2=1} \left\langle\left(\frac{C}{\log(d/\epsilon)}\G+\S\right)\u,\v\right\rangle^2\\
&\le \frac{C^2}{\log^2(d/\epsilon)}\sup_{\substack{\|\u\|_2=1\\ \|\v\|_2=1}} \langle \G\u,\v\rangle^2+ \sup_{\substack{\|\u\|_2=1\\ \|\v\|_2=1}} \langle\S\u,\v\rangle^2+\frac{2C}{\log(d/\epsilon)}\sup_{\substack{\|\u\|_2=1\\ \|\v\|_2=1}} \langle\G^\top\S\u,\v\rangle\\
&= \frac{C^2}{\log^2(d/\epsilon)}\|\G\|^2+\|\S\|^2+\frac{2C}{\log(d/\epsilon)}\|\G^\top\S\|\\
&\le \frac{C^2}{\log^2(d/\epsilon)}\|\G\|^2+\|\S\|^2+\frac{2C}{\log(d/\epsilon)}\|\G^\top\| \|\S\|\\
&\le \frac{C^2}{\log^2(d/\epsilon)}\|\G\|^2+\|\S\|_F^2+\frac{2C}{\log(d/\epsilon)}\|\G^\top\| \|\S\|_F
\end{split}
\end{equation*}
Now, observe that
\[
\|\S\|_F^2\le \frac{4 \theta d}{\epsilon\log^2(d/\epsilon)}
\]
and that, by setting $t = O(\sqrt{d}/\epsilon)$ in Lemma~\ref{lemma: vershynin concentration of Gaussian operator norm},
\[
\|\G\| = \cO\left(\frac{\sqrt{d}}{\epsilon}\right)
\]
with probability at least $1-2\exp(-\Omega(d/\epsilon^2))$, we have that
\begin{equation}\label{eqn:upper_bound_perturbed}
\left\|\frac{C}{\log(d/\epsilon)}\G+\S\right\|^2\le \frac{C^2}{\log^2(d/\epsilon)}\|\G\|^2+\cO\left(\frac{(\theta+\sqrt{\theta})d}{\epsilon\log^2(d/\epsilon)}\right)\le \left(1+c_1\sqrt{\theta}\epsilon\right)\frac{C^2}{\log^2(d/\epsilon)}\|\G\|^2,
\end{equation}
where $c_1 > 0$ is an absolute constant.

On the other hand, with probability at least $1-c'\exp(-\Omega(d/\epsilon))$ it holds that
\begin{equation}\label{eqn:lower_bound_perturbed}
\begin{aligned}
\left\|\frac{C}{\log(d/\epsilon)}\G+\S\right\|_F^2
&\ge \left(\frac{C}{\log(d/\epsilon)}\|\G\|_F - \|\S\|_F\right)^2\\
&\ge \left(1-c_2\sqrt{\frac{\epsilon \theta}{d}}\right)\frac{C^2}{\log^2(d/\epsilon)}\|\G\|_F^2,
\end{aligned}
\end{equation}
where $c_2 > 0$ is an absolute constant.

Therefore,
\begin{equation*}
\srank\left(\frac{C}{\log(d/\epsilon)}\G+\S\right)\ge \left(1-c_3\sqrt{\theta}\epsilon\right)\srank\left(\frac{C}{\log(d/\epsilon)}\G\right)>d_0,
\end{equation*}
where $c_3 > 0$ is an absolute constant, and $\theta$ is small enough such that the last inequality holds.

We conclude that $\frac{C}{\log(d/\epsilon)}\G$ is $\epsilon_0/d_0$-far from having stable rank $\le d_0$. The proof is complete.
\end{proof}

\section{Non-Adaptive Testing of Schatten-$p$ Norm}

We study the problem of testing Schatten-$p$ norms in this section.

\subsection{Upper Bounds}

\begin{problem}[Schatten-$p$ Norm Testing in the Bounded Entry Model for $p>2$]
Let $p>2$ and $\A\in\mathbb{R}^{n\times n}$ be a matrix such that $\|\A\|_\infty\leq 1$. For an absolute constant $c$, the matrix $\A$ satisfies one of the promised properties:
\begin{itemize}
\item[$\mathsf{H0.}$]
$\|\A\|_{\cS_p}^p\geq cn^{p}$;
\item[$\mathsf{H1.}$]
$\A$ is $\epsilon$-far from $\|\A\|_{\cS_p}^p\geq cn^{p}$, meaning that it requires changing at least an $\epsilon$-fraction of the entries of $\A$ such that that $\|\A\|_{\cS_p}^p\geq cn^{p}$.
\end{itemize}
The problem is to design a non-adaptive property testing algorithm that outputs $\mathsf{H0}$ with probability at least $0.9$ if $\A\in \mathsf{H0}$, and output $\mathsf{H1}$ with probability at least $0.99$ if $\A\in\mathsf{H1}$, with the least number of queried entries.
\end{problem}

First we prove a lemma showing that $\mathsf{H0}$ and $\mathsf{H1}$ can be distinguished by the Schatten $p$-norm.
\begin{lemma}\label{lem:schatten_norm_gap}
Suppose that $p>2$ is a constant. There exist constants $c = c(p)$, $C = C(p)$ and $\epsilon_0 = \epsilon(p)$ such that for any $\epsilon \in [C/n, \epsilon_0]$, when $\A\in \mathsf{H1}$ it holds that $\|\A\|_{\cS_p}^p\leq (c-c'\epsilon)n^p$ for some small constant $c'$ that may depend on $p$.
\end{lemma}
\begin{proof}
Assume that $\|\A\|_{\cS_p}^p \geq c_1 n^p$ for some constant $c_1 < c$, otherwise there is already a constant-factor gap. Together with the assumption that $\|\A\|_\infty\leq 1$ and thus $\|\A\|_F^2\leq n^2$, it must hold that $\|\A\| \geq c_2 n$ for $c_2 = c_1^{1/(p-1)}$.

We claim that we can find a set $T$ of $\epsilon n$ rows such that $\|\A_{T^c,:}\|_{\cS_p}^p \geq (1-C'\epsilon)\|\A\|_{\cS_p}^p$ for some $C'$ (which may depend on $p$) and therefore $\|\A_{T^c,:}\| \geq c_2' n$, where $\A_{T^c,:}$ stands for a submatrix of $\A$ with rows restricted on the set $T^c$. Consider a random subset $T$ formed by including rows independently with probability $\epsilon$, that is, let $\delta_i$ be the indicator variable whether $i\in T$ and $\bbE \delta_i = \epsilon$. Denote by $\A_{i,:}\in\R^{1\times n}$ the $i$-th row of $\A$. Then we have, by the standard symmetrization trick (see, e.g.,~\cite[Lemma 6.3]{LT91}), that
\begin{align*}
\bbE_{\delta_i}\|\A_{T,:}\|_{\cS_p}^2 = \bbE_{\delta_i}\left\|\sum_i \delta_i \A_{i,:}^\top \A_{i,:}\right\|_{\cS_{p/2}}  &\le  \bbE_{\delta_i}\left(\left\|\sum_i (\delta_i-\epsilon) \A_{i,:}^\top \A_{i,:}\right\|_{\cS_{p/2}} + \left\| \sum_i \epsilon \A_{i,:}^\top\A_{i,:} \right\|_{\cS_{p/2}} \right)\\
&\leq 2\bbE_{\delta_i}\bbE_{\epsilon_i}\left\| \sum_i \epsilon_i\delta_i \A_{i,:}^\top\A_{i,:} \right\|_{\cS_{p/2}} + \epsilon\|\A\|_{\cS_p}^2,
\end{align*}
where $\epsilon_i$'s are i.i.d.\ $\{\pm 1\}$-valued Rademacher variables with $\Pr(\epsilon_i=+1)=\Pr(\epsilon_i=-1)=1/2$.
%\begin{equation*}
%\begin{split}
%\bbE&\left\|\sum_i (\delta_i-\epsilon) \A_{i,:}^\top \A_{i,:}\right\|_{\cS_{p/2}}=\bbE\left\|\sum_i (\delta_i-\bbE \delta_i) \A_{i,:}^\top \A_{i,:}\right\|_{\cS_{p/2}}\\
%&=\bbE\left\|\sum_i \bbE[(\delta_i- \delta_i')\mid \delta_i] \A_{i,:}^\top \A_{i,:}\right\|_{\cS_{p/2}}\quad (\text{$\delta_i'$ is an independent copy of $\delta_i$})\\
%&\le \bbE\bbE\left[\left\|\sum_i (\delta_i- \delta_i') \A_{i,:}^\top \A_{i,:}\right\|_{\cS_{p/2}} \Bigg| \{\delta_i\}_{i=1}^n\right]\quad (\text{by Jensen's inequality})\\
%&= \bbE\left\|\sum_i (\delta_i- \delta_i') \A_{i,:}^\top \A_{i,:}\right\|_{\cS_{p/2}}\\
%&= \bbE\left\|\sum_i \epsilon_i(\delta_i- \delta_i') \A_{i,:}^\top \A_{i,:}\right\|_{\cS_{p/2}}\quad (\text{because $\epsilon_i$ is the Rademacher variable})\\
%&\le \bbE\left\|\sum_i \epsilon_i\delta_i \A_{i,:}^\top \A_{i,:}\right\|_{\cS_{p/2}}+\bbE\left\|\sum_i \epsilon_i\delta_i' \A_{i,:}^\top \A_{i,:}\right\|_{\cS_{p/2}}\quad (\text{by triangle inequality})\\
%&= 2\bbE\left\|\sum_i \epsilon_i\delta_i \A_{i,:}^\top \A_{i,:}\right\|_{\cS_{p/2}}.
%\end{split}
%\end{equation*}
Applying the Non-Commutative Khintchine Inequality (abbreviated as NCKI)~\cite{LP86} yields that
\begin{align*}
\bbE_{\epsilon_i}\left\| \sum_i \epsilon_i\delta_i \A_{i,:}^\top\A_{i,:} \right\|_{\cS_{p/2}}
&\leq \left(\bbE_{\epsilon_i}\left\| \sum_i \epsilon_i\delta_i \A_{i,:}^\top\A_{i,:} \right\|_{\cS_{p/2}}^{p/2}\right)^{2/p}\quad (\text{by Jensen's inequality})\\
&\leq C_1\sqrt{\frac{p}{2}} \left\|\left(\sum_i \delta_i (\A_{i,:}^\top\A_{i,:})^2\right)^{\frac12}\right\|_{\cS_{p/2}}\quad (\text{by NCKI})\\
&\leq C_1\sqrt{\frac{p}{2}} \left\|\max_i \|\A_{i,:}\|_2 \cdot \left(\sum_i \delta_i\A_{i,:}^\top\A_{i,:}\right)^{\frac12}\right\|_{\cS_{p/2}}\\
&\leq C_1\sqrt{\frac{p}{2}} \sqrt{n} \|\A_{T,:}\|_{\cS_{p/2}}\\
&\leq C_1\sqrt{\frac{p}{2}} n^{\frac12}|T|^{\frac1p} \|\A_{T,:}\|_{\cS_p},\quad (\text{by H\"older's inequality})
\end{align*}
where the third inequality holds since $\sum_i \delta_i(\A_{i,:}^\top\A_{i,:})^2\preceq \max_i \|\A_{i,:}\|_2^2\cdot \sum_i\delta_i \A_{i,:}^\top\A_{i,:}$.
Hence, taking expectation on both sides w.r.t.\ $\delta_i$,
\begin{align*}
\bbE\|\A_{T,:}\|_{\cS_p}^2 &\leq C_1\sqrt{\frac{p}{2}}n^{\frac{1}{2}} (\bbE|T|^{\frac 2p})^{\frac{1}{2}} (\bbE\|\A_{T,:}\|_{\cS_p}^2)^{\frac 12} + \epsilon\|\A\|_{\cS_p}^2 \quad (\text{by Cauchy-Schwarz inequality})\\
&\leq C_1\sqrt{\frac{p}{2}}n^{\frac{1}{2}} (\bbE|T|)^{\frac{1}{p}} (\bbE\|\A_{T,:}\|_{\cS_p}^2)^{\frac 12} + \epsilon\|\A\|_{\cS_p}^2\quad (\text{by Jensen's inequality})\\
&\leq C_1\sqrt{\frac{p}{2}} n^{\frac12}(\epsilon n)^{\frac 1p} (\bbE\|\A_{T,:}\|_{\cS_p}^2)^{\frac 12} + \epsilon\|\A\|_{\cS_p}^2,
\end{align*}
whence we can solve that
\[
\bbE\|\A_{T,:}\|_{\cS_p}^2 \leq C_1^2 \frac{p}{2}\epsilon^{\frac2p}n^{1+\frac2p} + 4\epsilon\|\A\|_{\cS_p}^2 \leq C_2\epsilon\|\A\|_{\cS_p}^2.
\]
That is, we can find $T$ such that $\|\A_{T,:}\|_{\cS_p}^2\leq C_2\epsilon\|\A\|_{\cS_p}^2$ and thus
\[
\|\A_{T^c,:}\|_{\cS_p}^2 = \|\A_{T^c,:}^\top\A_{T^c,:}\|_{\cS_{p/2}} = \|\A^\top\A - \A_{T,:}^\top\A_{T,:}\|_{\cS_{p/2}}\geq (1-C_2\epsilon)\|\A^\top\A\|_{\cS_{p/2}} = (1-C_2\epsilon)\|\A\|_{\cS_p}^2
\]
as desired. A Chernoff bound shows that $|T|\geq 0.9\epsilon n$ with at least a high constant probability. We can assume that it happens, since conditioning on this event will increase $\bbE\|\A_{T,:}\|_{\cS_p}^2$ just by a constant factor. When $|T| > \epsilon n$ we can just remove rows from $T$, which only decreases $\|\A_{T,:}\|_{\cS_p}$.

Let $\v$ be the normalized eigenvector associated with the largest eigenvalue of $\A_{T^c,:}^\top \A_{T^c,:}$. We shall change the rows of $T$ all to $\tilde\v := \operatorname{sgn}(\v)$, obtaining a matrix $\B$. Note that $\B^\top \B = \A_{T^c,:}^\top \A_{T^c,:} + |T| \v\v^\top\succeq \A_{T^c,:}^\top \A_{T^c,:}$ is a rank-$1$ PSD perturbation of $\A_{T^c,:}^\top \A_{T^c,:}$ and $\v$ is the leading eigenvector of $\A_{T^c,:}^\top \A_{T^c,:}$, we have that for the $i$-th eigenvalue $\lambda_i(\cdot)$,
\[
\lambda_i(\B^\top\B) \geq \lambda_i(\A_{T^c,:}^\top \A_{T^c,:}),\quad i\geq 2.
\]
and the largest eigenvalue
\begin{align*}
\lambda_1(\B^\top\B) &= \sup_{\x: \|\x\|_2 = 1} \x^\top (\A_{T^c,:}^\top \A_{T^c,:} + \epsilon n \tilde{\v} \tilde{\v}^\top)\x \\
&\geq \v^\top (\A_{T^c,:}^\top \A_{T^c,:} + |T| \tilde{\v} \tilde{\v}^\top)\v \\
&= \lambda_1(\A_{T^c,:}^\top \A_{T^c,:}) + |T| \|\v\|_1^2.
\end{align*}
Observe that
\[
\lambda_1(\A_{T^c,:}^\top \A_{T^c,:}) = \|\A_{T^c,:} \v\|_2^2 = \sum_{i\in T^c} \langle \A_{i,:},\v\rangle^2 \leq \sum_{i\in T^c} \|\A_{i,:}\|_\infty^2 \|\v\|_1^2 \leq (n-|T|) \|\v\|_1^2.
\]
Then
\[
\lambda_1(\B^\top\B) \geq \left(1+\frac{0.9\epsilon}{1-\epsilon}\right)\lambda_1(\A_{T^c,:}^\top \A_{T^c,:})
\]
and so
\begin{align*}
c n^p > \|\B\|_{\cS_p}^p &\geq \left(1+\frac{0.9\epsilon}{1-\epsilon}\right)^{\frac p2}\lambda_1^{\frac{p}{2}}(\A_{T^c,:}^\top \A_{T^c,:}) + \sum_{i\geq 2} \lambda_i^{\frac p2}(\A_{T^c,:}^\top \A_{T^c,:})\\
&\geq \left(\left(1+\frac{0.9\epsilon}{1-\epsilon}\right)^{\frac p2}-1\right)(c_2' n)^p + \|\A_{T^c,:}\|_{\cS_p}^p\\
&\geq c_3 p \epsilon n^p + (1-C'\epsilon)\|\A\|_{\cS_p}^p,
\end{align*}
whence it follows that
\[
\|\A\|_{\cS_p}^p \leq (c - (c_3 p - c C') \epsilon) n^p,
\]
provided that $c$ and $\epsilon$ are sufficiently small.
\end{proof}

\begin{algorithm}
\caption{Algorithm for Schatten-$p$ norm testing ($p>2$)}
\label{algorithm: algorithm for the Schatten-$p$ norm testing problem}
\begin{algorithmic}[1]
\LineComment Lines 1-2 estimate the Frobenius norm of $\A$.
\State Uniformly sample $q_0=\cO(\frac{1}{\epsilon^2})$ entries $\A$, forming vector $\y$.
\State $X\leftarrow \frac{n^2}{q_0}\|\y\|_2^2$.  \Comment{$X$ is an estimator of $\|\A\|_F^2$}.
\State Uniformly sample a $q\times q$ submatrix $\widetilde\A'$ with $q=\cO(\frac{\log n}{\epsilon})$.
\If {$\|\widetilde\A'\|\le C_0\sqrt{X}\frac{q}{n}$}
	\State Output ``$\A\in\mathsf{H1}$''.
\Else
	\State Run Algorithm \ref{algorithm: the sampling algorithm for even p} with $\tau=\Theta(\epsilon^{p/(p-2)}/p)$ and obtain indices $\cI_{\mathsf{row}}$ and $\cI_{\mathsf{col}}$. \label{alg: step 7}
	\State $\A_0\gets \A_{\cI_{\mathsf{row}},\cI_{\mathsf{col}}}$. \label{alg: step 8}
	\State $\mathcal{I}\gets \{i\mid \sigma_i(\A_0) > (1+\epsilon/(3p))n(c\epsilon/3)^{1/(p-2)}\}$. \label{alg: step 9}
	\If {$\sum_{i\in\cI}\sigma_i^p(\A_0)\ge cn^p$}  \label{alg: step 10}
		\State Output ``$\A\in\mathsf{H0}$''.
	\Else
		\State Output ``$\A\in\mathsf{H1}$''.
	\EndIf
\EndIf
\end{algorithmic}
\end{algorithm}

\begin{theorem}
\label{theorem: Schatten-p norm upper bound}
Let $p>2$ be a constant, and $c$ and $\epsilon$ be as in Lemma~\ref{lem:schatten_norm_gap}. Then Algorithm \ref{algorithm: algorithm for the Schatten-$p$ norm testing problem} is a correct algorithm for the Schatten-$p$ norm testing problem under the sampling model with probability at least $0.99$. It reads $\cO\left(\frac{\log^2 n}{\epsilon^{4p/(p-2)}}\right)$ entries.
\end{theorem}

\begin{proof}
%We discuss the case when $p>\log n$. When $\A\in \mathsf{H0}$, we claim that
%\begin{equation*}
%n\ge \|\A\|_F\ge \|\A\|_{\cS_p} \ge \|\A\|\ge \frac{1}{2e}n.
%\end{equation*}
%To see this, we note that
%\begin{equation*}
%1\le \frac{\|\A\|_{\cS_p}}{\|\A\|}=\left(1+\left(\frac{\sigma_2(\A)}{\sigma_1(\A)}\right)^p+...+\left(\frac{\sigma_n(\A)}{\sigma_1(\A)}\right)^p\right)^{1/p}\le n^{1/p}\le e,
%\end{equation*}
%and
%\begin{equation*}
%n\ge \|\A\|_F\ge \|\A\|_{\cS_p} \ge c^{1/p}n\ge\frac{1}{2}n.
%\end{equation*}
%Therefore, $\A\in \mathsf{H0}$ implies that $\srank(\A)\le 4e^2$. Namely, if $\srank(\A)> 4e^2$, then we can safely output ``$\A\in \mathsf{H1}$''. \hongyang{[why $\A\in\mathsf{H0}$ implies that $\srank(\A)=\mathsf{poly}(1/\epsilon)$ for arbitrary $p>2$? so I assume $p>\log n$ here.]}

When $\A\in \mathsf{H0}$, we claim that $\srank(\A)=\cO(1)$ which is independent of $n$ and $1/\epsilon$. Otherwise, suppose $\srank(\A)=f(n,1/\epsilon)$. Then $\|\A\|=\|\A\|_F/\sqrt{\srank(\A)}\le n/f(n,1/\epsilon)=o(n)$. In this case, $\|\A\|_{\cS_p}^p$ is maximized when the first $r$ singular values are equal to $\|\A\|$, where $r\le n^2/\|\A\|^2$ in order to satisfy $\|\A\|_F\le n$. So the maximal $\|\A\|_{\cS_p}^p$ is $r\|\A\|^p\le n^2\|\A\|^{p-2}=o(n^p)$, which leads to a contradiction with $\A\in\mathsf{H0}$. That is, $\srank(\A)$ is an absolute constant which is independent of $n$ and $1/\epsilon$, say $4e^2$. Thus, when $\A\in\mathsf{H0}$ we have $\|\A\|=\Theta(n)$ and $\|\A\|_F=\Theta(n)$, because
$
n\ge \|\A\|_F\ge \|\A\|_{\cS_p}\ge c^{1/p}n.
$

We note that by sampling $q_0$ entries from $\A$ and stacking them as vector $\y$, the resulting estimator $X=\frac{n^2}{q_0^2}\|\y\|_2^2$ satisfies $\mathbb{E}[X]=\|\A\|_F^2$ and $\mathsf{Var}[X]\le n^2(n^4/q_0^2)(q_0/n^2)=n^4/q_0$. Taking $q_0=\cO(1/\epsilon^2)$,  we have, by Chebyshev's inequality, that
\begin{equation*}
\Pr\left[|X-\|\A\|_F^2|>\epsilon n^2\right]\le \frac{n^4/q_0}{\epsilon^2 n^4}\le 0.999.
\end{equation*}
Thus with constant probability,
$
\left|X-\|\A\|_F^2\right| \le \epsilon n^2.
$

We argue that uniformly sampling an $\cO(\log (n)/\epsilon)\times \cO(\log(n)/\epsilon)$ submatrix of $\A$ suffices to distinguish $\srank(\A)\le 4e^2$ v.s. $\srank(\A)> 4c_1e^2$ for a large absolute constant $c_1$ with a constant probability. To see this, when $\srank(\A)> 4c_1e^2$, let $\U$ be a uniformly random $n\times n$ orthogonal matrix and let $\A_{\mathsf{row}}'$ be the matrix after uniform row sampling of $\A$ of expected cardinality $q$. Note that $\|\A_{\mathsf{row}}'\|=\|\A_{\mathsf{row}}'\U\|$, and $(\A\U)_{i,:}=\A_{i,:}\U$ is uniform on $\|\A_{i,:}\|_2\cdot\mathbb{S}^{n-1}$. So $\|\A_{i,:}\U\|_\infty^2\le 2\|\A_{i,:}\U\|_2^2\log (n)/n$ for any fixed $i$ with probability at least $1-1/n^2$ by Lemma \ref{Lemma: infty norm of uniform sampling from unit sphere}. Therefore, with probability at least $1-1/n$ by union bound over all rows, $\|\A\U\|_{\mathsf{col}}^2\le 2\|\A\|_F^2\log (n)/n$, where $\|\A\|_{\mathsf{col}}$ represents the maximum $\ell_2$ norm among all columns of $\A$. By Lemma \ref{Lemma: Vershynin},
\begin{equation*}
\mathbb{E}\|\A_{\mathsf{row}}'\|\le C_1'\sqrt{\frac{q}{n}}\|\A\|+C_2'\sqrt{\log q}\sqrt{\frac{\log n}{n}}\|\A\|_F
\end{equation*}
for absolute constants $C_1'$ and $C_2'$, and by a Markov bound, with probability at least $0.999$,
\begin{equation*}
\begin{split}
\|\A_{\mathsf{row}}'\|&\le C_1\sqrt{\frac{q}{n}}\|\A\|+C_2\sqrt{\log q}\sqrt{\frac{\log n}{n}}\|\A\|_F\\
&\le C_1\sqrt{\frac{q}{n}}\frac{\|\A\|_F}{\sqrt{4c_1 e^2}}+C_2\sqrt{\log q}\sqrt{\frac{\log n}{n}}\|\A\|_F\quad(\text{since }\srank(\A)>4c_1 e^2)
\end{split}
\end{equation*}
for absolute constants $C_1$ and $C_2$. By a Markov bound, we also have with constant probability that
\begin{equation*}
\|\A_{\mathsf{row}}'\|_F^2\le c\frac{q}{n}\|\A\|_F^2.
\end{equation*}
Conditioning on this event, by applying the same argument on the column sampling of $\A_{\mathsf{row}}'$, we have
\begin{equation*}
\begin{split}
\|\widetilde\A'\|&\le C_1\sqrt{\frac{q}{n}}\|\A_{\mathsf{row}}'\|+C_2\sqrt{\log q}\sqrt{\frac{\log n}{n}}\|\A_{\mathsf{row}}'\|_F\\
&\le C_1\sqrt{\frac{q}{n}} \left(C_1\sqrt{\frac{q}{n}}\frac{\|\A\|_F}{\sqrt{4c_1 e^2}}+C_2\sqrt{\log q}\sqrt{\frac{\log n}{n}}\|\A\|_F\right)+C_2\sqrt{\log q}\sqrt{\frac{\log n}{n}}\sqrt{\frac{cq}{n}}\|\A\|_F\\
&\le C_0\frac{\|\A\|_F}{\sqrt{4c_1e^2}}\frac{q}{n},\qquad(\text{because the first term dominates as } q\gg \text{a constant})
\end{split}
\end{equation*}
where $\widetilde\A'$ is the matrix after the column sampling of $\A_{\mathsf{row}}'$, and $C_0$ is an absolute constant.
On the other hand, when $\srank(\A)\le 4e^2$ and $q=\cO(\frac{\log n}{\epsilon})$, we have with high probability that
\begin{equation*}
\|\widetilde\A'\|\ge C\frac{q}{n}\|\A\|= C\frac{q}{n}\frac{\|\A\|_F}{\sqrt{\srank(\A)}}\ge C\frac{q}{n}\frac{\|\A\|_F}{\sqrt{4e^2}},
\end{equation*}
where the first inequality holds by applying Lemma \ref{lemma: Ismail} twice on row and column sampling (set $\beta=\Theta(\epsilon)$ and $\tau=\Theta(1)$ there). By setting $c_1$ as a large absolute constant, we have
\begin{equation*}
C_0\frac{\|\A\|_F}{\sqrt{4c_1e^2}}\frac{q}{n} < C\frac{q}{n}\frac{\|\A\|_F}{\sqrt{4e^2}}.
\end{equation*}
Thus we can distinguish (a) $\srank(\A)\le 4e^2$ and (b) $\srank(\A)>4c_1e^2$ by checking $\|\widetilde\A'\|$. If we find $\srank(\A)> 4c_1e^2$, then we can safely output ``$\A\in \mathsf{H1}$''. Therefore, in the following we can assume $\srank(\A)\le 4c_1e^2$.

Recall that, according to Lemma~\ref{lem:schatten_norm_gap}, there is a multiplicative gap in the Schatten-$p$ norm between case $\mathsf{H0}$ and case $\mathsf{H1}$. Without loss of generality, we assume $\|\A\|_{\cS_p}^p=(1\pm\epsilon)cn^p$ in the following, which represents the hardest case to distinguish $\mathsf{H0}$ from $\mathsf{H1}$. In that case, $\srank(\A)=\cO(1)$ and $\|\A\|_F^2=\Theta(n^2)$, as we have shown in the beginning of the proof.

We now show that with $\mathsf{poly}(1/\epsilon)$ sampled entries, we can have an estimator which approximates $\|\A\|_{\cS_p}^p$ up to $(1\pm\epsilon)$ factor; therefore, we can distinguish $\mathsf{H0}$ from $\mathsf{H1}$ due to the gap of $\|\A\|_{\cS_p}^p$ in the two cases.
Consider all singular values of $\A$ which are at most $n/\sqrt{r}$ and consider $\|\A\|_{\cS_p}^p$. This is maximized when there are as many singular values as possible that are equal to $n/\sqrt{r}$. Note that there can be at most $r$ singular values of value $n/\sqrt{r}$, since $\|\A\|_F^2 \leq n^2$ for $\|\A\|_\infty\leq 1$. Therefore, the total contribution of singular values which are no larger than $n/\sqrt{r}$ is at most $n^p r^{1-p/2}$. So if $r=(c\epsilon/3)^{\frac{2}{2-p}}$, this quantity is at most $c\epsilon n^p/3$. Thus all singular values less than $n(c\epsilon/3)^{\frac{1}{p-2}}$ contribute not too much, at most $c\epsilon n^p/3$. For the remaining singular values (i.e., $\sigma_i(\A)>n(c\epsilon/3)^{\frac{1}{p-2}}$), by Theorem \ref{theorem: operator norm estimator by Ismail}, with $\cO\left(\frac{p^4}{\epsilon^{4p/(p-2)}}\log^2 n\right)$ samples we have
\begin{equation*}
\begin{split}
\left|\sigma_i^2(\A)-\sigma_i^2(\A_0)\right|&=\left|\sigma_i(\A^\top\A)-\sigma_i(\A_0^\top\A_0)\right|\le \|\A^\top\A-\A_0^\top\A_0\|\le \frac{2\epsilon}{3p}\left(\frac{c\epsilon}{3}\right)^{\frac{2}{p-2}}\|\A\|^2 \\
&\le \frac{2\epsilon}{3p}\left(\frac{c\epsilon}{3}\right)^{\frac{2}{p-2}} n^2<\frac{2\epsilon}{3p}\sigma_i^2(\A),
\end{split}
\end{equation*}
namely, $\sigma_i^p(\A_0) = (1\pm\epsilon/3)\sigma_i^p(\A)$. Therefore,
\begin{equation*}
\sum_{i:\ \sigma_i(\A) > n(c\epsilon/3)^{1/(p-2)}} \sigma_i^p(\A_0)= (1\pm \epsilon/3)\sum_{i:\ \sigma_i(\A) > n(c\epsilon/3)^{1/(p-2)}} \sigma_i^p(\A)= (1\pm 2\epsilon/3)\|\A\|_{\cS_p}^p,
\end{equation*}
where the last $\subseteq$ holds because all singular values less than $n(c\epsilon/3)^{\frac{1}{p-2}}$ contribute at most $c\epsilon n^p/3$. Let $\mathcal{I}=\{i\mid \sigma_i(\A_0) > (1+\epsilon/(3p))n(c\epsilon/3)^{1/(p-2)}\}$ and $\mathcal{J}=\{i\mid \sigma_i(\A) > n(c\epsilon/3)^{1/(p-2)}\}$. We note that $\mathcal{I}\subseteq \mathcal{J}$, and that all singular values of $\A$ less than $(1+\epsilon/(3p))n(c\epsilon/3)^{\frac{1}{p-2}}$ contribute not too much, at most $c\epsilon n^p/2$, by a similar analysis as above. Therefore, those singular values of $\A$ that lie in $\mathcal{J}\backslash\cI$ contribute at most $c\epsilon n^p/6$, and by the relation $\sigma_i^p(\A_0) = (1\pm\epsilon/3)\sigma_i^p(\A)$ for all $i\in\mathcal{J}$, those singular values of $\A_0$ that lie in $\mathcal{J}\backslash\cI$ contribute at most $c\epsilon n^p/5$.
Therefore
\begin{equation*}
\sum_{i\in\cI} \sigma_i^p(\A_0)=\sum_{i\in\mathcal{J}} \sigma_i^p(\A_0)\pm \frac{c\epsilon n^p}{5} =(1\pm \epsilon)\|\A\|_{\cS_p}^p,
\end{equation*}
as desired.
\end{proof}

\subsection{Lower Bounds}

\begin{problem}[Schatten-$p$ Norm Testing in the Bounded Entry Model for $p\in [1,2)$]
Let $p\in [1,2)$ and $\A\in\mathbb{R}^{n\times n}$ with $\|\A\|_\infty\leq 1$. For a constant $c$, the matrix $\A$ satisfies one of the promised properties:
\begin{itemize}
\item[$\mathsf{H0.}$]
$\|\A\|_{\cS_p}^p\geq cn^{1+p/2}$;
\item[$\mathsf{H1.}$]
$\A$ is $\epsilon$-far from $\|\A\|_{\cS_p}^p\geq cn^{1+p/2}$, meaning that it requires changing at least an $\epsilon$-fraction of entries of $\A$ such that $\|\A\|_{\cS_p}^p\geq cn^{1+p/2}$.
\end{itemize}
The problem is to design a non-adaptive property testing algorithm that outputs $\mathsf{H0}$ with probability at least $0.9$ if $\A\in \mathsf{H0}$, and output $\mathsf{H1}$ with probability at least $0.9$ if $\A\in\mathsf{H1}$, with the least number of queried entries.
\end{problem}

Suppose that $\G\sim \cG(n,n)$ and $\bO$ is a random $n\times n$ orthogonal matrix. Consider two distributions $\mathcal{D}_1 = \frac{1+\eta}{\sqrt n}\G$ and $\mathcal{D}_2 = \bO + \frac{\eta}{\sqrt n}\G$, where $\eta > 0$ is a small absolute constant. The following lemma comes from a manuscript of Li et al.~\cite{LNW:schatten_unpublished}.

\begin{lemma}[\cite{LNW:schatten_unpublished}]
\label{lem:schatten_p_lb_hard_instance}
Consider a linear sketch of length $m$ for random matrices drawn from $\cD_1$ or $\cD_2$. Let $\cL_1$ and $\cL_2$ be the induced distribution of the linear sketch of $\cD_1$ and $\cD_2$, respectively. There exists $\alpha = \alpha(\eta)\in (0,1)$ such that whenever $m\leq \alpha n$, it holds that $d_{TV}(\cL_1,\cL_2) < 1/10$.
\end{lemma}

\begin{theorem}
\label{thm:schatten_p_lb}
Let $p\in [1,2)$ be a constant. There exist constants $c = c(p)$ and $\epsilon_0 = \epsilon_0(p)$ such that for any $\epsilon \leq \epsilon_0$ and $\A\in\R^{n\times n}$, any non-adaptive algorithm that correctly tests $\mathsf{H0}$ against $\mathsf{H1}$ with probability at least $0.99$ must make $\Omega(n)$ queries (i.e., the sketch size is $\Omega(n)$). 
\end{theorem}
\begin{proof}
Consider the hard distributions $\cD_1$ and $\cD_2$ for Lemma~\ref{lem:schatten_p_lb_hard_instance}. For $p<2$, it is a well-known fact that with high probability over $\G\sim \cG(n,n)$, it holds that $\|\frac{1}{\sqrt n}\G\| \leq 2(1+o(1))$ and $\|\frac{1}{\sqrt n}\G\|_{\cS_p}^p \leq (1+o(1))c_p n$ for some constant $c_p < 1$ that depends only on $p$. Hence with high probability, when $\A\sim \cD_1$, it holds that $\|\A\|_{\cS_p}^p \leq (1+o(1))(1+\eta)c_p n$. On the other hand, with high probability, when $\A\sim \cD_2$, it follows from the triangle inequality that $\|\A\|_{\cS_p}^p\geq (1-(1+o(1))\eta c_p^{1/p})^p n$. Therefore, when $\eta$ is sufficiently small (depending on $p$ only), there is a constant-factor multiplicative gap in $\|\A\|_{\cS_p}^p$ between $\cD_1$ and $\cD_2$.

Let $C$ be a large constant to be determined. We truncate $\cD_1$ and $\cD_2$ by applying the map
\[
x\mapsto \max\left\{\min\left\{x, \frac{C}{\sqrt n}\right\}, -\frac{C}{\sqrt n}\right\}
\]
entrywise to the matrices, resulting in two new distributions $\widetilde \cD_1$ and $\widetilde \cD_2$. We claim that with high probability, there remains a constant-factor multiplicative gap in $\|\A\|_{\cS_p}^p$ between $\widetilde\cD_1$ and $\widetilde\cD_2$. It suffices to show that with high probability, truncation incurs only a change of $c n$ for some small constant $c>0$ in $\|\A\|_{\cS_p}^p$ for both $\cD_1$ and $\cD_2$.

Suppose that $\A\sim \cD_1$, and let $\widetilde \A$ be the truncated matrix. We can write $\widetilde\A = \A + \frac{1}{\sqrt n}\B$, where $\B$ is a random matrix with i.i.d.\ entries following a truncated Gaussian distribution $\tilde\cN_C(0,1)$ whose probability density function is
\[
f_C(t) = (1-p_C)\delta(t) + \frac{1}{\sqrt{2\pi}}\exp\left(-\frac{(|t|+C)^2}{2}\right),
\]
where $\delta(t)$ is the Dirac delta function and
\[
p_C = \frac{2}{\sqrt{2\pi}}\int_{C}^{\infty}\exp\left(-\frac{x^2}{2}\right)dx =: \erfc\left(\frac{C}{\sqrt 2}\right).
\]
One can also calculate that
\begin{equation}
\label{equ: m_C}
m_C := \bbE|\B_{ij}|^2
\end{equation}
has a subgaussian decay w.r.t.\ $C$. It follows from a Chernoff bound that $\|\B\|_F^2\leq 2m_C n^2$ with high probability. Since $\|\B\|_{\cS_p} \leq n^{\frac1p-\frac12}\|\B\|_F \leq \sqrt{2m_C}n^{\frac1p+\frac12}$, we see that $\|\frac{1}{\sqrt n}\B\|_{\cS_p} \leq \sqrt{2m_C}n^{\frac{1}{p}}$, where the constant factor can be made arbitrarily small by choosing $C$ large enough; that is, truncating $\A\sim \cD_1$ incurs only a constant factor loss (where the constant can be made arbitrarily small) in $\|\A\|_{\cS_p}^p$ with probability $\geq 0.999$.
%
%The Marchenko-Pastur Law applies to $\frac{1}{\sqrt n}\B$ and hence $\frac{1}{n}\E\|\frac{1}{\sqrt n}\B\|_{\cS_p}^p \to c_p'$ as $n\to\infty$ for some constant $c_p' > 0$ that can be made arbitrarily small. Therefore, when $n$ is sufficiently large, with probability $\geq 0.999$, $\|\frac{1+\eta}{\sqrt{n}}\B\|_{\cS_p}^p \leq 2000(1+\eta)c_p' n$, that is, truncating $\A\sim \cD_1$ incurs only a constant factor loss (where the constant can be made arbitrarily small) in $\|\A\|_{\cS_p}^p$ with probability $\geq 0.999$.
%

Next, suppose that $\A'\sim \cD_2$ and we write the truncation as $\tilde \A' = \A' + \B'$. It is a classical result~\cite{borel} that
\begin{equation}\label{eqn:orthogonal_matrix_entry}
\lim_{n\to\infty} \Pr\{\sqrt{n}\bO_{ij} \leq t\} = \Pr_{g\sim\cN(0,1)}\{g \leq t\}.
\end{equation}
Observe that $\A'_{ij} \stackrel{\text{dist}}{=} \bO_{ij} + \frac{\eta}{\sqrt n}g'$ and $\A_{ij} \stackrel{\text{dist}}{=} \frac{1}{\sqrt{n}}g + \frac{\eta}{\sqrt n} g'$ with the same additive `noise' $\frac{\eta}{\sqrt n}g'$, where $g,g'\sim N(0,1)$ are independent, it follows that $\Pr\{\sqrt n\A'_{ij} \leq t\} \to \Pr\{\sqrt{n}\A_{ij} \leq t\}$ for any (fixed) $t$ as $n\to\infty$ (note that $\sqrt{n}\A_{ij} \stackrel{\text{dist}}{=} (1+\eta)g$ and does not depend on $n$). Hence each entry of $\B_{ij}'$ is stochastically dominated by $\frac{1}{\sqrt n}\tilde\cN_{C/2}(0,1)$. Similarly to before, $\bbE\|\B'\|_{\cS_p} \leq n^{\frac1p-\frac12}\bbE\|\B'\|_F \leq \sqrt{m_{C/2}}n^{\frac1p}$, and thus by Markov's inequality, with probability $\geq 0.999$ it holds that $\|\B'\|_{\cS_p} \leq 1000\sqrt{m_{C/2}}n^{\frac1p}$; that is, truncating $\A'\sim \cD_2$ incurs only a constant factor loss (where the constant can be made arbitrarily small) in $\|\A'\|_{\cS_p}^p$ with probability $\geq 0.999$.

% Next we shall show that $\bbE\|\B'\|$ is a small constant that can be made arbitrarily small by increasing $C$. By symmetrization trick, $\bbE\|\B'\|_{\cS_p}\leq 2\mathbb{E}\|(\epsilon_{ij}\B'_{ij})\|_{\cS_p}$, where $\epsilon_{ij}$ are independent random $\pm 1$s. Invoking the Non-commutative Khintchine's Inequality we have that
%\[
%\bbE\|(\epsilon_{ij}\B'_{ij})\|_{\cS_p} \lesssim_p \bbE\max\left\{ \sum_i\left(\sum_j (\B'_{ij})^2\right)^p , \sum_j\left(\sum_i (\B'_{ij})^2\right)^p\right\}^{\frac{1}{p}}
%\]
%Observe that for each row $i$, $(\A'_{i1})^2,\dots,(\A'_{in})^2$ are negatively associated and so are $(\B'_{i1})^2,\dots,(\B'_{in})^2$. Thus we can apply Chernoff bound to conclude that $\sum_j (\B'_{ij})^2$ concentrates around its expectation, which is at most $m_{C/2}$. This implies that $\sum_i(\sum_j (\B'_{ij})^2)^p$ concentrates around $n m_{C/2}^p$, and that $\bbE(\sum_i(\sum_j (\B'_{ij})^2)^p)^{1/p}\lesssim m_{C/2} n^{1/p}$. The same result holds for $\bbE(\sum_j(\sum_i (\B'_{ij})^2)^p)^{1/p}$. We conclude that $\bbE\|\B'\|_{\S_p}\lesssim m_{C/2} n^{1/p}$ and the constant can be made arbitrarily small by choosing $C$ large enough. Therefore, truncating $\hat \A'\sim \cD_2$ incurs only a constant factor loss (where the constant can be made arbitrarily small) in $\|\A'\|_{\cS_p}^p$ with probability $\geq 0.999$.

Now, the matrices from $\frac{\sqrt n}{C}\widetilde{\cD_1}$ and $\frac{\sqrt n}{C}\widetilde{\cD_2}$ have entries bounded by $1$, and with high probability, $\|\A\|_{\cS_p}\leq c_1 n^{\frac{1}{2}+\frac{1}{p}}$ when $\A\sim \frac{\sqrt n}{C}\widetilde{\cD_1}$ and
$\|\A\|_{\cS_p}\geq c_2 n^{\frac{1}{2}+\frac{1}{p}}$ when $\A\sim \frac{\sqrt n}{C}\widetilde{\cD_2}$, for constants $c_1 < c_2$ (depending on $\eta$ and $C$). Our result of the theorem would follow immediately from Theorem~\ref{thm:schatten_p_lb} once we establish that with high probability, a random matrix from $\frac{\sqrt n}{C}\widetilde{\cD_1}$ is $\epsilon$-far from having Schatten $p$-norm at least $c_2 n^{\frac{1}{2}+\frac{1}{p}}$. Indeed, let $\E$ denote the pertubation to $\A$ such that $\|\E\|_0\leq \epsilon n^2$ and $\|\A+\E\|_\infty\leq 1$. Since $\|\A\|_\infty\leq 1$, it must hold that $\|\E\|_\infty\leq 2$. Thus $\|\E\|_{\cS_p} \leq n^{\frac1p-\frac12}\|\E\|_F \leq 2 \sqrt{\epsilon}  n^{\frac1p+\frac12}$. When $\epsilon$ is sufficiently small, it is easy to see via triangle inequality that there remains a constant-factor gap between $\|\A+\E\|_{\cS_p}$ and $c_2 n^{\frac12+\frac1p}$.
\end{proof}

\section{Non-adaptive Testing of Matrix Entropy}
\label{section: Entropy}

Recall that we defined in Eq.~\eqref{eqn:entropy_defn} the entropy of a matrix $\A\in \R^{n\times n}$ as
\[
H(\A) = \frac{-\sum_i \frac{\sigma_i^2(\A)}{n^2}\log\frac{\sigma_i^2(\A)}{n^2}}{\sum_i \frac{\sigma_i^2(\A)}{n^2}},
\]
with the convention that $0\cdot\infty = 0$. For matrices $\A$ satisfying $\|\A\|_\infty\leq 1$, it holds that $\sigma_i(\A)\leq n$ for all $i$ and the entropy above coincides with the usual Shannon entropy. Note that scaling only changes the entropy additively; that is, $H(\beta\A) = H(\A) - \log\beta^2$.

Our goal in this section is to show that the following testing problem on matrix entropy has a lower bound of $\Omega(n)$ queries. 

\begin{problem}[Entropy Testing in the Bounded Entry Model]
Let $\A\in\mathbb{R}^{n\times n}$ be a matrix with bounded absolute values of entries by $1$. For some absolute constant $c$, $\A$ has one of promised properties:
\begin{itemize}
\item[$\mathsf{H0.}$]
$H(\A)\leq \log n + \log \log \log n - c$;
\item[$\mathsf{H1.}$]
$\A$ is $\frac{\epsilon}{\log n \log\log n}$-far from having entropy at most $\log n + \log\log \log n - c$, meaning that $\A$ requires changing at least an $\frac{\epsilon}{\log n\log\log n}$-fraction of its entries so that $H(\A)\leq \log n + \log \log \log n - c$.
\end{itemize}
The problem is to design a non-adaptive property testing algorithm that outputs $\mathsf{H0}$ with probability at least $0.9$ if $\A\in \mathsf{H0}$, and output $\mathsf{H1}$ with probability at least $0.9$ if $\A\in\mathsf{H1}$, with the least number of queried entries.
\end{problem}

\begin{theorem}
\label{thm:entropy_lb}
There exist absolute constants $c > 0$ and $\epsilon_0 > 0$ such that for any $\epsilon \leq \epsilon_0$ and $\A\in\R^{n\times n}$, any non-adaptive algorithm that correctly tests $\mathsf{H0}$ against $\mathsf{H1}$ with probability at least $0.99$ must make $\Omega(n)$ queries (i.e., the sketch size is $\Omega(n)$). 
\end{theorem}
\begin{proof}
Consider the hard instances $\widetilde{\cD}_1$ of truncated Gaussian random matrices and $\widetilde{\cD}_2$ of truncated random orthogonal matrices in the proof of Theorem~\ref{thm:schatten_p_lb}. We first claim that with high probability, their entropies are different by an additive constant.

Suppose that $\G\sim\cG(n,n)$ and $\bO$ is a uniformly random orthogonal matrix. It is clear that $H(\sqrt n\bO) = \log n$. Let $\G' = \frac{1}{\sqrt n}\G$, we have that
\[
H(\G) = H(\sqrt n \G') = -\frac{\sum_i \frac{\sigma_i^2(\G')}{n}\log\frac{\sigma_i^2(\G')}{n}}{\sum_i \frac{\sigma_i^2(\G')}{n}} = -\frac{\sum_i \sigma_i^2(\G')\log \sigma_i^2(\G')}{\sum_i \sigma_i^2(\G')} + \log n.
\]
It follows from the Marchenko-Pastur law that, even with truncation of entries (this fact is actually used in the proof of the Marchenko-Pastur law, see, e.g.,~\cite[Section 3.1.3]{BS10}), the denominator and the numerator both converge to its expected value almost surely, with a rate of convergence of $\cO(1/\sqrt{n})$. Let
\[
p(x) = \frac{\sqrt{x(4-x)}}{2\pi x},\quad x\in[0,4]
\]
be the density function of the Marchenko-Pastur law, then the first term is concentrated around
\[
-\frac{\int_0^4 (x \log x)p(x) \,dx}{\int_0^4  x p(x)\, dx} = -\frac{1}{2}.
\]
This shows an additive constant difference between the entropy of truncated Gaussian random matrices and scaled random orthogonal matrices. Specifically, $H(\frac{1}{C}\G) \leq \log n + 2\log C - c$ with high probability for some absolute constant $c$ close to $1/2$, and $H(\frac{\sqrt n}{C}\bO) = \log n + 2\log C$.

Next we shall show that the perturbation $\frac{\eta}{\sqrt{n}}\G$, truncation and the change of at most $\epsilon n^2/(\log n\log\log n)$ entries altogether will affect the entropy of $\frac{\sqrt n}{C}\bO$ by at most an additive constant that can be made arbitrarily small (depending on $\eta$, the truncation threshold $C$ and $\epsilon$). Viewing all changes as a perturbation matrix $\E$, we have with probability $\geq 0.99$ that $\|\E\|_F^2\leq 3(\frac{2\eta}{n} + 100m_{C/2}+\frac{\epsilon}{\log n \log\log n})n^2$, where $\frac{2\eta}{n}n^2$ comes from $\frac{\eta}{\sqrt{n}}\G$, $100m_{C/2} n^2$ from the truncation of entries (recall the fact that an individual entry of $\bO$ can be approximated by $\cN(0,1/n)$; see the proof of Theorem~\ref{thm:schatten_p_lb}) and $\frac{\epsilon}{\log n \log \log n}n^2$ from the change of the entries. Setting $C = C_1\sqrt{\log\log n}$ for some large absolute value $C_1$, we have that
\[
\|\E\|_F^2 \leq 3\left(\frac{2\eta}{n} + \frac{1}{\operatorname{poly}(\log n)} + \frac{\epsilon_0}{\log n\log\log n}\right)n^2 \leq \frac{4\epsilon_0}{\log n \log\log n} n^2.
\]

For notational simplicity, let $\epsilon' = 4\epsilon_0$, $s_i = \sigma_i(\bO+\frac{C}{\sqrt n}\E)$ and $\delta_i = s_i - 1$. Since
\[
H\left(\frac{\sqrt{n}\bO+\E}{C}\right) = H\left(\bO + \frac{C}{\sqrt n}\E\right) - \log n + \log C^2 = - \frac{\sum_i s_i^2\log s_i^2}{\sum_i s_i^2} + \log n + \log C^2,
\]
we wish to show that the first term on the right-hand side is at most a constant. By the Hoffman-Wielandt inequality, we have that
\begin{equation}\label{eqn:constraint}
\sum_i \delta_i^2 \leq \left\|\frac{C}{\sqrt n}\E\right\|_{F}^2 \leq \frac{C_1^2\epsilon' n}{\log n}
%\sum_i |\delta_i| \leq \left\|\frac{C}{\sqrt n}\E\right\|_{\cS_1} \leq \sqrt{n}\left\|\frac{C}{\sqrt n}\E\right\|_F \leq C_1\sqrt{\frac{\epsilon'}{\log n}}n, \notag
\end{equation}
%and thus
%\[
%S := \sum_i s_i^2  = \sum_i (1+\delta_i)^2 \in \left[\left(1-C_1\sqrt{\frac{\epsilon'}{\log n}}\right)^2n, \left(1+C_1\sqrt{\frac{\epsilon'}{\log n}}\right)^2 n\right].
%\]
%
Consider the problem of minimizing
\[
f(\delta_1,\dots,\delta_n) = - \frac{\sum_i (1+\delta_i)^2\log (1+\delta_i)^2}{\sum_i (1+\delta_i)^2}
\]
subject to the constraint \eqref{eqn:constraint}. In the interior of the region, one can verify using calculus that the extremal value is attained when $\delta_i$'s are all equal for those $\delta_i\neq -1$. One can further verify that the extremal value in the interior of the region is the maximal value. On the boundary of the region, one can use Lagrange multipliers to obtain that the extremal value is attained when nonzero $\delta_i$'s are all equal. It then follows that the function attains the minimum value when $\delta_1 =  C_1\sqrt{\epsilon' n/\log n}$ and $\delta_2 = \cdots = \delta_n = 0$, in which case one can calculate
\[
f(\delta_1,\dots,\delta_n) = -C_1^2\epsilon' + o(1).
\]

We can take $\epsilon_0$ to be sufficiently small such that $C_1^2\epsilon'$ is smaller than the additive gap between $H(\frac{1}{C}\G)$ and $H(\frac{\sqrt{n}}{C}\bO)$. This shows that there would still be an additive constant gap between $\widetilde{D}_1$ and $\widetilde{D}_2$ after changing $\epsilon n^2/(\log n \log\log n)$ entries to a matrix drawn from $\widetilde{D}_2$. The lower bound follows from Theorem~\ref{lem:schatten_p_lb_hard_instance}, where we can ignore the additive Gaussian term in the two hard instances there by letting $\eta$ be arbitrarily small.
\end{proof}

\paragraph{Acknowledgements} 
Part of this work was done while D.\@ W.\@ and H.\@ Z.\@ were visiting the Simons Institute for the Theory of Computing. We thank Sepehr Assadi, Yury Polyanskiy, Ran Raz, Avishay Tal, Zhengyu Wang, and Yihong Wu for useful discussions. This work was supported in part by grants NSF CCF-1422910, NSF CCF-1535967,  NSF IIS-1618714, Office of Naval Research (ONR) grant N00014-18-1-2562, and a Microsoft Research Fellowship.

\bibliographystyle{plain}
\bibliography{reference}

\begin{thebibliography}{10}

\bibitem{alon2013approximate}
Noga Alon, Troy Lee, Adi Shraibman, and Santosh Vempala.
\newblock The approximate rank of a matrix and its algorithmic applications:
  approximate rank.
\newblock In {\em ACM Symposium on Theory of Computing}, pages 675--684, 2013.

\bibitem{assadi2017estimating}
Sepehr Assadi, Sanjeev Khanna, and Yang Li.
\newblock On estimating maximum matching size in graph streams.
\newblock In {\em ACM-SIAM Symposium on Discrete Algorithms}, pages 1723--1742,
  2017.

\bibitem{awasthi2016learning}
Pranjal Awasthi, Maria-Florina Balcan, Nika Haghtalab, and Hongyang Zhang.
\newblock Learning and 1-bit compressed sensing under asymmetric noise.
\newblock In {\em Annual Conference on Learning Theory}, pages 152--192, 2016.

\bibitem{BS10}
Zhidong Bai and Jack~W Silverstein.
\newblock {\em Spectral analysis of large dimensional random matrices},
  volume~20.
\newblock Springer, 2010.

\bibitem{balcan2011learning}
Maria-Florina Balcan and Nicholas~JA Harvey.
\newblock Learning submodular functions.
\newblock In {\em ACM Symposium on Theory of Computing}, pages 793--802, 2011.

\bibitem{balcan2018matrix}
Maria-Florina Balcan, Yingyu Liang, David~P. Woodruff, and Hongyang Zhang.
\newblock Matrix completion and related problems via strong duality.
\newblock In {\em Innovations in Theoretical Computer Science}, volume~94,
  2018.

\bibitem{balcan2016noise}
Maria-Florina Balcan and Hongyang Zhang.
\newblock Noise-tolerant life-long matrix completion via adaptive sampling.
\newblock In {\em Advances in Neural Information Processing Systems}, pages
  2955--2963, 2016.

\bibitem{barman2016dictionary}
Siddharth Barman, Arnab Bhattacharyya, and Suprovat Ghoshal.
\newblock Testing sparsity over known and unknown bases.
\newblock {\em arXiv preprint arXiv:1608.01275}, 2016.

\bibitem{borel}
\'Emile Borel.
\newblock Sur les principes de la th\'eorie cin\'etique des gaz.
\newblock {\em Annales scientifiques de l'\'Ecole Normale Sup\'erieure, S\'erie
  3}, 23:9--32, 1904.

\bibitem{bouwmans2018applications}
Thierry Bouwmans, Sajid Javed, Hongyang Zhang, Zhouchen Lin, and Ricardo Otazo.
\newblock On the applications of robust {PCA} in image and video processing.
\newblock {\em Proceedings of the IEEE}, 106(8):1427--1457, 2018.

\bibitem{bury2015sublinear}
Marc Bury and Chris Schwiegelshohn.
\newblock Sublinear estimation of weighted matchings in dynamic data streams.
\newblock In {\em Algorithms-ESA 2015}, pages 263--274. 2015.

\bibitem{chan2014optimal}
Siu-On Chan, Ilias Diakonikolas, Paul Valiant, and Gregory Valiant.
\newblock Optimal algorithms for testing closeness of discrete distributions.
\newblock In {\em ACM-SIAM Symposium on Discrete Algorithms}, pages 1193--1203,
  2014.

\bibitem{clarkson2009numerical}
Kenneth~L Clarkson and David~P Woodruff.
\newblock Numerical linear algebra in the streaming model.
\newblock In {\em ACM Symposium on Theory of Computing}, pages 205--214, 2009.

\bibitem{clarkson2013low}
Kenneth~L Clarkson and David~P Woodruff.
\newblock Low rank approximation and regression in input sparsity time.
\newblock In {\em ACM Symposium on Theory of Computing}, pages 81--90, 2013.

\bibitem{daskalakis2013testing}
Constantinos Daskalakis, Ilias Diakonikolas, Rocco~A Servedio, Gregory Valiant,
  and Paul Valiant.
\newblock Testing k-modal distributions: Optimal algorithms via reductions.
\newblock In {\em ACM-SIAM Symposium on Discrete Algorithms}, pages 1833--1852,
  2013.

\bibitem{deshpande2011algorithms}
Amit Deshpande, Madhur Tulsiani, and Nisheeth~K Vishnoi.
\newblock Algorithms and hardness for subspace approximation.
\newblock In {\em ACM-SIAM symposium on Discrete Algorithms}, pages 482--496,
  2011.

\bibitem{esfandiari2014streaming}
Hossein Esfandiari, Mohammad~T Hajiaghayi, Vahid Liaghat, Morteza Monemizadeh,
  and Krzysztof Onak.
\newblock Streaming algorithms for estimating the matching size in planar
  graphs and beyond.
\newblock In {\em ACM-SIAM Symposium on Discrete Algorithms}, pages 1217--1233,
  2014.

\bibitem{hardt2014understanding}
Moritz Hardt.
\newblock Understanding alternating minimization for matrix completion.
\newblock In {\em IEEE Symposium on Foundations of Computer Science}, pages
  651--660, 2014.

\bibitem{hardt2012simple}
Moritz Hardt, Katrina Ligett, and Frank McSherry.
\newblock A simple and practical algorithm for differentially private data
  release.
\newblock In {\em Advances in Neural Information Processing Systems}, pages
  2339--2347, 2012.

\bibitem{hardt2013algorithms}
Moritz Hardt and Ankur Moitra.
\newblock Algorithms and hardness for robust subspace recovery.
\newblock In {\em Annual Conference on Learning Theory}, pages 354--375, 2013.

\bibitem{jain2013low}
Prateek Jain, Praneeth Netrapalli, and Sujay Sanghavi.
\newblock Low-rank matrix completion using alternating minimization.
\newblock In {\em ACM Symposium on Theory of Computing}, pages 665--674, 2013.

\bibitem{khetan2017spectrum}
Ashish Khetan and Sewoong Oh.
\newblock Matrix norm estimation from a few entries.
\newblock In {\em Advances in Neural Information Processing Systems}, pages
  6424--6433. 2017.

\bibitem{kong2017spectrum}
Weihao Kong and Gregory Valiant.
\newblock Spectrum estimation from samples.
\newblock {\em The Annals of Statistics}, 45(5):2218--2247, 2017.

\bibitem{koren2009matrix}
Yehuda Koren, Robert Bell, and Chris Volinsky.
\newblock Matrix factorization techniques for recommender systems.
\newblock {\em Computer}, 42(8), 2009.

\bibitem{krauthgamer2003property}
Robert Krauthgamer and Ori Sasson.
\newblock Property testing of data dimensionality.
\newblock In {\em ACM-SIAM Symposium on Discrete Algorithms}, pages 18--27,
  2003.

\bibitem{laurent2000adaptive}
Beatrice Laurent and Pascal Massart.
\newblock Adaptive estimation of a quadratic functional by model selection.
\newblock {\em Annals of Statistics}, 28(5):1302--1338, 2000.

\bibitem{LT91}
Michel Ledoux and Michel Talagrand.
\newblock {\em Probability in Banach Spaces: Isoperimetry and Processes}.
\newblock A Series of Modern Surveys in Mathematics Series. Springer, 1991.

\bibitem{li2014sketching}
Yi~Li, Huy~L. Nguyen, and David~P. Woodruff.
\newblock On sketching matrix norms and the top singular vector.
\newblock In {\em {ACM-SIAM} Symposium on Discrete Algorithms}, pages
  1562--1581, 2014.

\bibitem{LNW:schatten_unpublished}
Yi~Li, Huy~L. Nguy{\^e}n, and David~P. Woodruff.
\newblock On approximating matrix norms in a stream.
\newblock 2017.
\newblock Submitted.

\bibitem{li2014improved}
Yi~Li, Zhengyu Wang, and David~P. Woodruff.
\newblock Improved testing of low rank matrices.
\newblock In {\em International Conference on Knowledge Discovery and Data
  Mining}, pages 691--700, 2014.

\bibitem{li2016approximating}
Yi~Li and David~P. Woodruff.
\newblock On approximating functions of the singular values in a stream.
\newblock In {\em {ACM} Symposium on Theory of Computing}, pages 726--739,
  2016.

\bibitem{li2016tight}
Yi~Li and David~P. Woodruff.
\newblock Tight bounds for sketching the operator norm, {Schatten} norms, and
  subspace embeddings.
\newblock In {\em International Conference on Approximation Algorithms for
  Combinatorial Optimization Problems, and International Conference on
  Randomization and Computation}, pages 39:1--39:11, 2016.

\bibitem{LW17}
Yi~Li and David~P. Woodruff.
\newblock Embeddings of {Schatten} norms with applications to data streams.
\newblock In {\em International Colloquium on Automata, Languages, and
  Programming}, pages 60:1--60:14, 2017.

\bibitem{lin2017low}
Zhouchen Lin and Hongyang Zhang.
\newblock {\em Low-rank Models in Visual Analysis: Theories, Algorithms, and
  Applications}.
\newblock Academic Press, 2017.

\bibitem{LP86}
F.~Lust-Piquard.
\newblock In\'egalit\'es de {Khintchine} dans $c_p$ ($1<p<\infty$).
\newblock {\em C. R. Math. Acad. Sci. Paris S\'er I Math.}, 303:289--292, 1986.

\bibitem{magdon2010row}
Malik Magdon-Ismail.
\newblock Row sampling for matrix algorithms via a non-commutative {Bernstein}
  bound.
\newblock {\em arXiv preprint arXiv:1008.0587}, 2010.

\bibitem{mahoney2011randomized}
Michael~W Mahoney.
\newblock Randomized algorithms for matrices and data.
\newblock {\em Foundations and Trends in Machine Learning}, 3(2):123--224,
  2011.

\bibitem{nakos2018improved}
Vasileios Nakos, Xiaofei Shi, David~P Woodruff, and Hongyang Zhang.
\newblock Improved algorithms for adaptive compressed sensing.
\newblock In {\em International Colloquium on Automata, Languages, and
  Programming}, pages 90:1--90:14, 2018.

\bibitem{parnas2003testing}
Michal Parnas and Dana Ron.
\newblock Testing metric properties.
\newblock {\em Information and Computation}, 187(2):155--195, 2003.

\bibitem{rudelson2007sampling}
Mark Rudelson and Roman Vershynin.
\newblock Sampling from large matrices: An approach through geometric
  functional analysis.
\newblock {\em Journal of the ACM}, 54(4):21, 2007.

\bibitem{RV:hanson-wright}
Mark Rudelson and Roman Vershynin.
\newblock {Hanson-Wright} inequality and sub-{Gaussian} concentration.
\newblock {\em Electronic Communications in Probability}, 18:1--9, 2013.

\bibitem{sun2015guaranteed}
Ruoyu Sun and Zhi-Quan Luo.
\newblock Guaranteed matrix completion via nonconvex factorization.
\newblock In {\em IEEE Symposium on Foundations of Computer Science}, pages
  270--289, 2015.

\bibitem{tropp2015introduction}
Joel~A Tropp.
\newblock An introduction to matrix concentration inequalities.
\newblock {\em Foundations and Trends in Machine Learning}, 8(1-2):1--230,
  2015.

\bibitem{valiant1977graph}
Leslie Valiant.
\newblock Graph-theoretic arguments in low-level complexity.
\newblock {\em Mathematical Foundations of Computer Science}, pages 162--176,
  1977.

\bibitem{verbin2011streaming}
Elad Verbin and Wei Yu.
\newblock The streaming complexity of cycle counting, sorting by reversals, and
  other problems.
\newblock In {\em ACM-SIAM Symposium on Discrete Algorithms}, pages 11--25,
  2011.

\bibitem{vershynin2010introduction}
Roman Vershynin.
\newblock Introduction to the non-asymptotic analysis of random matrices.
\newblock {\em Compressed Sensing}, pages 210--268, 2010.

\bibitem{woodruff2014sketching}
David~P. Woodruff.
\newblock Sketching as a tool for numerical linear algebra.
\newblock {\em Foundations and Trends in Theoretical Computer Science},
  10(1--2):1--157, 2014.

\bibitem{zhang2013counterexample}
Hongyang Zhang, Zhouchen Lin, and Chao Zhang.
\newblock A counterexample for the validity of using nuclear norm as a convex
  surrogate of rank.
\newblock In {\em Joint European Conference on Machine Learning and Knowledge
  Discovery in Databases}, pages 226--241, 2013.

\bibitem{zhang2016completing}
Hongyang Zhang, Zhouchen Lin, and Chao Zhang.
\newblock Completing low-rank matrices with corrupted samples from few
  coefficients in general basis.
\newblock {\em IEEE Transactions on Information Theory}, 62(8):4748--4768,
  2016.

\bibitem{zhang2015exact}
Hongyang Zhang, Zhouchen Lin, Chao Zhang, and Edward~Y Chang.
\newblock Exact recoverability of robust {PCA} via outlier pursuit with tight
  recovery bounds.
\newblock In {\em AAAI Conference on Artificial Intelligence}, pages
  3143--3149, 2015.

\bibitem{zhang2014robust}
Hongyang Zhang, Zhouchen Lin, Chao Zhang, and Junbin Gao.
\newblock Robust latent low rank representation for subspace clustering.
\newblock {\em Neurocomputing}, 145:369--373, 2014.

\bibitem{zhang2015relations}
Hongyang Zhang, Zhouchen Lin, Chao Zhang, and Junbin Gao.
\newblock Relations among some low-rank subspace recovery models.
\newblock {\em Neural Computation}, 27(9):1915--1950, 2015.

\end{thebibliography}

\newpage

\newpage
\appendix

\section{Other Related Works}

\medskip
\noindent{\textbf{Property Testing of Low-Rank Matrices.}} Krauthgamer and Sasson~\cite{krauthgamer2003property} studied the problem of property testing of data dimensionality, building upon earlier work of Parnas and Ron \cite{parnas2003testing}. They presented algorithms for testing low dimensionality
of a set of vectors and for testing whether a matrix is of low rank. Their algorithm achieves $\cO(d^2/\epsilon^2)$ non-adaptive samples by uniformly sampling an $\cO(d/\epsilon)\times \cO(d/\epsilon)$ submatrix. Later Li et al.~\cite{li2014improved} studied the \emph{adaptive} testing of matrix rank with a sample complexity upper bound $\widetilde\cO(d^2/\epsilon)$. 
%They also studied the \emph{non-adaptive} rank testing problem for the most special case $d=1$, and claimed an $\widetilde\cO(1/\epsilon)$ upper bound. Unfortunately, we find an error in the proof of Theorem 3 in \cite{li2014improved}. 
Despite a large amount of work on the positive results of rank testing, non-trivial negative results in this direction remain absent. Barman et al.~\cite{barman2016dictionary} studied a slightly different setting of the rank problem by testing whether $\mathsf{H0}$: $\rank(\A)\le d$ or $\mathsf{H1}$: $\epsilon$-far from $\rank(\A)\le 20d/\epsilon^2$ with a different definition of ``$\epsilon$-far'' in terms of $\epsilon$-approximate rank~\cite{alon2013approximate}. The $\epsilon$-approximate rank is defined as the minimum rank over matrices that approximate every entry of $\A$ to within an additive $\epsilon$. In contrast to these works, we provide the first $\widetilde\cO(d^2/\epsilon)$ sample complexity upper bound for the more traditional rank testing problem without any rank gap between $\mathsf{H0}$ and $\mathsf{H1}$. We complement this positive result with various matching negative results, showing that any algorithm requires at least $\widetilde\Omega(d^2/\epsilon)$ samples in order to succeed with constant probability over various fields. We also extend the results to sensing oracles and obtain an $\cO(d^2)$ upper bound and an $\widetilde\Omega(d^2)$ matching lower bound.

\medskip
\noindent{\textbf{Property Testing of Stable Rank.}} To the best of our knowledge, this is the first work that studies the stable rank (and the Schatten-$p$ norm) testing problem in the \emph{bounded entry model}. Perhaps the most related work to ours is \cite{li2014improved}, which studied non-adaptive testing of stable rank in the \emph{bounded row model}. In this model, the rows of $\A$ and of the matrix after change have Euclidean norm at most $1$. The algorithm determines if $\A$ has stable rank at most $d$, or requires changing an $\epsilon/d$-fraction of \emph{rows} to have stable rank at most $d$. For this problem, Li et al.~\cite{li2014improved} provided a tight $\Theta(d/\epsilon^2)$ bound. We argue that the bounded entry model is more challenging than the bounded row model, as our restrictions are in fact weaker, which allows for more flexible options of changing the matrix $\A$.

\medskip
\noindent{\textbf{Estimation of Rank.}} Estimating the matrix rank is a learning version of the rank testing problem. Balcan and Harvey~\cite{balcan2011learning} showed that the rank of a subsampled submatrix is highly concentrated around its expectation. Balcan and Zhang~\cite{balcan2016noise} proved that uniformly sampling an $\cO(\mu d\log d)\times \cO(\mu d\log d)$ submatrix suffices to preserve the rank of the original matrix $\A$, under the standard incoherence assumption that the underlying rank-$d$ matrix $\A$ admits a skinny SVD, i.e., $\A=\U\mathbf{\Sigma}\V^\top$ satisfies $\max\{\|\U^\top\e_i\|_2^2,\|\V^\top\e_i\|_2^2\}\le \frac{\mu d}{n}$ for all $i$. Unfortunately, in the worst case the incoherence parameter $\mu$ may be as large as $\poly(n)$, e.g., when $\A$ is a sparse matrix. In contrast, we show that it is possible to detect the rank inexpensively without polynomial dependence in $n$ in the sample complexity.

\medskip
\noindent{\textbf{Estimation of Schatten-$p$ Norm.}} The Schatten-$p$ norm has found many applications in differential privacy~\cite{hardt2012simple} and non-convex optimization~\cite{balcan2018matrix,deshpande2011algorithms} for $p=1$, and in numerical linear algebra~\cite{mahoney2011randomized} for $p\in\{2,\infty\}$. The paper of \cite{li2014sketching} studied the problem of sketching Schatten-$p$ norms for various $p$ under the bilinear sketch and general sketch models. Both the upper bounds and the lower bounds there depend polynomially on $n$. For even $p\ge 4$, they also proposed the first cycle estimator with a $(1\pm\tau)$ approximation in the sketching model. More recently, Kong and Valiant~\cite{kong2017spectrum} applied a similar cycle estimator to approximate the Schatten-$p$ norm of the covariance matrix with computationally efficient algorithms. Khetan and Oh~\cite{khetan2017spectrum} estimated the Schatten-$p$ norm in the sampling model by connecting the cycle estimator with the $p$-cyclic pseudograph, and showed that when $p\in\{3, 4, 5, 6, 7\}$, the estimator can be calculated in $\cO(n^\omega)$ time , where $\omega<2.373$ is the exponent
of matrix multiplication. For the special case of $p=\infty$ (i.e., estimating the largest singular value), to obtain a $(1\pm\epsilon)$ approximation, one would need to raise $p$ to as large as $\widetilde\Theta(1/\epsilon)$. However, the sample complexity in prior work blows up if $p$ goes beyond an absolute constant. Though Magdon-Ismail~\cite{magdon2010row} showed that non-uniform sampling of rows of matrix provides a $(1\pm\epsilon)$ approximation to the largest singular value with small samples, the sampling probability depends on the unknown $\ell_2$ norm of each row. In contrast, we provide the first non-adaptive algorithm to estimate the largest singular value up to $(1\pm\epsilon)$ relative error with sample complexity $\poly(d/\epsilon)$, under modest assumptions that the input matrix has stable rank $d$ and a large Frobenius norm. For constant-factor approximation to the largest singular value $\|\cdot\|$, Rudelson and Vershynin~\cite{rudelson2007sampling} showed that uniformly sampling $q$ rows of a matrix $\A$ gives a subsampled matrix $\A_{\mathsf{row}}$ such that $\|\A_{\mathsf{row}}\|\lesssim \sqrt{\frac{q}{n}}\|\A\|+\sqrt{\log q}\|\A\|_{(n/q)}$, where $\|\A\|_{(n/q)}$ is the average of $n/q$ biggest Euclidean lengths of the columns of $\A$.

\section{New Operator Norm Estimators}

In this section, we develop new $(1\pm\tau)$-approximation estimators to the operator norm in sampling and sensing models.

\subsection{Sampling Algorithms}

We first discuss the sampling algorithms which are only allowed to read the entries of a matrix.

\subsubsection{Estimation without Eigengap}
Before proceeding, we first cite the following result from \cite{magdon2010row}.
\begin{lemma}[Theorem 20, \cite{magdon2010row}]
\label{lemma: Ismail}
Let $\A\in\R^{n\times n}$ have rows $\{\A_{t,:}\}_{t=1}^n$. Independently sample $q$ rows $\A_{t_1,:},\dots,\A_{t_q,:}$ with replacement from $\A$ according to the probabilities:
\begin{equation*}
p_t\ge \beta \frac{\|\A_{t,:}\|_2^2}{\|\A\|_F^2}
\end{equation*}
for $\beta<1$.
Let
\begin{equation*}
\A_0=
\begin{bmatrix}
\frac{\A_{t_1,:}}{\sqrt{q p_{t_1}}}\\
\vdots\\
\frac{\A_{t_q,:}}{\sqrt{q p_{t_q}}}\\
\end{bmatrix}.
\end{equation*}
Then if $q\ge \frac{4\srank(\A)}{\beta \tau^2}\log\frac{2n}{\delta}$, with probability at least $1-\delta$, we have
\begin{equation*}
\|\A^\top\A-\A_0^\top\A_0\|\le \tau \|\A\|^2.
\end{equation*}
\end{lemma}

\begin{remark}
Lemma \ref{lemma: Ismail} implies that
\begin{equation*}
(1-\tau)\|\A\|^2\le \|\A_0\|^2\le (1+\tau)\|\A\|^2,
\end{equation*}
because
\begin{equation*}
\left|\|\A\|^2-\|\A_0\|^2\right|=\left|\|\A^\top\A\|-\|\A_0^\top\A_0\|\right|\le \|\A^\top\A-\A_0^\top\A_0\|\le \tau\|\A\|^2.
\end{equation*}
\end{remark}

\begin{algorithm}
\caption{The sampling algorithm to estimate $\|\A\|$ up to $(1\pm\tau)$ relative error}
\label{algorithm: the sampling algorithm for even p}
\begin{algorithmic}[1]
\LineComment Lines 1-5 estimates the row norms of $\A$ and then sample rows non-uniformly.
\State Sample each row of $\A$ by Bernoulli distribution with probability $\cO(\frac{1}{n\tau})$. Denote by $\cS_{\mathsf{row}}$ the sampled set and $q=|\cS_{\mathsf{row}}|$.
\For {$i\gets 1$\textbf{ to }$q$}
	\State Uniformly sample $\cO(\frac{1}{\tau})$ entries from $\A_{\cS_{\mathsf{row}}(i),:}$, forming vector $\x$. \label{alg: r estimator 1}
	\State $r_i\leftarrow \max\{\tau n\|\x\|_2^2, \tau n\}$. \label{alg: r estimator 2}
\EndFor
\State Sample $q_{\mathsf{row}}=\cO(\frac{d\log n}{\tau^2})$ indices in $\cS_{\mathsf{row}}$ independently with replacement according to the probability $p_i=\frac{r_i}{r}$, where $r=\sum_{j=1}^{q} r_j$. Denote by $\cI_{\mathsf{row}}$ the sampled row indices. \label{alg: row index sampling by Ismail}
\Statex
\LineComment Lines 6-10 estimates the column norms of $\A$ and then sample columns non-uniformly.
\State Sample each row with probability $\cO(\frac{1}{n\tau})$. Repeat the procedure $n$ times with replacement. Denote the sampled set by $\cS_{\mathsf{col}}$ and $q'=|\cS_{\mathsf{col}}|$.
\For{$i\gets 1$\textbf{ to }$q'$}
	\State Uniformly sample $\cO(\frac{1}{\tau})$ entries from $\A_{\cI_{\mathsf{row}},\cS_{\mathsf{col}(i)}}$, forming vector $\x$.
	\State $r_i'\leftarrow \max\{\tau q\|\x\|_2^2, \tau q\}$.
\EndFor
\State Sample $q_{\mathsf{col}}=\cO(\frac{d\log n}{\tau^2})$ indices in $\cS_{\mathsf{col}}$ independently with replacement according to the probability $p_i'=\frac{r_i'}{r'}$, where $r'=\sum_{j=1}^{q'} r_j'$. Denote by $\cI_{\mathsf{col}}$ the sampled row indices.
\Statex
\State $\widetilde\A\leftarrow \A_{\cI_{\mathsf{row}},\cI_{\mathsf{col}}}$. Rescale the rows of $\widetilde\A$ by $\left\{\sqrt{\frac{q}{p_i q_{\mathsf{row}}}}\right\}$ and the columns of $\widetilde\A$ by $\left\{\sqrt{\frac{q'}{p_i'q_{\mathsf{col}}}}\right\}$.
\State \Return index sets $\cI_{\mathsf{row}}$, $\cI_{\mathsf{col}}$, scaling factors $\left\{\sqrt{\frac{q}{p_i q_{\mathsf{row}}}}\right\}$, $\left\{\sqrt{\frac{q'}{p_i'q_{\mathsf{col}}}}\right\}$, $\widetilde\A$, and $\|\widetilde\A\|$.
\end{algorithmic}
\end{algorithm}

\begin{theorem}
\label{theorem: operator norm estimator by Ismail}
Suppose that $\A$ is an $n\times n$ matrix satisfying that $\|\A\|_F^2=\Omega(\tau n^2)$, $\|\A\|_\infty \le 1$ and $\srank(\A) = \cO(d)$. Then with probability at least $0.9$, the output of Algorithm \ref{algorithm: the sampling algorithm for even p} satisfies $(1-\tau)\|\A\|\le \|\widetilde\A\|\le (1+\tau)\|\A\|$. The sample complexity is $\cO(d^2\log^2(n)/\tau^4)$.
\end{theorem}

\begin{proof}[Proof of Theorem \ref{theorem: operator norm estimator by Ismail}]
We note that for any row $\A_{i,:}$ such that $|\A_{i,j}|\le 1$ and $\eta\le \|\A_{i,:}\|_2^2\le n$, uniformly sampling $\Theta(\frac{n}{\eta})$ entries of $\A_{i,:}$ suffices to estimate $\|\A_{i,:}\|_2^2$ within a constant multiplicative factor. To see this, we use Chebyshev's inequality. Let $s=\Theta(\frac{n}{\eta})$ be the number of sampled entries, $Z_j$ be the square of the $j$-th sampled entry $\A_{i,l(j)}$ of vector $\A_{i,:}$, and $Z=\frac{n}{s}\sum_{j=1}^s Z_j$. So $Z$ is an unbiased estimator:
\begin{equation*}
\bbE[Z]=\frac{n}{s}s\mathbb{E}[Z_1]=n \sum_{j=1}^n \frac{1}{n} \A_{i,l(j)}^2=\|\A_{i,:}\|_2^2.
\end{equation*}
For the variance, we have
\begin{equation*}
\begin{split}
\mathsf{Var}[Z]&=\frac{n^2}{s^2}\sum_{j=1}^s \mathsf{Var}[Z_j]\le \frac{n^2}{s^2} \sum_{j=1}^s \bbE[Z_j^2]=\frac{n^2}{s} \bbE[Z_1^2]=\frac{n^2}{s} \sum_{j=1}^n\frac{1}{n}\A_{i,j}^4\\
&\le \frac{n}{s} \sum_{j=1}^n\A_{i,j}^2\qquad(\mbox{since } |\A_{i,j}|\le 1)\\
&=\Theta(\eta)\|\A_{i,:}\|_2^2\\
&\le \Theta(\|\A_{i,:}\|_2^4).\qquad(\mbox{since }\eta\le \|\A_{i,:}\|_2^2)
\end{split}
\end{equation*}
Therefore, by Chebyshev's inequality, we have
\begin{equation*}
\Pr\left[\left|Z-\|\A_{i,:}\|_2^2\right|\ge 10\|\A_{i,:}\|_2^2\right]\le \frac{1}{3}.
\end{equation*}

Note that in Step \ref{alg: row index sampling by Ismail} of Algorithm \ref{algorithm: the sampling algorithm for even p}, in total we sample $q_{\mathsf{row}}=\cO(\frac{d\log n}{\tau^2})$ row indices, obeying the conditions in Lemma \ref{lemma: Ismail} for a constant $\beta$. By concentration, with high probability $r=\cO(\frac{\|\A\|_F^2}{\tau n})$ in Step \ref{alg: row index sampling by Ismail}, because in expectation we sample $\cO(\frac{1}{\tau^2})$ entries to estimate $r$ and we scale $\|\x\|_2^2$ by a $\tau n$ factor in Steps \ref{alg: r estimator 1} and \ref{alg: r estimator 2}, and that $\|\A\|_F^2$ is as large as $\Omega(\tau n^2)$. The probability that any given row $i$ is sampled is equal to $\frac{1}{n\tau}\times\frac{r_i}{r} = \Omega(\frac{ r_i}{\|\A\|_F^2})$.
Suppose first that $\|\A_{i,:}\|_2^2 \le \tau n$. Then we have $r_i = \tau n$. Consequently, for such $i$, the probability of sampling row $i$ is at least $\Omega(\frac{\tau n}{\|\A\|_F^2}) \ge \Theta(\frac{\|\A_{i,:}\|_2^2}{\|\A\|_F^2})$, just as in Lemma \ref{lemma: Ismail}. Suppose next that $\|\A_{i,:}\|_2^2 \ge \tau n$. Then we have $r_i = \Theta(\|\A_{i,:}\|_2^2)$. Consequently, for such $i$, the probability of sampling row $i$ is at least $\Omega(\frac{\|\A_{i,:}\|_2^2}{\|\A\|_F^2})$, just as in Lemma \ref{lemma: Ismail}. Therefore, in the followings we can set $\beta$ in Lemma \ref{lemma: Ismail} as an absolute constant.

It follows from Lemma \ref{lemma: Ismail} that with probability at least $0.9$,
\begin{equation*}
(1-\tau)\|\A\|^2\le \|\A_{\mathsf{row}}\|^2\le (1+\tau)\|\A\|^2,
\end{equation*}
where $\A_{\mathsf{row}}$ is the \emph{scaled} row sampling of $\A$ as in Lemma \ref{lemma: Ismail}.
Conditioning on this event, by applying Lemma \ref{lemma: Ismail} again to the column sampling of $\A_{\mathsf{row}}$, we have with high probability,
\begin{equation}
\label{equ: small stable rank, operator norm}
(1-\tau)^2\|\A\|^2\le (1-\tau)\|\A_{\mathsf{row}}\|^2\le \|\widetilde\A\|^2\le (1+\tau)\|\A_{\mathsf{row}}\|^2\le (1+\tau)^2\|\A\|^2,
\end{equation}
where we have used the fact that $\srank(\A_{\mathsf{row}})= \cO(d)$. The statement $\srank(\A_{\mathsf{row}})= \cO(d)$ holds because $\bbE \|\A_{\mathsf{row}}\|_F^2=\|\A\|_F^2$ and
by the Markov bound, we have with constant probability that
\begin{equation*}
\|\A_{\mathsf{row}}\|_F^2\le c\|\A\|_F^2,
\end{equation*}
so
\begin{equation*}
\srank(\A_{\mathsf{row}})=\frac{\|\A_{\mathsf{row}}\|_F^2}{\|\A_{\mathsf{row}}\|^2}\le \frac{c\|\A\|_F^2}{(1-\tau)\|\A\|^2}\le C\srank(\A)\le C'd.\qedhere
\end{equation*}
\end{proof}

\subsubsection{Estimation with Eigengap}
\label{section: Estimation with Eigengap}

Let $\A\in\R^{n\times n}$. Suppose that $p=2q$. We define a cycle $\sigma$ to be an ordered pair of a sequence of length $q$: $\lambda=((i_1,...,i_q),(j_1,...,j_q))$ such that $i_r,j_r\in[k]$ for all $r$. Now we associate with $\lambda$ a scalar
\begin{equation}
\label{equ: cycle estimator}
\A_\lambda=\prod_{\ell=1}^q\A_{i_\ell,j_\ell}\A_{i_{\ell+1},j_\ell},
\end{equation}
where for convention we define that $i_{q+1}=i_1$. Denote by
\begin{equation}
\label{equ: averaging}
Z=\frac{1}{N}\sum_{i=1}^N\A_{\lambda_i}.
\end{equation}
Our goal is to estimate $\sigma_1(\A)$ up to $(1\pm\tau)$ relative error, which is an $(1\pm\tau)$ approximation to $\|\A\|$.

\begin{algorithm}
\caption{Estimate $\|\A\|$ up to $(1\pm\tau)$ relative error}
\label{algorithm: the sampling algorithm for even p with eigengap}
\begin{algorithmic}[1]
\INPUT Cycle length $q$, matrix size $n$.
\OUTPUT $(1\pm\tau)$-approximation estimator.
\For{$i=1$ \textbf{to} $N$}
\State Uniformly sample a cycle $\lambda_i$ of length $q$.
\State Compute $\A_{\lambda_i}$ by Eqn. \eqref{equ: cycle estimator}.
\EndFor
\State Compute $Z$ as defined in \eqref{equ: averaging}.
\State \Return $Z^{1/(2q)}n$. \label{alg: return}
\end{algorithmic}
\end{algorithm}

\begin{theorem}
\label{theorem: operator norm estimation of small matrix with eigengap}
Let $\tau\in(0,\frac{1}{2})$ be the accuracy parameter and suppose that the input matrix $\A \in \R^{n\times n}$ satisfies
\begin{itemize}[topsep=0pt,itemsep=0pt,partopsep=1ex,parsep=1ex]
\item $\|\A\|_\infty\leq 1$;
\item $\|\A\|_F\geq cn$ for some absolute constant $c > 0$;
\item $\sigma_2(\A)/\sigma_1(\A)\leq \tau^\gamma$ for some absolute constant $\gamma > 0$;
\item $\srank(\A) = \cO(1)$.
\end{itemize}
Let $N = \frac{C_1}{\tau^2}\exp(\frac{c_1}{\gamma})$ and $q=\frac{C_2}{\gamma}$ for some large constants $C_1, C_2 > 0$ and some small constant $c_1 > 0$. Then with probability at least $0.9$, the estimator returned by Algorithm \ref{algorithm: the sampling algorithm for even p with eigengap} satisfies $(1-\tau)\|\A\|\le Z^{1/(2q)}n\le (1+\tau)\|\A\|$. The sample complexity is $\Theta(Nq) = \Theta\left(\frac{1}{\gamma\tau^2}\exp(\frac{c_1}{\gamma})\right)$.
\end{theorem}

\begin{proof}[Proof of Theorem \ref{theorem: operator norm estimation of small matrix with eigengap}]
We show that the cycle estimator approximates $\|\A\|$ within a $(1\pm\tau)$ relative error.
Let $\lambda=(\{i_s\},\{j_s\})$ which is chosen uniformly with replacement. Recall that
\begin{equation*}
\A_{\lambda}=\prod_{\ell=1}^q \A_{i_\ell,j_\ell}\A_{i_{\ell+1},j_\ell}.
\end{equation*}
Hence
\[
\mathbb{E}\A_\lambda = \mathbb{E}\left[\prod_{\ell=1}^q \A_{i_\ell,j_\ell}\A_{i_{\ell+1},j_\ell}\right]
=\frac{1}{n^{2q}}\left[\sum_{i_1, i_2,..., i_q,j_1, j_2,..., j_q}\prod_{\ell=1}^q \A_{i_\ell,j_\ell}\A_{i_{\ell+1},j_\ell}\right].
\]
Note that (see, e.g.,~\cite{li2016approximating})
\[
\sum_{i_1, i_2,\dots, i_q,j_1, j_2,\dots, j_q}\prod_{\ell=1}^q \A_{i_\ell,j_\ell}\A_{i_{\ell+1},j_\ell} = \|\A\|_{2q}^{2q},
\]
and by the assumption on the singular values and the stable rank,
\[
\sigma_1(\A)^{2q} \leq \|\A\|_{2q}^{2q} \leq (1+\tau) \sigma_1(\A)^{2q},
\]
provided that $q\geq \frac{1}{2\gamma}(\frac{\log \srank(\A)}{\log(1/\tau)}+1)$, and thus it suffices to take $q = \Theta(\frac{1}{\gamma})$.

%
%Let $\u^{(i)}$ and $\v^{(i)}$ be the $i$-th principal component of $\A$ and $\G^{(i)}=\A-\sum_{k=1}^i\u^{(k)}\v^{(k)}$, we have
%\begin{equation*}
%\begin{split}
%&\sum_{i_1, i_2,..., i_q,j_1, j_2,..., j_q}\prod_{\ell=1}^q \A_{i_\ell,j_\ell}\A_{i_{\ell+1},j_\ell}\\
%&=\sum_{i_1,i_2,...,i_q,j_1,j_2,...,j_q}\prod_{\ell=1}^q [\u_{i_\ell}^{(1)}\v_{j_\ell}^{(1)}+\G_{i_\ell,j_\ell}^{(1)}][\u_{i_{\ell+1}}^{(1)}\v_{j_\ell}^{(1)}+\G_{i_{\ell+1},j_\ell}^{(1)}]\\
%&=\sum_{i_1,i_2,...,i_q,j_1,j_2,...,j_q}[\u_{i_1}^{(1)2}\u_{i_2}^{(1)2}...\u_{i_q}^{(1)2}\v_{j_1}^{(1)2}\v_{j_2}^{(1)2}...\v_{j_q}^{(1)2}]+\sum_{i_1, i_2,..., i_q,j_1, j_2,..., j_q}\prod_{\ell=1}^q \G_{i_\ell,j_\ell}^{(1)}\G_{i_{\ell+1},j_\ell}^{(1)}\\
%&=\left\|\u^{(1)}\right\|_2^{2q}\left\|\v^{(1)}\right\|_2^{2q}+\sum_{i_1, i_2,..., i_q,j_1, j_2,..., j_q}\prod_{\ell=1}^q \G_{i_\ell,j_\ell}^{(1)}\G_{i_{\ell+1},j_\ell}^{(1)}\\
%&=\left\|\u^{(1)}\right\|_2^{2q}\left\|\v^{(1)}\right\|_2^{2q}+\left\|\u^{(2)}\right\|_2^{2q}\left\|\v^{(2)}\right\|_2^{2q}+...+\left\|\u^{(n)}\right\|_2^{2q}\left\|\v^{(n)}\right\|_2^{2q}\\
%&=\sigma_1(\A)^{2q}\left[1+\left(\frac{\sigma_2(\A)}{\sigma_1(\A)}\right)^{2q}+...+\left(\frac{\sigma_n(\A)}{\sigma_1(\A)}\right)^{2q}\right]\\
%&=\sigma_1(\A)^{2q}(1\pm \Theta(\epsilon^2)),
%\end{split}
%\end{equation*}
%where the second equality holds because $\mathsf{Col}(\G^{(1)})\perp\u^{(1)}$ and $\mathsf{Row}(\G^{(1)})\perp\v^{(1)}$ imply that the cross terms vanish, and the last equality holds due to the assumption on the singular values of $\A$.
Therefore, noting that $\mathbb{E}[Z] = \mathbb{E}[\A_\lambda]$,
\begin{align}
\mathbb{E}[Z]&\leq \frac{1+\tau}{n^{2q}}\sigma_1(\A)^{2q}\leq 1+\tau, \label{equ: expectation of cycle}\\
\mathbb{E}[Z]&\geq \frac{1}{n^{2q}}\sigma_1(\A)^{2q} \geq \frac{1}{n^{2q}}\left(\frac{\|\A\|_F^2}{\srank(\A)}\right)^q\geq \left(\frac{c^2}{\srank(\A)}\right)^q = \exp\left(\frac{c_1}{\gamma}\right). \label{eqn:expectation_cycle_lb}
\end{align}
%that is, $\bbE[Z] = \Theta(1)$.

We now bound the variance of $\A_\lambda$. Observe that
\begin{equation*}
\begin{split}
\mathsf{Var}[\A_\lambda]&\le \mathbb{E}[\A_\lambda^2]\leq 1,
\end{split}
\end{equation*}
because $|\A_{i,j}|\le 1$ for all $i,j\in[n]$.
Thus by repeating the procedure $N = \frac{C_1}{\tau^2}\exp\left(\frac{2c_1}{\gamma}\right)$ times, we have
\begin{equation*}
\mathsf{Var}[Z]=\frac{1}{N}\mathsf{Var}[\A_\lambda] \leq \frac{1}{10}\tau^2\exp\left(\frac{2c_1}{\gamma}\right),
\end{equation*}
by choosing $C_1$ sufficiently large.
It follows from the Chebyshev inequality that
\begin{equation*}
\Pr\left[\left|\mathbb{E}[Z]-Z\right|> \tau \mathbb{E}[Z]\right]\le \frac{\mathsf{Var}[Z]}{\tau^2\mathbb{E}[Z]^2}\le \frac{1}{10},
\end{equation*}
where we have used the lower bound~\eqref{eqn:expectation_cycle_lb}.
This together with \eqref{equ: expectation of cycle} and \eqref{eqn:expectation_cycle_lb}
implies that
\begin{equation*}
\Pr\left[(1- \tau)\frac{1}{n^{2q}}\sigma_1(\A)^{2q}\le Z\le (1+\tau)^2\frac{1}{n^{2q}}\sigma_1(\A)^{2q}\right]>\frac{9}{10}.
\end{equation*}
So
\begin{equation*}
\Pr\left[\left(1-\tau\right)\sigma_1(\A)\le Z^{1/(2q)}n\le \left(1+\tau\right)\sigma_1(\A)\right]>\frac{9}{10}.
\end{equation*}
\end{proof}

\subsection{Sensing Algorithms}

\begin{theorem}
\label{theorem: operator norm estimator under sensing model}
Suppose that $\A$ is an $n\times n$ matrix such that $\|\A\|_F^2=\Omega(\tau n^2)$, $\|\A\|_\infty\le 1$ and $\srank(\A) = \cO(d)$. Then Algorithm \ref{algorithm: the sketching algorithm for even p} outputs a value $Z$, which satisfies $(1-\tau)\|\A\|\le Z\le (1+\tau)\|\A\|$ with probability at least $0.9$. The sketching complexity is $\cO(\max\{\log^2(d\log(n)/\tau),d^2\log(n)\}/\tau^2)$.
\end{theorem}

\begin{remark}
The optimality of Theorem \ref{theorem: operator norm estimator under sensing model} follows from the hard instance in the proof of Theorem \ref{theorem: stable rank lower bound}.
\end{remark}

\begin{algorithm}
\caption{The sketching/sensing algorithm to estimate $\|\A\|$ up to $(1\pm\tau)$ relative error}
\label{algorithm: the sketching algorithm for even p}
\begin{algorithmic}[1]
\State Obtain indices $\cI_{\mathsf{row}}$, $\cI_{\mathsf{col}}$ and scaling factors $\left\{\sqrt{\frac{q}{p_i q_{\mathsf{row}}}}\right\}$, $\left\{\sqrt{\frac{q'}{p_i'q_{\mathsf{col}}}}\right\}$ by Algorithm \ref{algorithm: the sampling algorithm for even p} with $|\cI_{\mathsf{row}}|=|\cI_{\mathsf{col}}|=\cO(d\log (n)/\tau^2)$.
\State Let $\G$ and $\H$ be $\Theta(\frac{\max\{\log(d\log(n)/\tau),d\}}{\tau})\times \cO(\frac{d\log n}{\tau^2})$ matrices with i.i.d. $\cN(0,1)$ entries. Scale the columns of $\G$ by $\left\{\sqrt{\frac{q}{p_i q_{\mathsf{row}}}}\right\}$ and the columns of $\H$ by $\left\{\sqrt{\frac{q'}{p_i'q_{\mathsf{col}}}}\right\}$.
\State Maintain $\G\A_{\cI_{\mathsf{row}},\cI_{\mathsf{col}}}\H^\top$.
\State Compute $Y$ defined in Eqn. \eqref{equ: cycle estimator main}.
\State \Return $Y^{\tau/(2\log (d\log(n)/\tau^2))}$.
\end{algorithmic}
\end{algorithm}

Before proving Theorem \ref{theorem: operator norm estimator under sensing model}, we introduce a new estimator of operator norm under the sensing model, which approximates the operator norm by the Schatten-$p$ norm of large $p$.

Specifically, let $\A$ be an $n\times n$ matrix. We define a cycle $\sigma$ to be an ordered pair of a sequence of length $q$ with $p=2q$: $\lambda=((i_1,\dots,i_q),(j_1,\dots,j_q))$ such that $i_r,j_r\in[k]$ for all $r$, $i_r\neq i_s$ and $j_r\neq j_s$ for $r\neq s$. Now we associate with $\lambda$ a scalar
\begin{equation}
\label{equ: cycle estimator by reading entries}
\A_\lambda=\prod_{\ell=1}^q\A_{i_\ell,j_\ell}\A_{i_{\ell+1},j_\ell},
\end{equation}
where for convention we define that $i_{q+1}=i_1$. Denote by $\cC$ the set of cycles. We define
\begin{equation}
\label{equ: cycle estimator main}
Y=\frac{1}{|\cC|}\sum_{\lambda\in\cC}(\G\A\H^\top)_\lambda
\end{equation}
for even $p$, where $\G\sim\cG(k,n)$, $\H\sim\cG(k,n)$, and $k\ge q$. This estimator, akin to that in~\cite{li2014sketching}, approximates the Schatten-$p$ and thus the operator norm, as we shall show below.

\begin{lemma}
\label{lemma: operator norm estimator under sensing model}
Suppose that $\A$ is a $n\times n$ matrix of stable rank at most $d$. Let $k=\Theta(\max\{\sqrt{n d},\log n\})$ and $Y$ be the estimator defined in \eqref{equ: cycle estimator main}. With probability at least $0.9$, it holds that $(1-\tau)\|\A\|\le Y^{\tau/(2\log (n))}\le (1+\tau)\|\A\|$. The sketching complexity is $\cO(k^2) = \cO(\max\{nd,\log^2n\})$.
\end{lemma}

\begin{proof}[Proof of Lemma \ref{lemma: operator norm estimator under sensing model}]
We first show that $\|\A\|_{\cS_p}$ and $\|\A\|$ differ at most a $(1\pm\tau)$ factor for $p=2\lceil\log (n)/\tau\rceil$. To see this,
\begin{equation*}
1\le \frac{\|\A\|_{\cS_p}^p}{\|\A\|^p}=\frac{\sigma_1^p(\A)+\sigma_2^p(\A)+\dots+\sigma_n^p(\A)}{\sigma_1^p(\A)}\le n,
\end{equation*}
and therefore
\begin{equation*}
1\le \frac{\|\A\|_{\cS_p}}{\|\A\|}\le n^{1/p}\le 1+\frac{1}{2}\tau.
\end{equation*}

We now show that the cycle estimator $Y^{1/p}$ approximates $\|\A\|_{\cS_p}$ within a $(1\pm\frac{1}{2}\tau)$ relative error.
We say that two cycles $\lambda=(\{i\},\{j\})$ and $\tau=(\{i'\},\{j'\})$ are $(a_1,a_2)$-disjoint if $|i\Delta i'|=2a_1$ and $|j\Delta j'|=2a_2$, denoted by $|\lambda\Delta\tau|=(a_1,a_2)$. Here $\Delta$ is the symmetric difference.
Denote by $\A=\U\mathbf{\Sigma}\V^\top$ the skinny SVD of $\A$. Let $\G$ and $\H$ be random matrices with i.i.d. $\cN(0,1)$ entries. Note that $\G\A\H^\top$ is identically distributed as $\G\mathbf{\Sigma}\H^\top$ by rotational invariance. Let $\widetilde\A$ be the $k\times k$ matrix $\G\mathbf{\Sigma}\H^\top$, where $k\ge q$. It is clear that
\begin{equation*}
\widetilde\A_{s,t}=\sum_{i=1}^n \sigma_i \G_{s,i}\H_{t,i}.
\end{equation*}
Define
\begin{equation*}
Y=\frac{1}{|\cC|}\sum_{\lambda\in\cC}\widetilde\A_{\lambda}.
\end{equation*}
Let $\lambda=(\{i_s\},\{j_s\})$. Then
\begin{equation*}
\widetilde\A_{\lambda}=\sum_{\substack{\ell_1\in[n],\dots,\ell_q\in[n]\\m_1\in[n],\dots,m_q\in[n]}} \prod_{s=1}^q \sigma_{\ell_s}\sigma_{m_s} \G_{i_s,\ell_s}\H_{j_s,\ell_s}\G_{i_{s+1},m_s}\H_{j_s,m_s}.
\end{equation*}
We note that
\begin{equation*}
\mathbb{E}Y=\bbE\widetilde\A_\lambda=\sum_{i=1}^n \sigma_i^{2q}=\|\A\|_{\cS_p}^p.
\end{equation*}
We now bound the variance of $Y$. Let $\tau=(\{i_s'\},\{j_s'\})$. Observe that
\begin{equation*}
\mathbb{E}Y^2=\frac{1}{|\cC|^2}\sum_{a_1=0}^q\sum_{a_2=0}^q\sum_{\substack{\lambda,\tau\in\cC\\|\lambda\Delta\tau|=(a_1,a_2)}}\bbE(\widetilde\A_{\lambda}\widetilde\A_{\tau}),
\end{equation*}
where
\begin{equation}
\label{equ: expectation of pair of cycles}
\begin{split}
\bbE(\widetilde\A_{\lambda}\widetilde\A_{\tau})=\sum_{\substack{\ell_1\in[n],\dots,\ell_q\in[n]\\\ell_1'\in[n],\dots,\ell_q'\in[n]\\m_1\in[n],\dots,m_q\in[n]\\m_1'\in[n],\dots,m_q'\in[n]}} &\left(\prod_{i=1}^q\sigma_{\ell_i}\sigma_{m_i}\sigma_{\ell_i'}\sigma_{m_i'}\right)\bbE\left(\prod_{s=1}^q \G_{i_s,\ell_s}\G_{i_{s+1},m_s}\G_{i_s',\ell_s'}\G_{i_{s+1}',m_s'}\right)\\
&\times\mathbb{E}\left(\prod_{s=1}^q \H_{j_s,\ell_s}\H_{j_{s},m_s}\H_{j_s',\ell_s'}\H_{j_{s}',m_s'}\right).
\end{split}
\end{equation}

For any fixed cycles $\lambda=(\{i_s\},\{j_s\})$ and $\tau=(\{i_s'\},\{j_s'\})$ such that $|\lambda\Delta \tau|=(a_1,a_2)$, we notice that
\begin{equation}
\label{equ: expectation}
\bbE(\widetilde\A_\lambda\widetilde\A_\tau)\le (2cnd)^{p}\|\A\|_{\cS_p}^{2p},
\end{equation}
for an absolute constant $c$.
To see this, we observe that for the expectation $\bbE(\widetilde\A_\lambda\widetilde\A_\tau)$ to be non-zero, we must have that each appeared $\G$ and $\H$ in Eqn. \eqref{equ: expectation of pair of cycles} repeats an even number of times. Though there are totally $n^{4q}$ many of configurations for $\{\ell_s\}$, $\{\ell_s'\}$, $\{m_s\}$ and $\{m_s'\}$, there are at most $n^{2q}3^q$ non-zero terms among the summation in Eqn. \eqref{equ: expectation of pair of cycles}. This is because each $\G$ and $\H$ must have power $2$ or $4$ by the construction of the cycle. We know that for each fixed configuration of blocks there are at most $n^{2q}$ free variables, and there are at most $16^q$ different kinds of configurations of blocks because the size of each block is at most $4$. So the number of non-zero terms is at most $(4n)^{2q}$. This is true no matter whether there exists some $i_r,i_s'$ or $j_r,j_s'$ such that
$i_r=i_s'$ or $j_r=j_s'$. We also claim that for each non-zero term in the summation of Eqn. \eqref{equ: expectation of pair of cycles},
\begin{equation*}
\bbE\left(\prod_{s=1}^q \G_{i_s,\ell_s}\G_{i_{s+1},m_s}\G_{i_s',\ell_s'}\G_{i_{s+1}',m_s'}\right)\cdot \bbE\left(\prod_{s=1}^q \H_{j_s,\ell_s}\H_{j_{s},m_s}\H_{j_s',\ell_s'}\H_{j_{s}',m_s'}\right)\le 25^q.
\end{equation*}
This is because $\bbE\G^2=\bbE\H^2=1$ and $\bbE\G^4=\bbE\H^4=3$. Therefore, for a certain configuration in which $p_1,\dots,p_w$ are free variables with multiplicity $r_1,\dots,r_w\ge 2$, the summation in Eqn. \eqref{equ: expectation of pair of cycles} is bounded by
\begin{equation*}
4n^{2q} 100^q\sum_{p_1,\dots,p_w}\sigma_{p_1}^{r_1}\cdots\sigma_{p_w}^{r_w}\le (2n)^{p}\|\A\|_{\cS_{r_1}}^{r_1}\cdots\|\A\|_{\cS_{r_w}}^{r_w}\le (2nd)^{p}\|\A\|_{\cS_{p}}^{2p},
\end{equation*}
where the last inequality follows from the facts that $\sum_{i=1}^w r_i=2p$ and, by the assumption $\srank(A)\leq d$, that $\|\A\|_{\cS_r}\le \|\A\|_{F} \le \sqrt{d}\|\A\|_{\cS_p}$ for any $r\geq 2$. Thus we obtain Eqn. \eqref{equ: expectation}.

We now bound $\mathbb{E}Y^2$. Note that $|\cC|=\Theta(k^p)$ and there are
\begin{equation*}
\binom{k}{q} \binom{q}{q-a_1} \binom{k-(q-a_1)}{a_1} \binom{k}{q} \binom{q}{q-a_2} \binom{k-(q-a_2)}{a_2}
\end{equation*}
pairs of $(a_1,a_2)$-disjoint cycles, which can be upper bounded by $\cO(10^q)$.
Hence
\begin{equation*}
\mathbb{E}Y^2=\frac{1}{|\cC|^2}\sum_{a_1=0}^q\sum_{a_2=0}^q\sum_{\substack{\lambda,\tau\in\cC\\|\lambda\Delta\tau|=(a_1,a_2)}}\bbE(\widetilde\A_{\lambda}\widetilde\A_{\tau})\leq C' \frac{1}{k^{2p}}q^210^q (2nd)^p\|\A\|_{\cS_p}^{2p}\leq \|\A\|_{\cS_p}^{2p},
\end{equation*}
by the assumption that $k=\Omega(\sqrt{ nd})$.

It follows that
\begin{equation*}
\mathsf{Var}[Y]\le \bbE Y^2\leq \|\A\|_{\cS_p}^{2p}.
\end{equation*}
Then by the Chebyshev inequality,
\begin{equation*}
\Pr\left[\left|\|\A\|_{\cS_p}^p-Y\right|> \frac{1}{2}\|\A\|_{\cS_p}^p\right]\le \frac{\mathsf{Var}[Y]}{4\|\A\|_{\cS_p}^{2p}}\le \frac{1}{10},
\end{equation*}
namely,
\[
\Pr\left[\left(1-\frac{1}{2}\tau\right)\|\A\|_{\cS_p}\le Y^{1/p}\le \left(1+\frac{1}{2}\tau\right)\|\A\|_{\cS_p}\right]>\frac{9}{10}.
\]
This together with the fact that $\|\A\|\le \|\A\|_{\cS_p}\le (1+\frac{1}{2}\tau)\|\A\|$ implies that
\[
\Pr\left[(1-\tau)\|\A\|\le Y^{1/p}\le (1+\tau)\|\A\|\right]>\frac{9}{10},
\]
as desired. This completes the proof of Lemma~\ref{lemma: operator norm estimator under sensing model}.
\end{proof}

We are now ready to prove Theorem~\ref{theorem: operator norm estimator under sensing model}. Recall that we have shown that by focusing on an $\cO(\frac{d\log n}{\tau^2})\times \cO(\frac{d\log n}{\tau^2})$ submatrix (without sampling it), we can achieve guarantee \eqref{equ: small stable rank, operator norm} when $\|\A\|_F^2=\Omega(\tau n^2)$ and $\|\A\|_\infty\le 1$. Letting $d\gets c_1 d$ and $n\gets \cO(\frac{d\log n}{\tau^2})$ in Lemma \ref{lemma: operator norm estimator under sensing model} concludes the proof of Theorem~\ref{theorem: operator norm estimator under sensing model}.

\end{document}